\documentclass[11pt,twoside]{article}

\usepackage{fancyhdr}
\usepackage{epsfig}
\usepackage{latexsym}
\usepackage{microtype}
\usepackage{lastpage}
\usepackage{hyperref}
\usepackage{graphicx}
\usepackage{lipsum}
\usepackage[latin9]{inputenc}
\usepackage{amsmath}
\usepackage{amssymb}
\usepackage{amsfonts}

\usepackage[dvipsnames*,svgnames]{xcolor}
\hypersetup{colorlinks=true,allcolors=DarkBlue,linkcolor=black}

\newtheorem{definition}{Definition}
\newtheorem{theorem}{Theorem}
\newtheorem{proposition}{Proposition}
\newtheorem{corollary}{Corollary}
\newtheorem{lemma}{Lemma}
\newtheorem{remark}{Remark}

\newtheorem{example}{Example}

\newenvironment{proof}[1][Proof]{\noindent\textbf{#1.} }{\ \rule{0.5em}{0.5em}}

\newcommand{\TheAuthor}{}
\newcommand{\Author}[1]{\renewcommand{\TheAuthor}{#1}}

\newcommand{\TheTitle}{}
\newcommand{\Title}[1]{\renewcommand{\TheTitle}{#1}}

\Author{Eleni Mandrali}
\Title{A translation of weighted LTL formulas to weighted Büchi automata over
$\omega$-valuation monoids}
\newcommand\blfootnote[1]{%
	\begingroup
	\renewcommand\thefootnote{}\footnote{#1}%
	\addtocounter{footnote}{-1}%
	\endgroup
}

\begin{document}

	\blfootnote{This postdoctoral research is cofinanced by
Greece and European Union (European Social Fund) through the Operational
program "Development of human resources, education and lifelong learning " in
the context of the Project "REINFORCEMENT OF POSTDOCTORAL RESEARCHERS- Second
Cycle " (MIS\ 5033021) that is implemented by State Scholarship Foundation
(IKY). }   	

\parindent=8mm
\noindent

\underline{{\noindent \bf \scriptsize  \thepage--\hyperlink{lastpage}{\pageref{LastPage}}}}
\vskip -3mm
\noindent
\vskip -1mm
\noindent

\vspace{1cm}
\begin{center}
{\Large\bf {A translation of weighted LTL formulas to weighted Büchi automata over
$\omega$-valuation monoids }
}
\end{center}
\vspace{4mm}

\begin{center}
{\large Eleni MANDRALI}\footnote{Department of Mathematics, Aristotle University of Thessaloniki, 54124 Thessaloniki,
\c Greece, Email: {\tt elemandr@math.auth.gr}}
\end{center}
\vspace{3ex}

\date{}

\begin{abstract}

In this paper we introduce a weighted LTL over product $\omega$-valuation
monoids that satisfy specific properties. We also introduce weighted
generalized Büchi automata with $\varepsilon$-transitions, as well as weighted
Büchi automata with $\varepsilon$-transitions over product $\omega$-valuation
monoids and prove that these two models are expressively equivalent and also
equivalent to weighted Büchi automata already introduced in the literature. We
prove that every formula of a syntactic fragment of our logic can be
effectively translated to a weighted generalized Büchi automaton with
$\varepsilon$-transitions. For generalized product $\omega$-valuation monoids that satisfy specific properties we define a weighted LTL, weighted generalized Büchi automata with $\varepsilon$-transitions, and weighted Büchi automata with $\varepsilon$-transitions, and we prove the aforementioned results for generalized product $\omega$-valuation monoids as well. The translation of weighted LTL formulas to weighted generalized Büchi automata with $\varepsilon$-transitions is now obtained for a restricted syntactical fragment of the logic.

\smallskip


\end{abstract}

\section{Introduction}

Weighted automata over finite and infinite words, defined in \cite{Sc-On} and
\cite{Es-Ku}, \cite{Es-On} respectively, are essential models in theoretical
computer science suitable to describe quantitative features of systems'
behavior. They can be seen as classical automata whose transitions are
equipped with some value, usually taken from a semiring. Weighted automata
have already been successfully used in applications in digital image
compression and natural language processing (cf. Chapters 11 and 14
respectively in \cite{Dr-Ku}), and there is a constantly increasing interest
for possible use of these models in other fields also, e.g., in medicine,
biology (cf. \cite{Yi-Ma},\cite{Yi-Li}). Chatterjee, Doyen, and Henzinger in
\cite{Ch-Do} defined automata with weights over the real numbers. The behavior
of these automata is not computed with the use of the structure of the
semiring. More precisely, the weight of a run (finite or infinite) is computed
by using a function that assigns a real value to the (finite or infinite) run
of the automaton. Examples of such functions are $Max$ and $Sum$ for finite
runs, and $Sup,$ $Limsup,$ $Liminf,$ limit average, and discounted sum for
infinite runs. The real value that eventually the automaton assigns to a word
is computed as the maximum (resp. supremum for infinite words) of the values
of all possible runs of the automaton on the word. In that work, Chatterjee,
Doyen and Henzinger presented answers to decidability problems and studied
their computational complexity, and further compared the expressive power of
their model for different functions. Similar questions were answered in
\cite{Ch-He}, \cite{Ch-Ex}, \cite{Ch-Pr} where other kinds of automata that
use functions for the computation of the weight of a run were presented. With
the functions mentioned above we can model a wide spectrum of procedures of
the behavior of several systems. The peak of power consumption can be modeled
as the maximum of a sequence of real numbers that represent power consumption,
while average response time can be modeled as the limit average
(\cite{Ch-Ch}, \cite{Ch-De}). For a detailed reference on the importance of
valuation functions we refer to \cite{Ch-Do}. Droste and Gastin introduced a weighted MSO logic in \cite{Dr-We}, and Droste and Meinecke extended this logic in
\cite{Dr-Me} to a weighted MSO logic capable of describing properties of
the automata of \cite{Ch-Do}, and introduced the structures of valuation
monoids and $\omega$-valuation monoids as a formalism capable of describing in
a generic way their behavior for different functions. The authors further
defined the structures of product valuation monoids and product $\omega
$-valuation monoids by equipping valuation and $\omega$-valuation monoids with
a multiplicative operation that is not necessarily associative or commutative.
Under the consideration of specific properties of the aforementioned
structures the authors proved for finite (resp. infinite) words the expressive
equivalence of syntactical fragments of their logic with weighted automata
(resp. weighted Muller automata) whose behavior is computed with the use of
valuation functions (resp. $\omega$-valuation functions). In \cite{Me-Do}, the
structure of valuation monoids was equipped with a family of product
operations, as well as with a Cauchy product and iteration of series, and the
expressive equivalence of weighted automata over valuation monoids and
weighted rational expressions was proved. In the same work, similar results
were obtained for the case of infinite words.

In the field of quantitative description of systems, the interest is also
focused in the development of tools able to perform quantitative analysis and
verification of systems \cite{Fa-Le}, \cite{No-Pa}. A possible road to follow
is the definition of quantitative specification languages and the
investigation of their relation with weighted automata. Such a study would set
the foundations for a successful generalization of the automata
theoretic-approach in model-checking (cf. \cite{Va-Wo},\cite{Va-Re}) in the
quantitative setup.

An automata theoretic approach for reasoning about multivalued objects was
proposed in \cite{Ku-Lu}. More, precisely, the authors defined a weighted
\textit{LTL }and weighted automata over De Morgan Algebras and presented a
translation of the formulas of the logic to weighted automata. In \cite{Ma-We}
the author defined a weighted\textit{ LTL} with weights and discounting
parameters over the max-plus semiring and introduced the model of weighted
generalized Büchi automata with $\varepsilon$-transitions and discounting. In
that work, formulas of a syntactic fragment of the proposed logic were
effectively translated to weighted generalized Büchi automata with discounting
and $\varepsilon$-transitions and this model was proved expressively
equivalent to weighted Büchi automata with discounting introduced in
\cite{Dr-Sk}. In \cite{Ma-Co} (Chapter 4) it was shown that the aforementioned
translation is also possible for formulas of a larger fragment of that logic.

It is the aim of this work to introduce a weighted \textit{LTL} over product
$\omega$-valuation monoids capable of describing how the quantitative behavior
of systems changes over time and present a translation of formulas of a
fragment of the logic to weighted generalized Büchi automata with
$\varepsilon$-transitions, and provide in this way a theoretical basis for the
definition of algorithms that can be used for the verification of quantitative
properties of systems. As mentioned before the structure of product $\omega
$-valuation monoids refers to a wide range of applications.

More precisely, we introduce a weighted \textit{LTL} with weights over
product $\omega$-valuation monoids (resp. generalized product $\omega$-valuation monoids), and prove the results of \cite{Ma-Co}
(Chapter 4) for a restricted syntactical fragment of the proposed logic. In
more detail, the content of this paper can be described as follows. After
presenting some preliminary notions in Section 3, in Section 4 we present the
structures of product $\omega$-valuation monoids and generalized product $\omega$-valuation monoids and we study their properties. In
Section 5, for product $\omega$-valuation monoids (resp. generalized product $\omega$-valuation monoids) that satisfy specific
properties we define the models of weighted generalized Büchi automata with
$\varepsilon$-transitions and weighted Büchi automata with $\varepsilon
$-transitions. We prove that these two models are equivalent and also
equivalent to weighted Büchi automata over product $\omega$-valuation momoids (resp. over generalized product $\omega$-valuation monoids).
In Section 6, we introduce the weighted \textit{LTL} over product $\omega
$-valuation monoids that satisfy specific properties and prove that the
formulas of a syntactic fragment of the proposed \textit{LTL} can be
effectively translated to weighed generalized Büchi automata with
$\varepsilon$-transitions following the constructive approach of \cite{Ma-We}.
In Section 7 we obtain the results of Section 6 for a restricted syntactical fragment of the weighted \textit{LTL} over generalized product $\omega$-valuation monoids.

\section{Related work}

We recall from the introduction that weighted versions of \textit{LTL }and
translations of formulas of the proposed logics to weighted automata were
presented in \cite{Ku-Lu} and in \cite{Ma-We} (see also in \cite{Ma-Co}). Both
constructions, the one in \cite{Ku-Lu} and the one in \cite{Ma-We} aim to simulate the inductive computation of the semantics of
the formulas of the proposed logics, nevertheless different algebraic
properties of the underlying structures lead to different constructive
approaches. More precisely, in \cite{Ku-Lu} the authors treat the formulas as
classical ones where elements of De Morgan Algebras are considered as atomic
propositions and obtain by \cite{Va-Re} the corresponding Büchi automaton.
Then, the automaton is transformed into a weighted one where the weights of
the transitions are indicated by the sets of atomic propositions with which
the unweighted automaton moves between two states. In \cite{Ma-We} the approach of \cite{Va-Re} is also followed in the sense that
the states of the automaton are sets of formulas satisfying discrete
conditions of consistency, and the final subsets are defined with respect to
the until operators. Nevertheless, the effective simulation of the computation
of the semantics of the given formula requires the existence of a maximal
formula (according to subformula relation) in each state that will indicate
the induction (and thus the operations) connecting the formulas $k\in K$ of
the state. In addition, as in \cite{De-Ga}, $\varepsilon$-transitions are used
to reduce formulas. However, in \cite{De-Ga} the goal of the reduction is the
production of sets of formulas whose elements are atomic propositions, or
their negations, or formulas with outermost connective the next operator. In
the case of \cite{Ma-We} $\varepsilon$-transitions are used to
reduce the maximal formula of a set, and to ensure that the state set of the
automaton is finite. In this work we follow the constructive approach of
\cite{Ma-We}, however the lack of algebraic properties, with
which every semiring is equipped with, imposes the need for a stronger
syntactical restriction on the formulas of our logic, in order to achieve the
desired translation of formulas to weighted generalized Büchi automata with
$\varepsilon$-transitions. Another quantitative version
of \textit{LTL} with values over [0,1] and discounting parameters is presented
in \cite{Al-Bo}, \ where the authors show that threshold model checking
problems can be decided by translating the weighted \textit{LTL} formulas of
that logic into Boolean nondeterministic Büchi automata.

Lately, classical results for \textit{LTL }have been generalized\ in the
weighted set-up. More precisely, in \cite{Dr-Vo} the authors proved for
(infinitary) series over arbitrary bounded lattices the coincidence of\textit{
LTL}-definability,\textit{ FO}-definability, star-freeness and aperiodicity.
In \cite{Ma - On} (cf. also Chapter 5 in \cite{Ma-Co}) the expressive
equivalence of (fragments of) \textit{LTL}-definable, \textit{FO}-definable,
star-free and counter-free series infinitary series over the max-plus semiring
with discounting was proved. This result was generalized in \cite{Ma-Ch} (cf.
also Chapter 5 in \cite{Ma-Co}) for infinitary series over totally commutative
complete, idempotent and zero-divisor free semirings.

\section{Preliminaries}
Let $C,K$ be sets. If $B$ is a subset of $C$ (resp. proper subset of $C$), we shall write
$B\subseteq C$ (resp. $B\subset C$). We shall denote by $\mathcal{P}\left(C\right)$ the powerset of $C$. An index set $I$ of $C$ is a subset of $C$ whose elements are used to label the elements of another set. A family of elements of $K$ over the index set $I$, denoted by $\left(k_i\right)_{i\in I}$, is a mapping $f$ from $I$ to $K$ where $k_i=f\left(i\right)$ for all $i\in I$. We shall denote by $\mathbb{N}$ the set of non-negative integers.

\textit{Words }Let $A$ be an alphabet, i.e., a finite non-empty set. As
usually, we denote by $A^{\ast}$ the set of all finite words over $A$ and
$A^{+}=A^{\ast}\left\backslash \left\{  \varepsilon\right\}  \right.  ,$ where
$\varepsilon$ is the empty word. The set of all infinite sequences with
elements in $A$, i.e., the set of all infinite words over $A,$ is denoted by
$A^{\omega}.$ Let $w\in A^{\omega}.$ A word $v\in A^{\omega}$ is called a
suffix of $w$, if $w=uv$ for some $u\in A^{\ast}.$ Every infinite word
$w=a_{0}a_{1}\ldots$ with $a_{i}\in A\left(  i\geq0\right)  $ is written also
as $w=w\left(  0\right)  w\left(  1\right)  \ldots$ where $w\left(  i\right)
=a_{i}\left(  i\geq0\right)  $. The word $w_{\geq i}$ denotes the suffix of
$w$ that starts at position $i,$ i.e., $w_{\geq i}=w\left(  i\right)  w\left(
i+1\right)  \ldots.$

\textit{Monoids }A monoid $\left(  K,+,\mathbf{0}\right)  $ is an algebraic
structure equipped with a non-empty set $K$ and an associative additive
operation $+$ with a zero element\textbf{\ }$\mathbf{0},$ i.e., $\mathbf{0}%
+k=k+\mathbf{0}=k$ for every $k\in K.$ The monoid $K$ is called commutative if $+$ is commutative.

\bigskip{}A monoid $\left(  K,+,\mathbf{0}\right)  $ is called complete if it
is equipped, for every index set $I$, with an infinitary sum operation
$\sum_{I}:K^{I}\rightarrow K$ such that for every family $\left(
k_{i}\right)  _{i\in I}$ of elements of $K$ we have
\[
\underset{i\in\emptyset}{\sum}k_{i}=\mathbf{0},\underset{i\in\left\{
j\right\}  }{\sum}k_{i}=k_{j},\underset{i\in\left\{  j,l\right\}  }{%
{\displaystyle\sum}
}k_{i}=k_{j}+k_{l}\text{ for }j\neq l
\]
and
\[
\underset{j\in J}{%
{\displaystyle\sum}
}\left(  \underset{i\in I_{j}}{\sum}k_{i}\right)  =\underset{i\in I}{%
{\displaystyle\sum}
}k_{i},
\]
if $%
{\displaystyle\bigcup_{j\in J}}
I_{j}=I$ and $I_{j}%
{\displaystyle\bigcap}
I_{j^{\prime}}=\emptyset$ for $j\neq j^{\prime}.$
We note that every complete monoid is commutative.

Let $K$ be a complete monoid. $K$ is called additively
idempotent (or simply idempotent), if $k+k=k$ for every $k\in K.$ Furthermore, $K$ is zero-sum free if $k+k^{\prime}=\mathbf{0}$ implies
$k=k^{\prime}=\mathbf{0}$. It is well known that if $K$ is idempotent, then $K$ is
necessarily zero-sum free (\cite{Ak - Ga}). We recall (cf. \cite{Gu - An}) that idempotency gives rise to a natural
partial order in $K$ defined in the following way. Let $k,k^{\prime}\in K,$
then $k\leq k^{\prime}$ iff $k^{\prime}=k^{\prime}+k.$ Equivalently, it holds
$k\leq k^{\prime}$ iff $k^{\prime}=k^{\prime\prime}+k$ for some $k^{\prime
\prime}\in K$ (cf. Chapter 5 in {\cite{Dr-Ku}}). We recall that a partial order of a set $K$ is a total order, if $k\leq k^{\prime}$, or $k^{\prime}\leq k $ for all $k, k^{\prime}\in K. $
Let now $K^{\prime},K^{\prime\prime}$ be two non-empty subsets of a complete monoid $K$. We define the sum of $K^{\prime}$ and
$K^{\prime\prime}$ in the following way
\[
K^{\prime}+K^{\prime\prime}=\left\{  k\in K\mid\exists k^{\prime}\in
K^{\prime},k^{\prime\prime}\in K^{\prime\prime}\text{ s.t. }k=k^{\prime
}+k^{\prime\prime}\right\}  .
\]

\textit{Series }Let $A$ be an alphabet and\textit{ }$K$ be a complete monoid.
An infinitary series over $A$ and $K$ is a mapping $s:A^{\omega}\rightarrow
K.$ For every $w\in A^{\omega}$ we write $\left(  s,w\right)  $ for the value
$s\left(  w\right)  $ and refer to it as the coefficient of $s$ on $w.$ We
denote by $K\left\langle \left\langle A^{\omega}\right\rangle \right\rangle $
the class of all infinitary series over $A^{\omega}$ and $K.$

\section{Product $\omega$-valuation monoids, generalized product $\omega$-valuation monoids, and their properties}

For a set $K$ we denote by $L\subseteq_{fin}K$ the fact that $L$ is a finite
subset of $K$ and we let $\left(  K_{fin}\right)  ^{\omega}=\underset
{L\subseteq_{fin}K}{%
{\displaystyle\bigcup}
}L^{\omega}.$ We now recall the definition of $\omega$-valuation monoids and
product $\omega$-valuation monoids from \cite{Dr-Me}, with the difference that
we equip these structures with two additional properties.

\begin{definition}
An $\omega$-valuation monoid $\left(  K,+,Val^{\omega},\mathbf{0}\right)  $ is
a complete monoid $\left(  K,+,\mathbf{0}\right)  $ equipped with an $\omega
$-valuation function $Val^{\omega}:\left(  K_{fin}\right)  ^{\omega
}\rightarrow K$ such that $Val^{\omega}\left(  k_{i}\right)  _{i\in%
\mathbb{N}
}=\mathbf{0}$ whenever $k_{i}=\mathbf{0}$ for some $i\geq0.$ A product
$\omega$-valuation monoid $\left(  K,+,\cdot,Val^{\omega},\mathbf{0}%
,\mathbf{1}\right)  $ is an $\omega$-valuation monoid $\left(  K,+,Val^{\omega
},\mathbf{0}\right)  $ further equipped with a product operation $\cdot
:K^{2}\rightarrow K,$ with $\mathbf{1}\in K,\mathbf{1}\neq\mathbf{0},$ such that $Val^{\omega}\left(
\mathbf{1}^{\omega}\right)  =\mathbf{1}$ and $\mathbf{0}\cdot k=k\cdot
\mathbf{0=0,}$ $\mathbf{1\cdot}k=k\cdot\mathbf{1}=k$ for all $k\in K$;
additionally, for every index set $I$ and $k\in K,$ $\underset{I}{\sum}\left(
k\cdot\mathbf{1}\right)  =k\cdot$ $\underset{I}{\sum}\mathbf{1},$ and for
every $L\subseteq_{fin}K,$ finite index sets $I_{j}(j\geq0),$ and all $k_{i_{j}}\in
L\left(  i_{j}\in I_{j}\right)  $
\begin{equation}
Val^{\omega}\left(  \underset{i_{j}\in I_{j}}{\sum}k_{i_{j}}\right)  _{j\in%
\mathbb{N}
}=\underset{\left(  i_{j}\right)  _{j}\in I_{0}\times I_{1}\times\ldots}{%
{\displaystyle\sum}
}Val^{\omega}\left(  k_{i_{j}}\right)  _{j\in%
\mathbb{N}
}. \label{Sum-valuation}%
\end{equation}

\end{definition}

The property described by equation 1 expresses the \textbf{distributivity} of $Val^{\omega}$ over finite sums. We recall that this property has also been considered in
\cite{Me-Do} for the definition of Cauchy $\omega$-indexed valuation monoids.

\begin{remark}
Observe that for every $k\in K,$ it holds $k_{1}\cdot\left(  k_{2}\cdot
k_{3}\right)  =\left(  k_{1}\cdot k_{2}\right)  \cdot k_{3}$ for every
$k_{1},k_{2},k_{3}\in\left\{  \mathbf{0},\mathbf{1},k\right\}  $ such that $k_{i}\in K\backslash\left\{ \mathbf{0},\mathbf{1}\right\} $ for at most one $i\in \left\{1,2,3\right\}$.\label{remark_commutative}
\end{remark}

We introduce now the notion of \textit{generalized} product $\omega$-valuation monoids. These are defined with the same way as product $\omega$-valuation monoids with the difference that these structures are equipped with a restricted version of the distributivity property of $Val^{\omega}$ over finite sums.

\begin{definition}
A generalized product
$\omega$-valuation monoid $\left(  K,+,\cdot,Val^{\omega},\mathbf{0}%
,\mathbf{1}\right)  $ is an $\omega$-valuation monoid $\left( K,+,Val^{\omega
},\mathbf{0}\right)  $ further equipped with a product operation $\cdot
:K^{2}\rightarrow K,$ with $\mathbf{1}\in K,\mathbf{1}\neq\mathbf{0},$ such that $Val^{\omega}\left(
\mathbf{1}^{\omega}\right)  =\mathbf{1}$ and $\mathbf{0}\cdot k=k\cdot
\mathbf{0=0,}$ $\mathbf{1\cdot}k=k\cdot\mathbf{1}=k$ for all $k\in K;$
additionally, for every index set $I$ and $k\in K,$ $\underset{I}{\sum}\left(
k\cdot\mathbf{1}\right)  =k\cdot$ $\underset{I}{\sum}\mathbf{1},$ and the following hold: Let $L\subseteq_{fin}K,$ and $I_{j}(j\geq0)$ a family of finite index sets. If
for
all but a finite number of $j\geq0,$ one of the following holds $k_{i_{j}}\in L\backslash
\left\{\mathbf{0},\mathbf{1}  \right\}  $ for all $i_{j}\in I_{j},$ or $k_{i_{j}}%
\in\left\{  \mathbf{0},\mathbf{1}\right\}  $ for all $i_{j}\in I_{j}$, we have
\[
Val^{\omega}\left(  \underset{i_{j}\in I_{j}}{\sum}k_{i_{j}}\right)  _{j\in%
\mathbb{N}
}=\underset{\left(  i_{j}\right)  _{j}\in I_{0}\times I_{1}\times\ldots}{%
{\displaystyle\sum}
}Val^{\omega}\left(  k_{i_{j}}\right)  _{j\in%
\mathbb{N}
}.%
\]
\end{definition}

\bigskip Observe that every product $\omega$-valuation monoid is a generalized product $\omega$-valuation monoid. However, not every generalized product $\omega$-valuation monoid is a product $\omega$-valuation monoid (see Example \ref{example2}). We will call the product $\omega$-valuation monoid (resp. the generalized product $\omega$-valuation monoid) $\left(
K,+,\cdot,Val^{\omega},\mathbf{0},\mathbf{1}\right)  $ idempotent if the
complete monoid $\left(  K,+,\mathbf{0}\right)  $ is idempotent.

Subsequently, we derive properties of product $\omega$-valuation monoids and generalized product $\omega$-valuation monoids. For simplicity we provide the proofs only for product $\omega$-valuation monoids. The reader may check that essentially the same arguments also hold, if $K$ is a generalized product $\omega$-valuation monoid.
\begin{lemma}
\label{Properties Sum}Let $K$ be an idempotent product $\omega$-valuation
monoid or an idempotent generalized product $\omega$-valuation monoid. Then,

(i) {[}\cite{Dr-Ku}, Chapter 5, Lemma7.3{]} $%
{\displaystyle\sum\limits_{I}}
\mathbf{1}=\mathbf{1}$ for every set $I$ with size at most continuum.

(ii) $%
{\displaystyle\sum\limits_{I}}
k=k$ for every set $I$ with size at most continuum and every $k\in K.$

(iii) $%
{\displaystyle\sum\limits_{k\in K}}
k+%
{\displaystyle\sum\limits_{k^{\prime}\in K^{\prime}}}
k^{\prime}=$ $%
{\displaystyle\sum\limits_{k\in K}}
k$ for every $K^{\prime}\subseteq K.$

(iv) $%
{\displaystyle\sum\limits_{k^{\prime\prime}\in K^{^{\prime\prime}}}}
k^{^{\prime\prime}}+%
{\displaystyle\sum\limits_{k^{\prime}\in K^{\prime}}}
k^{\prime}=$ $%
{\displaystyle\sum\limits_{k\in K^{\prime}+K^{\prime\prime}}}
k$ $\ \ $for every $K^{\prime},K^{\prime\prime}$ non-empty subsets of $K$ with
size at most continuum$.$
\end{lemma}

\begin{proof} (ii) It holds  $\underset{I}{\sum}\left(  k\cdot\mathbf{1}\right)  =k\cdot$
$\underset{I}{\sum}\mathbf{1},$ for every $k\in K,$ and index set $I.$
Hence, by the above property and (i) we get $\sum_{I}k=k\cdot\sum_{I}\mathbf{1}=k\cdot\mathbf{1}=k$.

(iii) For $K^{\prime}=\emptyset$ it is obvious. Otherwise, we get\\\\
$\underset{k\in K}{\sum}k+\underset{k^{\prime}\in K^{\prime}}{\sum}k^{\prime}=\underset{k\in K\setminus K^{\prime}}{\sum}k+\underset{k^{\prime}\in K^{\prime}}{\sum}k^{\prime}+\underset{k^{\prime}\in K^{\prime}}{\sum}k^{\prime}=\underset{k\in K\setminus K^{\prime}}{\sum}k$
$+\underset{k{}^{\prime}\in K^{\prime}}{\sum}k^{\prime}=\underset{k\in K}{\sum}k$,\\\\where the first and last equality hold by the completeness axioms of the monoid, and the second one by idempotency.

(iv) For each $k^{\prime}\in K^{\prime}$ (resp. $k^{\prime\prime}\in K^{\prime\prime}$)
there exists an index set $I_{k^{\prime}}\neq\emptyset$ (resp. $I_{k^{\prime\prime}}\neq\emptyset$)
with size at most continuum such that
\[
\begin{array}{cl}
\underset{k\in K^{\prime}+K^{\prime\prime}}{\sum}k & =\underset{k\in K^{\prime}+K^{\prime\prime}}{\sum}\left(\underset{\textrm{s.t }k=k^{\prime}+k^{\prime\prime}}{\underset{\left(k^{\prime},k^{\prime\prime}\right)\in K^{\prime}\times K^{\prime\prime}}{\sum}}\left(k^{\prime}+k^{\prime\prime}\right)\right)\\
 & =\underset{k^{\prime}\in K^{\prime}}{\sum}\left(\sum_{I_{k^{\prime}}}k^{\prime}\right)+\underset{k^{\prime\prime}\in K^{\prime\prime}}{\sum}\left(\sum_{I_{k^{\prime\prime}}}k^{\prime\prime}\right)\\
 & =\underset{k^{\prime}\in K^{\prime}}{\sum}k^{\prime}+\underset{k^{\prime\prime}\in K^{\prime\prime}}{\sum}k^{\prime\prime}
\end{array}
\]
where the first and last equality holds by (ii), and the second equality by the completeness axioms of the monoid.
\end{proof}

\begin{lemma}
\label{Lemma inequality sum}Let $K$ be an idempotent product $\omega
$-valuation monoid or an idempotent generalized product $\omega$-valuation monoid, and $K^{\prime},K^{\prime\prime}\subseteq K$ such that the size of $K^{\prime}$ is at
most continuum and for every $k^{\prime}\in K^{\prime}$ there exists
$k^{\prime\prime}\in K^{\prime\prime}$ with $k^{\prime}\leq k^{\prime\prime}.$
Then, $\underset{k^{\prime}\in K^{\prime}}{%
{\displaystyle\sum}
}k^{\prime}\leq\underset{k^{\prime\prime}\in K^{\prime\prime}}{%
{\displaystyle\sum}
}k^{\prime\prime}.$
\end{lemma}

\begin{proof}
There exist index sets $I,J$, with the size of $I$ being at most continuum, such that $K^{\prime
}=\left\{  k_{i}^{\prime}\in K\mid i\in I\right\}  $, $K^{\prime\prime
}=\left\{  k_{j}^{\prime\prime}\in K\mid j\in J\right\}  .$ We let
\[
\overline{J}=\left\{  j\in J\mid\exists i_{j}\in I,k_{i_{j}}^{\prime}\in
K^{\prime}\text{ such that }k_{i_{j}}^{\prime}\leq k_{j}^{\prime\prime
}\right\}  .
\]
For every $j\in\overline{J}$ we let $I_{j}=\left\{  i\in I\mid k_{i}^{\prime
}\leq k_{j}^{\prime\prime}\right\}  .$ It holds $\underset{j\in\overline{J}}{%
{\displaystyle\bigcup}
}I_{j}=I.$ We fix a $j\in\overline{J}.$ For every $i\in I_{j},$ we have $k_{j}^{\prime\prime}=k_{i}^{\prime}+k_{j}^{\prime\prime}$ and by
idempotency and Lemma \ref{Properties Sum}(ii) it holds $k_{j}^{\prime\prime
}=\underset{i\in I_{j}}{%
{\displaystyle\sum}
}\left(  k_{i}^{\prime}+k_{j}^{\prime\prime}\right)  $. We thus get
\begin{align*}
\underset{j\in\overline{J}}{%
{\displaystyle\sum}
}k_{j}^{\prime\prime}  &  =\underset{j\in\overline{J}}{%
{\displaystyle\sum}
}\left(  \underset{i\in I_{j}}{%
{\displaystyle\sum}
}\left(  k_{i}^{\prime}+k_{j}^{\prime\prime}\right)  \right) \\
&  =\underset{j\in\overline{J}}{%
{\displaystyle\sum}
}\left(  \underset{i\in I_{j}}{%
{\displaystyle\sum}
}k_{i}^{\prime}\right)  +\underset{j\in\overline{J}}{%
{\displaystyle\sum}
}\left(  \underset{i\in I_{j}}{%
{\displaystyle\sum}
}k_{j}^{\prime\prime}\right) \\
&  =\underset{i\in I}{%
{\displaystyle\sum}
}k_{i}^{\prime}+\underset{j\in\overline{J}}{%
{\displaystyle\sum}
}k_{j}^{\prime\prime}.
\end{align*}
We conclude the second equality by the completeness axioms of the monoid, and the last one by Lemma \ref{Properties Sum}(ii), and the fact that $I_j$ has size at most continuum for every $j\in\overline{J}$. Hence, $\underset{i\in I}{%
{\displaystyle\sum}
}k_{i}^{\prime}\leq\underset{j\in\overline{J}}{%
{\displaystyle\sum}
}k_{j}^{\prime\prime}\leq\underset{j\in J}{%
{\displaystyle\sum}
}k_{j}^{\prime\prime}$ where the first inequality is concluded taking into account the definition of the natural order of $K$, and the second inequality holds by Lemma
\ref{Properties Sum}(iii), and again by the definition of the natural order of $K$.
\end{proof}

\begin{lemma}
\label{Valuation inequality}(i) Let $\left(  K,+,\cdot,Val^{\omega}%
,\mathbf{0},\mathbf{1}\right)  $ be an idempotent product $\omega$-valuation
monoid, and $L\subseteq_{fin}K.$ If $\left(  k_{i}^{1}\right)  _{i\geq0}$ and
$\left(  k_{i}^{2}\right)  _{i\geq0}$ are families of elements of $L$ such
that for every $i\geq0,k_{i}^{1}\leq k_{i}^{2},$ then $Val^{\omega}\left(
k_{i}^{1}\right)  _{i\geq0}\leq Val^{\omega}\left(  k_{i}^{2}\right)
_{i\geq0}.$

(ii) Let $\left(  K,+,\cdot,Val^{\omega},\mathbf{0},\mathbf{1}\right)  $ be an
idempotent generalized product $\omega$-valuation monoid, and $L\subseteq
_{fin}K.$ If $\left(  k_{i}^{1}\right)  _{i\geq0}$ and $\left(  k_{i}%
^{2}\right)  _{i\geq0}$ are families of elements of $L$ such that for every
$i\geq0,k_{i}^{1}\leq k_{i}^{2},$ and for all but a finite number of $i\geq0,$
it holds $\left\{  k_{i}^{1},k_{i}^{2}\right\}  \subseteq L\backslash\left\{
\mathbf{0},\mathbf{1}\right\}  ,$ or $\left\{  k_{i}^{1},k_{i}^{2}\right\}
\subseteq\left\{  \mathbf{0},\mathbf{1}\right\}  ,$ then $Val^{\omega}\left(
k_{i}^{1}\right)  _{i\geq0}\leq Val^{\omega}\left(  k_{i}^{2}\right)
_{i\geq0}.$
\end{lemma}

\begin{proof}
(i) It holds
\begin{align*}
Val^{\omega}\left(  k_{i}^{2}\right)  _{i\geq0}  &  =Val^{\omega}\left(
k_{i}^{2}+k_{i}^{1}\right)  _{i\geq0}=\underset{\left(  j_{0},j_{1}%
,\ldots\right)  \in\left\{  1,2\right\}  ^{\omega}}{\sum}Val^{\omega}\left(
k_{i}^{j_{i}}\right)  _{i\geq0}\\
&  =Val^{\omega}\left(  k_{i}^{1}\right)  _{i\geq0}+\underset{\left(
j_{0},j_{1},\ldots\right)  \neq1^{\omega}}{\underset{\left(  j_{0}%
,j_{1},\ldots\right)  \in\left\{  1,2\right\}  ^{\omega}}{\sum}}Val^{\omega
}\left(  k_{i}^{j_{i}}\right)  _{i\geq0}%
\end{align*}
where the second equality holds by the distributivity of $Val^{\omega}$ over finite sums, and the third one by the completeness axioms of the monoid, and this
proves our claim.

(ii) We can prove the claim with the same arguments used in the previous case.
\end{proof}

In the rest of this paper we will consider idempotent product $\omega
$-valuation monoids $\left(  K,+,\cdot,Val^{\omega},\mathbf{0},\mathbf{1}%
\right)  $ (resp. idempotent generalized product $\omega$-valuation monoids) that further satisfy the following properties (resp. further satisfy the following properties and the natural order is a total order). For all $k,k_{i}\in
K,\left(i\geq 1\right)$
\begin{equation}
Val^{\omega}\left(  \mathbf{1,}k_{1},k_{2},k_{3},\ldots\right)  =Val^{\omega
}\left(  k_{i}\right)  _{i\geq1}, \label{Property 1}%
\end{equation}%
\begin{equation}
k=Val^{\omega}\left(  k,\mathbf{1},\mathbf{1},\mathbf{1},\ldots\right)
\label{Property 3}%
\end{equation}

In the rest of the paper we shall call the properties described by equations \ref{Property 1}, \ref{Property 3}, \textbf{Property 2}, and \textbf{Property 3} respectively. We note that Properties 2 and 3 express a notion of neutrality of $\mathbf{1} $ over $Val^{\omega}$. Next, we present examples of product $\omega$-valuation monoids, and generalized product $\omega$-valuation monoids.

\begin{example}
Every idempotent totally commutative complete semiring (cf. \cite{Ma-Ch}) $\left(  K,+,\cdot
,\mathbf{0},\mathbf{1}\right)  $ can be considered as an idempotent product
$\omega$-valuation monoid $\left(  K,+,\cdot,\prod
,\mathbf{0},\mathbf{1}\right)  $ if we consider as the $\omega$-valuation function the countably infinite products operation $\prod$ that every totally commutative complete semiring is equipped with. Moreover, these structures satisfy Properties 2, and 3. We can verify these properties, as well as the ones in the definition of product $\omega$-valuation monoids in a straightforward way by applying the completeness axioms of the structures.
\end{example}

\begin{example}
\label{example2}
We consider the structure $\left(  \overline{%
\mathbb{R}
},\sup,\inf,\text{liminf, }-\infty,\infty\right)  $ where $\overline{%
\mathbb{R}
}=%
\mathbb{R}
\cup\left\{  \infty,-\infty\right\}  $ and liminf is an $\omega$-valuation
function from $\left(  \overline{%
\mathbb{R}
}_{fin}\right)  ^{\omega}$ to $\overline{%
\mathbb{R}
}$ defined by
\[
\text{liminf}\left(  \left(  d_{i}\right)  _{i\geq0}\right)  =\left\{
\begin{array}
[c]{ll}%
-\infty &
\begin{array}
[c]{l}%
\text{if there exists }i\geq0\text{ with}\\
d_{i}=-\infty
\end{array}
\\
& \\
\infty & \text{ \ if for all }i\geq0\text{, }d_{i}=\infty\\
& \\
\underset{i\geq0}{\text{lim}}(\inf\left\{  d_{k}\mid k\geq i,d_{k}\neq
\infty\right\}  ) &
\begin{array}
[c]{l}%
\text{if }d_{j}\neq-\infty\text{ for all }j\geq0,\\
\text{and there exist infinitely }\\
\text{many }i\geq0\text{ with }d_{i}\neq\infty
\end{array}
\\
& \\
\inf\left\{  d_{i}\mid i\geq0\text{ with }d_{i}\neq\infty\right\}   & \text{
\ otherwise}%
\end{array}
\right.
\]
$\left(  \overline{%
\mathbb{R}
},\sup,\inf,\text{liminf,}-\infty,\infty\right)  $ is an idempotent generalized product
$\omega$-valuation monoid that satisfy Properties \ref{Property 1}, and
\ref{Property 3}, and the natural order obtained in the structure is a total order.
We observe that $\left( \overline{%
\mathbb{R}
},\sup ,\inf ,\text{\textit{liminf}, }-\infty ,\infty \right) $ is not a
product $\omega $-valuation monoid. To verify this observation we present
the following counterexample. We consider the families of
elements of $\overline{%
\mathbb{R}
},$ $\left( k_{i_{j}}\right) _{i_{j}\in I_{j}},$ where for every $j\geq 0$
with $j\neq 1,$ it holds $I_{j}=\left\{ 1,2\right\} $, and $k_{1}=\infty
,k_{2}=6,$ and for $j=1,$ we have $I_{1}=\left\{ 1\right\} ,$ and $k_{1}=5.$
Then, liminf$\left( \underset{i_{j}\in I_{j}}{\sup }k_{i_{j}}\right)
_{j\in
\mathbb{N}
}=5,$ and $\underset{\left( i_{j}\right) _{j}\in I_{0}\times I_{1}\times
\ldots }{\sup }\left( \text{liminf}\left( k_{i_{j}}\right) _{j\in
\mathbb{N}
}\right) =6.$
\end{example}

In the Appendix we prove that the structure presented in the previous example is indeed a generalized product $\omega$-valuation monoid. In $\cite{Dr-Me}$ the authors have also considered an $\omega$-product valuation monoid where the classical liminf-function is used. The definition of $\it{liminf}$ in our example is motivated by the need to capture the semantics of weighted logics that will be introduced of Section 6. In particular, the semantics of our $\it{until}$-operator expresses the fact that whenever the $\omega$-valuation function is applied, then the valuation should take into account only a finite number of first terms of an infinite sequence. This in our example is expressed by the fourth case in the definition of the $\omega$-valuation function.

\section{Weighted generalized Büchi automata with $\varepsilon$-transitions
over product $\omega$-valuation monoids, and generalized product $\omega$-valuation monoids}

Let $\left(  K,+,\cdot,Val^{\omega},\mathbf{0}%
,\mathbf{1}\right)  $ be an idempotent product $\omega$-valuation monoid, and $A$ be an alphabet.
We introduce now the models of weighted generalized Büchi automata with
$\varepsilon$-transitions and weighted Büchi automata with $\varepsilon
$-transitions over $A$, and $K$. We note that weighted
Büchi automata over $\omega$-valuation monoids have already been considered in
\cite{Dr-Me}, \cite{Me-Do}.
\begin{definition}
\label{definition2}
(i) A weighted generalized Büchi automaton with $\varepsilon$-transitions
($\varepsilon$-wgBa for short) over $A$ and $K$ is a quadruple $\mathcal{A}%
=\left(  Q,wt,I,\mathcal{F}\right)  $, where $Q$ is the finite set of states,
$wt:Q\times\left(  A\cup\left\{  \varepsilon\right\}  \right)  \times
Q\rightarrow K$ is a mapping assigning weights to the transitions of the
automaton, $I$ is the set of initial states and $\mathcal{F=}\left\{
F_{1},\ldots,F_{l}\right\}  $ is the set of final sets $F_{i}\in
\mathcal{P}\left(  Q\right)  ,$ for every $1\leq i\leq l.$ For every $t\in
Q\times\left\{  \varepsilon\right\}  \times Q$ we require that $wt\left(
t\right)  =\mathbf{0}$ or $wt\left(  t\right)  =\mathbf{1}.$ Moreover, for
every $\left(  q,\varepsilon,q^{\prime}\right)  \in Q\times\left\{
\varepsilon\right\}  \times Q$ with $wt  \left(  q,\varepsilon
,q^{\prime}\right)   =\mathbf{1},$ and every $i\in\left\{
1,\ldots,l\right\}  $, we have $q\in F_{i}$ iff $q^{\prime}\in F_{i}.$

(ii) An $\varepsilon$-wgBa is a weighted Büchi automaton with $\varepsilon
$-transitions ($\varepsilon$-wBa for short) if $l=1,$ i.e., there is only one
final set.

(iii) An $\varepsilon$-wBa is a weighted Büchi automaton (wBa for short) if
for every $q,q^{\prime}\in Q$ it holds $wt  \left(  q,\varepsilon
,q^{\prime}\right)    =\mathbf{0}.$
\end{definition}

If $\mathcal{A=}\left(  Q,wt,I,\mathcal{F}\right)  $ is an $\varepsilon$-wBa,
then we simply write $\mathcal{A=}\left(  Q,wt,I,F\right)  $. Let
$w=a_{0}a_{1}\ldots\in A^{\omega}$ with $a_{i}\in A\left(  i\geq0\right)  .$ A
path $P_{w}$ of $\mathcal{A}$ over $w$ is an infinite sequence of transitions
$P_{w}=\left(  q_{j},b_{j},q_{j+1}\right)  _{j\geq0}$, $b_{j}\in A\cup\left\{
\varepsilon\right\}  $ $\left(  j\geq0\right)  $, such that $w=b_{0}%
b_{1}\ldots.$ Let $i_{0}\leq i_{1}\leq i_{2}\leq i_{3}\leq\ldots$ be the
sequence of positions with $b_{i_{k}}=a_{k}$ for every $k\geq0,$ and
$h_{0}\leq h_{1}\leq h_{2}\leq h_{3}\leq\ldots$ be the sequence of positions
with $b_{h_{l}}=\varepsilon$ for every $l\geq0$. Then, we let the weight of
$P_{w}$ be the value
\[
weight_{\mathcal{A}}\left(  P_{w}\right)  =%
\begin{cases}%
\begin{array}{ll}
Val^{\omega}\left( wt\left(   q_{i_{k}},a_{k},q_{i_{k}+1}\right)  \right)
_{k\geq0} &
\begin{array}{l}
\text{if\textrm{ }}wt\left(    q_{h_{l}},\varepsilon,q_{h_{l}+1}
\right)  =\mathbf{1}\\
\text{for every }l\geq0
\end{array}
\\
\mathbf{0} & \text{ otherwise}%
\end{array}
\end{cases}
\]

Let $P_{w}=\left(  q_{j},b_{j},q_{j+1}\right)  _{j\geq0}$, $b_{j}\in
A\cup\left\{  \varepsilon\right\}  $ $\left(  j\geq0\right)  $ be a path of
$\mathcal{A}$ over $w.$ The set of states that appear infinitely often along
$P_{w}$ is denoted by $In^{Q}\left(  P_{w}\right)  .$ The path $P_{w}$ is
called \textit{successful} if $q_{0}\in I$ and $In^{Q}\left(  P_{w}\right)
\cap F_{i}\neq\emptyset$, for every $i\in\left\{  1,\ldots,l\right\}  .$ We
shall denote by $succ_{\mathcal{A}}\left(  w\right)  $ the set of all
successful paths of $\mathcal{A}$ over $w.$ The behavior of $\mathcal{A}$ is
the infinitary series $\left\Vert \mathcal{A}\right\Vert :A^{\omega
}\rightarrow K$ with coefficients specified, for every $w\in A^{\omega}$,
\[
\left(  \left\Vert \mathcal{A}\right\Vert ,w\right)  =\underset{P_{w}\in
succ_{\mathcal{A}}\left(  w\right)  }{%
{\displaystyle\sum}
}weight_{\mathcal{A}}\left(  P_{w}\right)  .
\]

\begin{remark}
In the definition of our $\varepsilon$-wgBa we impose a restriction on the weights that $wt$ assigns to $\varepsilon$-transitions. More specifically, we require that $\varepsilon$-transitions have weight $\mathbf{0}$, or $\mathbf{1}$, and that $\mathbf{1}$-weight $\varepsilon$-transitions are only allowed between states that belong to the same final subsets of the automaton. As it will be presented in the sequel, this restriction is sufficient for expressing the intuition of the translation of our weighted formulas to $\varepsilon$-wgBa, as we will need transitions that reduce formulas without modifying the weight of the path, and in the same time respecting the conditions imposed by final subsets. In the general framework, with this definition we obtain a generalization of wBa that allows flexibility to move between states, and at the same time respects acceptance conditions, and weight computation that are determined by transitions that consume a letter of the input word.
\end{remark}
\begin{remark}
As mentioned before, wBa over $\omega$-valuation monoids have already been considered in
\cite{Dr-Me}, \cite{Me-Do}. In contrast to \cite{Dr-Me}, \cite{Me-Do}, in our notations we do not explicitly define a set of transitions as a subset of all possible triples $\left( q,a,q^{\prime}\right)$, and then use a weight function to assign weights to the elements of this subset. Our weight function $wt$ assigns weights to all possible transitions of the automaton, and then, similarly to \cite{Dr-Me}, \cite{Me-Do}, we obtain the weight of a path by applying $Val^{\omega}$ to the sequence of weights of the transitions of the path. By the definition of $Val^{\omega}$, if the weight of one transition is $\mathbf{0}$, then the weight of a path is $\mathbf{0}$. This implies that given a wBa defined by Definition \ref{definition2}, we can obtain an equivalent wBa defined as in \cite{Dr-Me} with a set of transitions at least the ones with non-zero weight at the original automaton, and vice-versa. Given a wBa defined as in \cite{Dr-Me}, we can construct a wBa defined by Definition \ref{definition2}, by assigning the weight $\mathbf{0}$ to every tuple $\left( q,a,q^{\prime}\right)$ that does not belong to the set of transitions of the original wBa. Hence the two notations lead to equivalent definitions. We note that for \cite{Me-Do} we refer to the simplest form of wBa introduced in that paper.
\end{remark}
Two $\varepsilon$-wgBa are called equivalent if they have the same behavior.
We shall also denote an $\varepsilon$-transition with weight=$\mathbf{1}$ by
$\overset{\varepsilon}{\rightarrow}$ and we will write $\overset{\ast
}{\rightarrow}$ for the transitive and reflexive closure of $\overset
{\varepsilon}{\rightarrow}.$ Finally, for every $w=a_{0}\ldots a_{n}\in A^{+}$
we shall denote by $q\overset{w}{\rightarrow}q^{\prime}$ a sequence of
transitions $\left(  q_{j},a_{j},q_{j+1}\right)  _{0\leq j\leq n}$ with
$q_{0}=q,$ and $q_{n+1}=q^{\prime}.$ Now, we let
\[
pri_{\mathcal{A}}\left(  w\right)  =\left\{  weight_{\mathcal{A}}\left(
P_{w}\right)  \mid P_{w}\in succ_{\mathcal{A}}\left(  w\right)  \right\}
\]
for every $w\in A^{\omega}.$

$\varepsilon$-wgBa (resp. $\varepsilon$-wBa, and
wBa) over generalized product $\omega$-valuation monoids are defined in the same
way with $\varepsilon$-wgBa (resp. $\varepsilon$-wBa, and wBa) over product $\omega$-valuation monoids presented above.

\begin{lemma}
Let $\left(  K,+,\cdot,Val^{\omega},\mathbf{0},\mathbf{1}\right)  $ be an
idempotent product $\omega$-valuation monoid or an idempotent generalized product $\omega$-valuation monoid. For every $\varepsilon$-wgBa over $A$ and $K$ we can effectively construct an
equivalent $\varepsilon$-wBa.
\end{lemma}

\begin{proof}
Let $\mathcal{A=}\left(  Q,wt,I,\mathcal{F}\right)  $ be an $\varepsilon$-wgBa
over $A$ and $K$ with $\mathcal{F=}\left\{  F_{1},\ldots,F_{l}\right\}  $. We
let $\mathcal{A}^{\prime}=\left(  Q^{\prime},I^{\prime},wt^{\prime},F^{\prime
}\right)  $ be an $\varepsilon$-wBa defined as follows:

\begin{itemize}
\item $Q^{\prime}=Q\times\left\{  1,\ldots,l\right\}  ,$

\item $I^{\prime}=I\times\left\{  1\right\}  ,$

\item For every $\left(  \left(  q,i\right)  ,b,\left(  q^{\prime},j\right)
\right)  \in Q^{\prime}\times\left(  A\cup\left\{  \varepsilon\right\}
\right)  \times Q^{\prime}$ we let:

$wt^{\prime}\left( \left(    q,i\right)  ,b,\left(  q^{\prime},j\right) \right)
\\=\left\{
\begin{array}
[c]{ll}%
wt \left(  q,b,q^{\prime}\right)    &
\begin{array}
[c]{l}%
\text{if \ }\left(  b\in A,i=j\text{ and }q\notin F_{i}\right) \\
\text{or }\left(  b\in A,j\equiv\left(  i+1\right)  \operatorname{mod}l\text{
and }q\in F_{i}\right) \\
\text{or }\left(  b=\varepsilon\text{ and }i=j\right)
\end{array}
\\
\mathbf{0} & \text{ otherwise}%
\end{array}
\right.  ,$

\item $F^{\prime}=F_{1}\times\left\{  1\right\}  .$
\end{itemize}

We will prove that $pri_{\mathcal{A}^{\prime}}\left(  w\right)  \backslash
\left\{  \mathbf{0}\right\}  =pri_{\mathcal{A}}\left(  w\right)
\backslash\left\{  \mathbf{0}\right\}  $ for every $w\in A^{\omega}.$ To this
end let $w=a_{0}a_{1}\ldots\in A^{\omega}$ and $P_{w}=\left(  q_{i}%
,b_{i},q_{i+1}\right)  _{i\geq0}$ be a successful path of $\mathcal{A}$ over
$w$ with $weight_{\mathcal{A}}\left(  P_{w}\right)  \neq\mathbf{0}.$ Moreover,
let $i_{0}\leq i_{1}\leq i_{2}\leq i_{3}\leq\ldots$ be the sequence of
positions with $b_{i_{k}}=a_{k}$ for every $k\geq0,$ and $h_{0}\leq
h_{1}\leq h_{2}\leq h_{3}\leq \ldots$ be the sequence of positions with $b_{h_{l}%
}=\varepsilon$ for every $l\geq0$. It holds
\[
Val^{\omega}\left(  wt\left(  q_{i_{k}},a_{k},q_{i_{k}+1}\right)  \right)
_{k\geq0}\neq\mathbf{0}\text{ and }wt\left(   q_{h_{k}},\varepsilon
,q_{h_{k}+1}  \right)  =\mathbf{1}\text{ for all }k\geq0\text{. }%
\]
We define the path
\[
P_{w}^{\prime}=  \left(  \left(  q_{i},g_{i}\right)  ,b_{i},\left(
q_{i+1},g_{i+1}\right)    \right)  _{i\geq0}%
\]
of $\mathcal{A}^{\prime}$ over $w$ by setting $g_{0}=1$ and for every $i\geq
1$, we point out the following cases. If $b_{i-1}\neq\varepsilon$, we let
$g_{i}=g_{i-1}$ if $q_{i-1}\notin F_{g_{i-1}},$ and $g_{i}\equiv\left(
g_{i-1}+1\right)  \operatorname{mod}l$ if $q_{i-1}\in F_{g_{i-1}}.$ If
$b_{i-1}=\varepsilon$, we let $g_{i}=g_{i-1}.$ By construction $P_{w}%
^{\prime}$ is successful and for every $i\geq0$ it holds
\[
wt^{\prime}\left(  \left(  q_{i},g_{i}\right)  ,b_{i},\left(  q_{i+1}%
,g_{i+1}\right)  \right)  =wt \left(  q_{i},b_{i},q_{i+1}
\right)
\]
which implies that
\begin{align*}
&  weight_{\mathcal{A}^{\prime}}\left(  P_{w}^{\prime}\right) \\
&  =Val^{\omega}\left(  wt^{\prime} \left(  \left(  q_{i_{k}},g_{i_{k}%
}\right)  ,a_{k},\left(  q_{i_{k}+1},g_{i_{k}+1}\right)  \right)
\right)  _{k\geq0}\\
&  =Val^{\omega}\left(  wt\left(    q_{i_{k}},a_{k},q_{i_{k}%
+1}\right)    \right)  _{k\geq0}\\
&  =weight_{\mathcal{A}}\left(  P_{w}\right)  .
\end{align*}
We thus get $pri_{\mathcal{A}}\left(  w\right)  \backslash\left\{
\mathbf{0}\right\}  \subseteq pri_{\mathcal{A}^{\prime}}\left(  w\right)
\backslash\left\{  \mathbf{0}\right\}  .$ In order to prove the opposite
inclusion we let $P_{w}^{\prime}=\left(  \left(  q_{i},g_{i}\right)
,b_{i},\left(  q_{i+1},g_{i+1}\right)  \right)  _{i\geq0}$ be a successful
path of $\mathcal{A}^{\prime}$ over $w$ with $weight_{\mathcal{A}^{\prime}%
}\left(  P_{w}^{\prime}\right)  \neq\mathbf{0},$ which by construction of
$\mathcal{A}^{\prime}$ implies $wt^{\prime}\left(  \left(  q_{i}%
,g_{i}\right)  ,b_{i},\left(  q_{i+1},g_{i+1}\right)    \right)
=wt\left(    q_{i},b_{i},q_{i+1}  \right)  $ and either $b_{i}\in
A,g_{i}=g_{i+1}$ and $q_{i}\notin F_{g_{i}},$ or $b_{i}\in A,$ $g_{i+1}%
=\left(  g_{i}+1\right)  \operatorname{mod}l$ and $q_{i}\in F_{g_{i}},$ or
$b_{i}=\varepsilon$ and $g_{i}=g_{i+1}.$ Then, $P_{w}=\left(  q_{i}%
,b_{i},q_{i+1}\right)  _{i\geq0}$ is a successful path of $\mathcal{A}$ over
$w$ and
\begin{align*}
&  weight_{\mathcal{A}^{\prime}}\left(  P_{w}^{\prime}\right) \\
&  =Val^{\omega}\left(  wt^{\prime}\left(  \left(  q_{i_{k}},g_{i_{k}%
}\right)  ,a_{k},\left(  q_{i_{k}+1},g_{i_{k}+1}\right)    \right)
\right)_{k\geq0} \\
&  =Val^{\omega}\left(  wt\left(   q_{i_{k}},a_{k},q_{i_{k}%
+1}\right)   \right)  _{k\geq0}\\
&  =weight_{\mathcal{A}}\left(  P_{w}\right)
\end{align*}
where the sequence of positions $i_{0}\leq i_{1}\leq i_{2}\leq i_{3}\leq
\ldots$ is defined as before. Thus, $pri_{\mathcal{A}^{\prime}}\left(
w\right)  \backslash\left\{  \mathbf{0}\right\}  \subseteq pri_{\mathcal{A}%
}\left(  w\right)  \backslash\left\{  \mathbf{0}\right\}  .$ Which implies
that
\[
pri_{\mathcal{A}^{\prime}}\left(  w\right)  \backslash\left\{  \mathbf{0}%
\right\}  =pri_{\mathcal{A}}\left(  w\right)  \backslash\left\{
\mathbf{0}\right\}  .
\]

Hence,
\begin{align*}
\left(  \left\Vert \mathcal{A}^{\prime}\right\Vert ,w\right)   &
=\underset{k\in pri_{\mathcal{A}^{\prime}}\left(  w\right)  }{%
{\displaystyle\sum}
}k=\underset{k\in pri_{\mathcal{A}^{\prime}}\left(  w\right)  \backslash
\left\{  \mathbf{0}\right\}  }{%
{\displaystyle\sum}
}k\\
&  =\underset{k\in pri_{\mathcal{A}}\left(  w\right)  \backslash\left\{
\mathbf{0}\right\}  }{%
{\displaystyle\sum}
}k=\underset{k\in pri_{\mathcal{A}}\left(  w\right)  }{%
{\displaystyle\sum}
}k=\left(  \left\Vert \mathcal{A}\right\Vert ,w\right)  ,
\end{align*}
where the first and last equality hold due to the completeness axioms of the monoid and Lemma \ref{Properties Sum}(ii), and this concludes our proof.
For idempotent generalized product $\omega$-valuation monoids we can prove the lemma's claim using the same arguments.
\end{proof}

We shall need some auxiliary definitions. Let $\mathcal{A=}\left(
Q,wt,I,F\right)  $ be an $\varepsilon$-wBa over $A$ and $K,$ $w=a_{0}%
a_{1}\ldots\in A^{\omega}$ and $P_{w}=\left(  q_{i},a_{i},q_{i+1}\right)
_{i\geq0}\in$ $succ_{\mathcal{A}}\left(  w\right)  $ with no $\varepsilon
$-transitions. We consider the set of paths $Paths\left(  P_{w}\right)
\subseteq$ $succ_{\mathcal{A}}\left(  w\right)  $ containing $P_{w}$ and every
path derived by $P_{w}$ if we replace one or more transitions $\left(
q_{i},a_{i},q_{i+1}\right)  ,i\geq0,$ by a sequence of transitions of the form
$q_{i}\overset{\ast}{\rightarrow}\widetilde{q}\overset{a_{i}}{\rightarrow
}\overline{q}\overset{\ast}{\rightarrow}q_{i+1}.$ Furthermore, we let
\[
\mathcal{V}_{P_{w}}=\left\{  k\in K\mid\exists\widetilde{P}_{w}\in
Paths\left(  P_{w}\right)  \text{ with }weight_{\mathcal{A}}\left(
\widetilde{P}_{w}\right)  =k\right\}  .
\]
\

\begin{lemma}
(i) Let $\left(  K,+,\cdot,Val^{\omega},\mathbf{0},\mathbf{1}\right)  $ be an
idempotent product $\omega$-valuation monoid. For every
$\varepsilon$-wBa over $A$ and $K$ we can effectively construct an equivalent
wBa over $A$ and $K$.

(ii) Let $\left(  K,+,\cdot,Val^{\omega},\mathbf{0},\mathbf{1}\right)  $ be an
idempotent generalized product $\omega$-valuation monoid. For every
$\varepsilon$-wBa over $A$ and $K$ we can effectively construct an equivalent
wBa over $A$ and $K$.
\end{lemma}

\begin{proof}
(i) Let $\mathcal{A=}\left(  Q,wt,I,F\right)  $ be an $\varepsilon$-wBa over
$A$ and $K.$ We define the wBa $\mathcal{A}^{\prime}=\left(  Q,wt^{\prime
},I,F\right)  $ by setting
\[
wt^{\prime}\left(    q,a,q^{\prime}  \right)  =\underset
{q\overset{\ast}{\rightarrow}\widetilde{q},\overline{q}\overset{\ast
}{\rightarrow}q^{\prime}}{\underset{\widetilde{q},\overline{q}\in Q}{\sum}%
}wt\left(    \widetilde{q},a,\overline{q} \right)  .
\]
We let $w=a_{0}a_{1}\ldots\in A^{\omega}$ and $P_{w}^{\prime}=\left(
q_{i},a_{i},q_{i+1}\right)  _{i\geq0}\in succ_{\mathcal{A}^{\prime}}\left(
w\right)  .$ For all $i\geq0$ it holds $wt^{\prime}\left(   q_{i}%
,a_{i},q_{i+1}  \right)  =\underset{q_{i}\overset{\ast}{\rightarrow
}\widetilde{q}_{i},\overline{q}_{i}\overset{\ast}{\rightarrow}q_{i+1}%
}{\underset{\widetilde{q}_{i},\overline{q}_{i}\in Q}{\sum}}wt\left(
\widetilde{q}_{i},a_{i},\overline{q}_{i} \right)  $. Then,
\begin{align*}
weight_{\mathcal{A}^{\prime}}\left(  P_{w}^{\prime}\right)   &  =Val^{\omega
}\left(  \underset{q_{i}\overset{\ast}{\rightarrow}\widetilde{q}_{i}%
,\overline{q}_{i}\overset{\ast}{\rightarrow}q_{i+1}}{\underset{\widetilde
{q}_{i},\overline{q}_{i}\in Q}{\sum}}wt\left(    \widetilde{q}_{i}%
,a_{i},\overline{q}_{i}\right)   \right)  _{i\geq0}\\
&  =\underset{q_{i}\overset{\ast}{\rightarrow}\widetilde{q}_{i},\overline
{q}_{i}\overset{\ast}{\rightarrow}q_{i+1}}{\underset{\widetilde{q}%
_{i},\overline{q}_{i}\in Q}{\underset{i\geq0}{%
{\displaystyle\sum}
}}}Val^{\omega}\left(  wt\left(   \widetilde{q}_{i},a_{i},\overline
{q}_{i}\right)    \right)  _{i\geq0}%
\end{align*}
where the second equality holds by the distributivity of $Val^{\omega}$ over finite sums. Clearly,
$P_{w}^{\prime}$ is also a successful path of $\mathcal{A}$ over $w$ and it
holds $weight_{\mathcal{A}^{\prime}}\left(  P_{w}^{\prime}\right)=\underset
{\widetilde{P}_{w}\in Paths\left(  P_{w}^{\prime}\right)  }{%
{\displaystyle\sum}
}weight_{\mathcal{A}}\left(  \widetilde{P}_{w}\right)
=\underset{k\in\mathcal{V}_{P_{w}^{\prime}}}{%
{\displaystyle\sum}
}k$. The last equality is concluded by the completeness axioms of the monoid and Lemma \ref{Properties Sum}(ii). Hence, we get
\begin{align*}
\left(  \left\Vert \mathcal{A}^{\prime}\right\Vert ,w\right)   &
=\underset{P_{w}^{\prime}\in succ_{\mathcal{A}^{\prime}}\left(  w\right)  }{%
{\displaystyle\sum}
}\left(  \underset{k\in\mathcal{V}_{P_{w}^{\prime}}}{%
{\displaystyle\sum}
}k\right)
&  =\underset{P_{w}^{\prime}\in succ_{\mathcal{A}^{\prime}}\left(  w\right)
}{%
{\displaystyle\sum}
}\left(  \underset{k\in\mathcal{V}_{P_{w}^{\prime}}\backslash\left\{
\mathbf{0}\right\}  }{%
{\displaystyle\sum}
}k\right)  .
\end{align*}

We show that $\left(  \underset{P_{w}^{\prime}\in succ_{\mathcal{A}^{\prime}%
}\left(  w\right)  }{%
{\displaystyle\bigcup}
}\mathcal{V}_{P_{w}^{\prime}}\right)  \backslash\left\{  \mathbf{0}\right\}
=pri_{\mathcal{A}}\left(  w\right)  \backslash\left\{  \mathbf{0}\right\}  .$
The first inclusion $\left(  \underset{P_{w}^{\prime}\in succ_{\mathcal{A}%
^{\prime}}\left(  w\right)  }{%
{\displaystyle\bigcup}
}\mathcal{V}_{P_{w}^{\prime}}\right)  \backslash\left\{  \mathbf{0}\right\}
\subseteq pri_{\mathcal{A}}\left(  w\right)  \backslash\left\{  \mathbf{0}%
\right\}  $ holds by definition of $\mathcal{V}_{P_{w}^{\prime}}.$ To prove
the converse inclusion, i.e., $pri_{\mathcal{A}}\left(  w\right)
\backslash\left\{  \mathbf{0}\right\}  \subseteq\left(  \underset
{P_{w}^{\prime}\in succ_{\mathcal{A}^{\prime}}\left(  w\right)  }{%
{\displaystyle\bigcup}
}\mathcal{V}_{P_{w}^{\prime}}\right)  \backslash\left\{  \mathbf{0}\right\}
$, we prove that for every $k\in pri_{\mathcal{A}}\left(  w\right)
\backslash\left\{  \mathbf{0}\right\}  $ there exists $P_{w}^{\prime}\in
succ_{\mathcal{A}^{\prime}}\left(  w\right)  $ such that $k\in\mathcal{V}%
_{P_{w}^{\prime}}.$ To this end we fix a $k\in pri_{\mathcal{A}}\left(
w\right)  \backslash\left\{  \mathbf{0}\right\}  $ and we let
\begin{align*}
P_{w}  &  :q_{0}\overset{\ast}{\rightarrow}q_{i_{1}}\overset{w\left(
0\right)  \ldots w\left(  k_{1}-1\right)  }{\rightarrow}q_{i_{1}+k_{1}%
}\overset{\ast}{\rightarrow}q_{i_{2}}\overset{w\left(  k_{1}\right)  \ldots
w\left(  k_{1}+k_{2}-1\right)  }{\rightarrow}\\
&  q_{i_{2}+k_{2}}\overset{\ast}{\rightarrow}q_{i_{3}}\overset{w\left(
k_{1}+k_{2}\right)  \ldots w\left(  k_{1}+k_{2}+k_{3}-1\right)  }{\rightarrow
}q_{i_{3}+k_{3}}\overset{\ast}{\rightarrow}\ldots
\end{align*}
be a successful path of $\mathcal{A}$ over $w$ with $weight_{\mathcal{A}%
}\left(  P_{w}\right)  =k,k_{j}\in%
\mathbb{N}
,j\geq1.$ We define the path $P_{w}^{\prime}$ of $\mathcal{A}^{\prime}$ by
setting
\begin{align*}
P_{w}^{\prime}  &  :q_{0}\overset{w\left(  0\right)  }{\longrightarrow
}q_{i_{1}+1}\overset{w\left(  1\right)  }{\longrightarrow}q_{i_{1}+2}\ldots\\
&  \ldots q_{i_{1}+k_{1}-1}\overset{w\left(  k_{1}-1\right)  }{\longrightarrow
}q_{i_{1}+k_{1}}\overset{w\left(  k_{1}\right)  }{\longrightarrow}q_{i_{2}%
+1}\ldots\\
&  \ldots q_{i_{2}+k_{2}-1}\overset{w\left(  k_{1}+k_{2}-1\right)
}{\longrightarrow}q_{i_{2}+k_{2}}\ldots.
\end{align*}
By Definition \ref{definition2} we get that $P_{w}^{\prime}\in
succ_{\mathcal{A}^{\prime}}\left(  w\right)  $ (observe that for every
$j\geq0,$ $q_{i_{j}+k_{j}}\in F\Leftrightarrow q_{i_{j+1}}\in F,$ where
$k_{0}=i_{0}=0).$ Moreover, $P_{w}^{\prime}\in succ_{\mathcal{A}}\left(
w\right)  $ and $weight_{\mathcal{A}}\left(  P_{w}\right)  \in\mathcal{V}%
_{P_{w}^{\prime}}$ as wanted. Hence,
\begin{align*}
\left(  \left\Vert \mathcal{A}^{\prime}\right\Vert ,w\right)   &
=\underset{P_{w}^{\prime}\in succ\left(  \mathcal{A}^{\prime}\right)  }{%
{\displaystyle\sum}
}\left(  \underset{k\in\mathcal{V}_{P_{w}^{\prime}}\backslash\left\{
\mathbf{0}\right\}  }{%
{\displaystyle\sum}
}k\right) \\
&  =\underset{k\in pri_{\mathcal{A}}\left(  w\right)  \backslash\left\{
\mathbf{0}\right\}  }{\sum}k=\underset{k\in pri_{\mathcal{A}}\left(  w\right)
}{\sum}k=\left(  \left\Vert \mathcal{A}\right\Vert ,w\right)
\end{align*}
and the proof is completed.

(ii) Let $\mathcal{A=}\left(  Q,wt,I,F\right)  $ be an $\varepsilon$-wBa over
$A$ and $K.$ We define the wBa $\overline{\mathcal{A}}=\left(  \overline
{Q},\overline{wt},\overline{I},\overline{F}\right)  $ by setting $\overline
{Q}=Q\cup S$ where $S=\left\{  s_{q}\mid q\in Q\right\}  $ is a set of new states. Moreover, we let
$\overline{I}=\left\{  s_{q}\mid q\in I\right\}\cup I  $, $\overline{F}=\left\{
s_{q}\mid q\in F\right\}\cup F ,$ and for every $\left(  p,a,\overline{p}\right)
\in\overline{Q}\times A\times\overline{Q}$ we set

\begin{itemize}
\item $\overline{wt}\left( p,a,\overline{p}\right) =\left\{
\begin{array}{ll}
\mathbf{1} &
\begin{array}{l}
\text{if }p=s_{q},\overline{p}=s_{\overline{q}}\in S \\
\text{and there exist }\widetilde{q},q^{\prime }\in Q, \\
\text{with }wt\left( \widetilde{q},a,q^{\prime }\right) =\mathbf{1}, \\
q\overset{\ast }{\rightarrow }\widetilde{q},q^{\prime }\overset{\ast }{%
\rightarrow }\overline{q}%
\end{array}
\\
&  \\
\underset{wt\left( \widetilde{q},a,\overline{q}\right) \neq \mathbf{1},%
\mathbf{0}}{\underset{q_{1}\overset{\ast }{\rightarrow }\widetilde{q},%
\overline{q}\overset{\ast }{\rightarrow }q_{2}}{\underset{\widetilde{q},%
\overline{q}\in Q}{\sum }}}wt\left( \widetilde{q},a,\overline{q}\right)  &
\begin{array}{l}
\text{if }q_{1}=p,q_{2}=\overline{p}\in Q,\text{ or if } \\
p=q_{1}\in Q,\text{ }\overline{p}=s_{q_{2}}\in S,\text{ \ or if} \\
p=s_{q_{1}}\in S,\overline{p}=q_{2}\in Q%
\end{array}
\\
&  \\
\mathbf{0} & \text{ otherwise}%
\end{array}%
\text{ }\right. .$
\end{itemize}

We let $w=a_{0}a_{1}\ldots\in A^{\omega}$ and $P_{w}^{\prime}=\left(
q_{i},a_{i},q_{i+1}\right)  _{i\geq0}\in succ_{\overline{\mathcal A}}\left(
w\right)  $, with non-zero weight, such that $q_{i}\in Q$ $\ $for all $i\geq0$. Then,
\begin{align*}
weight_{\mathcal{\overline{\mathcal A}}}\left(  P_{w}^{\prime}\right)   &  =Val^{\omega
}\left(  \underset{wt\left(  \widetilde{q}_{i},a,\overline{q}%
_{i}\right)    \neq\mathbf{1},\mathbf{0}}{\underset{q_{i}\overset{\ast}{\rightarrow
}\widetilde{q}_{i},\overline{q}_{i}\overset{\ast}{\rightarrow}q_{i+1}%
}{\underset{\widetilde{q}_{i},\overline{q}_{i}\in Q}{\sum}}}wt\left(
\widetilde{q}_{i},a_{i},\overline{q}_{i}\right)   \right)  _{i\geq0}\\
&  =\underset{wt\left(    \widetilde{q}_{i},a,\overline{q}_{i}
\right)  \neq\mathbf{1},\mathbf{0}}{\underset{q_{i}\overset{\ast}{\rightarrow}\widetilde{q}%
_{i},\overline{q}_{i}\overset{\ast}{\rightarrow}q_{i+1}}{\underset
{\widetilde{q}_{i},\overline{q}_{i}\in Q}{\sum}}}Val^{\omega}\left(  wt\left(
 \widetilde{q}_{i},a_{i},\overline{q}_{i}\right)  \right)
_{i\geq0}
\end{align*}
where the second equality is obtained by the distributivity of $Val^{\omega}$ over finite sums for generalized product $\omega$-valuation monoids. For every $P_{w}^{\prime}=\left(  q_{i},a_{i},q_{i+1}\right)  _{i\geq0}\in
succ_{\overline{\mathcal{A}}}\left(  w\right)  $ such that $q_{i}\in Q$ $\ $for
all $i\geq0,$  we consider the set of paths $\overline{Paths}\left(  P_{w}^{\prime}\right)  $
that contains all the paths derived by $P_{w}^{\prime}$ if we replace one or
more states in $P_{w}^{\prime}$ by its decoy in $S.$ Clearly, every path in
$\overline{Paths}\left(  P_{w}^{\prime}\right)  $ is a successful path of
$\overline{\mathcal{A}}$ over $w$ and it holds
\[
weight_{\overline{\mathcal{A}}}\left(  P_{w}^{\prime}\right)  +\underset
{\overline{P}_{w}\in\overline{Paths}\left(  P_{w}^{\prime}\right)  }{\sum
}weight_{\overline{\mathcal{A}}}\left(  \overline{P}_{w}\right)
=\underset{k\in\mathcal{V}_{P_{w}^{\prime}}}{%
{\displaystyle\sum}
}k.
\]
Hence, we get
\begin{align*}
&  \left(  \left\Vert \overline{\mathcal{A}}\right\Vert ,w\right)  \\
&  =\underset{P_{w}^{\prime}\in succ_{\overline{\mathcal{A}}}\left(  w\right)
\cap\left(  Q\times A\times Q\right)  ^{\omega}}{%
{\displaystyle\sum}
}weight_{\overline{\mathcal{A}}}\left(  P_{w}^{\prime}\right)  +\underset
{\overline{P}_{w}\in\overline{Paths}\left(  P_{w}^{\prime}\right)  }%
{\underset{P_{w}^{\prime}\in succ_{\overline{\mathcal{A}}}\left(  w\right)
\cap\left(  Q\times A\times Q\right)  ^{\omega}}{%
{\displaystyle\sum}
}}weight_{\overline{\mathcal{A}}}\left(  \overline{P}_{w}\right)  \\
&  =\underset{P_{w}^{\prime}\in succ_{\overline{\mathcal{A}}}\left(  w\right)
\cap\left(  Q\times A\times Q\right)  ^{\omega}}{%
{\displaystyle\sum}
}\left(  weight_{\overline{\mathcal{A}}}\left(  P_{w}^{\prime}\right)
+\underset{\overline{P}_{w}\in\overline{Paths}\left(  P_{w}^{\prime}\right)
}{%
{\displaystyle\sum}
}weight_{\overline{\mathcal{A}}}\left(  \overline{P}_{w}\right)  \right)  \\
&  =\underset{P_{w}^{\prime}\in succ_{\overline{\mathcal{A}}}\left(  w\right)
\cap\left(  Q\times A\times Q\right)  ^{\omega}}{%
{\displaystyle\sum}
}\left(  \underset{k\in\mathcal{V}_{P_{w}^{\prime}}}{%
{\displaystyle\sum}
}k\right)  \\
&  =\underset{P_{w}^{\prime}\in succ_{\overline{\mathcal{A}}}\left(  w\right)
\cap\left(  Q\times A\times Q\right)  ^{\omega}}{%
{\displaystyle\sum}
}\left(  \underset{k\in\mathcal{V}_{P_{w}^{\prime}}\backslash\left\{
\mathbf{0}\right\}  }{%
{\displaystyle\sum}
}k\right)  .
\end{align*}
We note that the second equality holds by the completeness axioms of the monoid.
Using the same arguments as in case (i), and taking into account that for every $k\in
pri_{\mathcal{A}}\left(  w\right)  \backslash\left\{  \mathbf{0}\right\}  $
there exists $P_{w}^{\prime}\in succ_{\overline{\mathcal{A}}}\left(  w\right)
\cap\left(  Q\times A\times Q\right)  ^{\omega}$ such that $k\in
\mathcal{V}_{P_{w}^{\prime}}$, we can prove that $\left(  \underset
{P_{w}^{\prime}\in succ_{\overline{\mathcal{A}}}\left(  w\right)  \cap\left(
Q\times A\times Q\right)  ^{\omega}}{%
{\displaystyle\bigcup}
}\mathcal{V}_{P_{w}^{\prime}}\right)  \backslash\left\{  \mathbf{0}\right\}
=pri_{\mathcal{A}}\left(  w\right)  \backslash\left\{  \mathbf{0}\right\}  ,$
and thus it holds%
\begin{align*}
\left(  \left\Vert \overline{\mathcal{A}}\right\Vert ,w\right)   &
=\underset{P_{w}^{\prime}\in succ_{\overline{\mathcal{A}}}\left(  w\right)
\cap\left(  Q\times A\times Q\right)  ^{\omega}}{%
{\displaystyle\sum}
}\left(  \underset{k\in\mathcal{V}_{P_{w}^{\prime}}\backslash\left\{
\mathbf{0}\right\}  }{%
{\displaystyle\sum}
}k\right)  \\
&  =\underset{k\in pri_{\mathcal{A}}\left(  w\right)  \backslash\left\{
\mathbf{0}\right\}  }{\sum}k=\underset{k\in pri_{\mathcal{A}}\left(  w\right)
}{\sum}k=\left(  \left\Vert \mathcal{A}\right\Vert ,w\right)
\end{align*}
as wanted.
\end{proof}

\section{Weighted LTL over product $\omega$-valuation monoids}

In what follows, we present our definition of the weighted $LTL$ over product $\omega$-valuation monoids. We recall that a weighted $LTL$ has appeared for the first time in \cite{Ku-Lu}. This follows the definition of weighted $MSO$ logic over semirings presented in \cite{Dr-We}. We also recall that a weighted $MSO$ logic over $\omega$-valuation monoids was defined in \cite{Dr-Me}, where the $\omega$-valuation function has been used for the definition of the semantics of the universal first order, and second order quantifier. Analogously, we will use the $\omega$-valuation function of the underlying structure for the definition of the semantics of the \emph{always} operator of our logic.

Let $AP$ be a finite set of atomic propositions and $K=\left(
K,+,\cdot,Val^{\omega},\mathbf{0},\mathbf{1}\right)  $ be an idempotent
product $\omega$-valuation monoid. In the sequel we shall denote the elements
of $AP$ by $a,b,c,\ldots.$ The syntax of the weighted $LTL$ over $AP$ and $K$
is given by the grammar
\[
\varphi::=k\mid a\mid\lnot a\mid\varphi\vee\varphi\mid\varphi\wedge\varphi
\mid\bigcirc\varphi\mid\varphi U\varphi\mid\square\varphi
\]
where $k\in K$ and $a\in AP.$

We denote by $LTL\left(K,AP\right)$ the class of all weighted $LTL$-formulas over $AP$ and $K$.

\bigskip{}

\begin{definition}
The semantics $\Vert\varphi\Vert$ of formulas $\varphi\in LTL\left(
K,AP\right)  $ are represented as infinitary series in $K\left\langle
\left\langle \left(  \mathcal{P}\left(  AP\right)  \right)  ^{\omega
}\right\rangle \right\rangle $ inductively defined in the following way. For
every $w\in\left(  \mathcal{P}\left(  AP\right)  \right)  ^{\omega}$ we set

\begin{itemize}
\item $\left(  \left\Vert k\right\Vert ,w\right)  =k,$

\item $\left(  \left\Vert a\right\Vert ,w\right)  =\left\{
\begin{array}
[c]{ll}%
\mathbf{1} & \text{if }a\in w\left(  0\right) \\
\mathbf{0} & \text{otherwise}%
\end{array}
\right.  ,$

\item $\left(  \left\Vert \lnot a\right\Vert ,w\right)  =\left\{
\begin{array}
[c]{ll}%
\mathbf{1} & \text{if }a\notin w\left(  0\right) \\
\mathbf{0} & \text{otherwise}%
\end{array}
\right.  ,$

\item $\left(  \left\Vert \varphi\wedge\psi\right\Vert ,w\right)  =\left(
\left\Vert \varphi\right\Vert ,w\right)  \cdot\left(  \left\Vert
\psi\right\Vert ,w\right)  ,$

\item $\left(  \left\Vert \varphi\vee\psi\right\Vert ,w\right)  =\left(
\left\Vert \varphi\right\Vert ,w\right)  +\left(  \left\Vert \psi\right\Vert
,w\right)  ,$

\item $\left(  \left\Vert \bigcirc\varphi\right\Vert ,w\right)  =\left(
\left\Vert \varphi\right\Vert ,w_{\geq1}\right)  ,$

\item $\left(  \left\Vert \varphi U\psi\right\Vert ,w\right)  $

$=%
{\displaystyle\sum\limits_{i\geq0}}
Val^{\omega}\left(  \left(  \left\Vert \varphi\right\Vert ,w_{\geq0}\right)
,\ldots,\left(  \left\Vert \varphi\right\Vert ,w_{\geq i-1}\right)  ,\left(
\left\Vert \psi\right\Vert ,w_{\geq i}\right)  ,\mathbf{1},\mathbf{1}%
,\ldots\right)  ,$

\item $\left(  \left\Vert \square\varphi\right\Vert ,w\right)  =Val^{\omega
}\left(  \left(  \left\Vert \varphi\right\Vert ,w_{\geq i}\right)  \right)
_{i\geq0}.$
\end{itemize}
\end{definition}

\bigskip{} We shall denote by $true$ the formula $\mathbf{1}\in K.$ The
syntactic boolean fragment $bLTL\left(  K,AP\right)  $ of $LTL\left(
K,AP\right)  $ is given by the grammar
\[
\varphi::=\mathbf{0}\mid true\mid a\mid\lnot a\mid\varphi\vee\varphi
\mid\varphi\wedge\varphi\mid\bigcirc\varphi\mid\varphi\mathit{U}\varphi
\mid\square\varphi
\]
where $a\in AP.$ Inductively, we can prove that $\operatorname{Im}\left(
\left\Vert \varphi\right\Vert \right)  \subseteq\left\{  \mathbf{0,1}\right\}  $ for
every $\varphi\in bLTL\left(  K,AP\right)  $ and the semantics of the formulas
of $bLTL\left(  K,AP\right)  $ and the corresponding classical $LTL$ formulas
coincide. Let $\varphi,\psi\in LTL\left(  K,AP\right)  .$ We will call
$\varphi,\psi$ equivalent if $\left(  \left\Vert \varphi\right\Vert ,w\right)
=\left(  \left\Vert \psi\right\Vert ,w\right)  $ for every $w\in\left(
\mathcal{P}\left(  AP\right)  \right)  ^{\omega}.$

\begin{proposition}
\label{LTL equiv}For every $\varphi,\psi\in LTL\left(  K,AP\right)  $ the
following equivalences hold:

\begin{itemize}
\item $\psi\wedge\psi\equiv\psi,$ whenever $\psi$ is boolean

\item $\varphi\wedge\psi\equiv\psi\wedge\varphi,$ whenever $\varphi$ is boolean

\item $\varphi\wedge true\equiv\varphi$

\item $\bigcirc\left(  \varphi\wedge\psi\right)  \equiv\left(  \bigcirc
\varphi\right)  \wedge\left(  \bigcirc\psi\right)  $

\item $\bigcirc\left(  \varphi\vee\psi\right)  \equiv\left(  \bigcirc
\varphi\right)  \vee\left(  \bigcirc\psi\right)  $

\item $\bigcirc\left(  \varphi U\psi\right)  \equiv\left(  \bigcirc
\varphi\right)  U\left(  \bigcirc\psi\right)  $

\item $\bigcirc\left(  \square\varphi\right)  \equiv\square\left(
\bigcirc\varphi\right)  $

\item $\bigcirc k\equiv k$, for all $k \in K$
\end{itemize}
\end{proposition}

As in {\cite{Ma-Co}} we let an $LTL$\textit{-step formula }be an $LTL\left(
K,AP\right)  $-formula of the form $\underset{1\leq i\leq n}{%
{\displaystyle\bigvee}
}\left(  k_{i}\wedge\varphi_{i}\right)  $ with $k_{i}\in K$ and $\varphi
_{i}\in bLTL\left(  K,AP\right)  $ for every $1\leq i\leq n.$ We denote by
$stLTL\left(  K,AP\right)  $ the class of all $LTL$-step formulas over $AP$ and
$K.$ We introduce now the syntactic fragment of \textit{restricted}
$U$-\textit{nesting} $LTL$-formulas.

\begin{definition}
The fragment of \textit{restricted} $U$-\textit{nesting} $LTL$-formulas over
$AP$ and $K$, denoted by $RULTL\left(  K,AP\right)  $, is the least class of
formulas in $LTL\left(  K,AP\right)  $ which is defined inductively in the
following way.

\begin{itemize}
\item[$\cdot$] $k\in RULTL\left(  K,AP\right)  $ for every $k\in K.$

\item[$\cdot$] $bLTL\left(  K,AP\right)  \subseteq RULTL\left(  K,AP\right)
.$

\item[$\cdot$] If $\varphi\in RULTL\left(  K,AP\right)  ,$ then $\bigcirc
\varphi\in RULTL\left(  K,AP\right)  .$\

\item[$\cdot$] If $\varphi,\psi\in RULTL\left(  K,AP\right)  ,$ then
$\varphi\vee\psi\in RULTL\left(  K,AP\right)  .$

\item[$\cdot$] If $\varphi\in bLTL\left(  K,AP\right)  $ and $\psi\in
RULTL\left(  K,AP\right)  ,$

then $\varphi\wedge\psi,\psi\wedge\varphi\in RULTL\left(  K,AP\right)  .$

\item[$\cdot$] If $\varphi,\psi\in stLTL\left(  K,AP\right)  ,$ then $\varphi
U\psi\in RULTL\left(  K,AP\right)  .$

\item[$\cdot$] If $\varphi\in stLTL\left(  K,AP\right)  ,$ then $\square
\varphi\in RULTL\left(  K,AP\right)  .$
\end{itemize}
\end{definition}

\begin{remark}
Let $\varphi\in LTL\left(  K,AP\right)  .$ We will say \ that $\varphi$ is of
form A, if it is of the form $\varphi=\underset{1\leq i\leq n}{%
{\displaystyle\bigwedge}
}\varphi_{i}$, $n\geq1,$ where there exist at most one $i\in\left\{
1,\ldots,n\right\}  $ such that $\varphi_{i}\notin bLTL\left(  K,AP\right)  ,$
and for all $i\in\left\{  1,\ldots,n\right\}  ,$ $\varphi_{i}$ is not a finite
conjunction $\lambda_{1}\wedge\ldots\wedge\lambda_{k}$ with $k\geq2.$ We can
prove inductively in the structure of $RULTL\left(  K,AP\right)  $-formulas
that every $\varphi\in RULTL\left(  K,AP\right)  $ is of form A. For
$\varphi=k,$ $\varphi=a,$ $\varphi=\lnot a,\varphi=\bigcirc\varphi^{\prime
},\varphi=\varphi^{\prime}\vee\varphi^{\prime\prime},\varphi=\psi U\xi,$
$\varphi=\square\psi$ where $\varphi^{\prime},\varphi^{\prime\prime}\in
RULTL\left(  K,AP\right)  ,\xi, \psi\in stLTL\left(  K,AP\right)  ,$ we have
$\varphi=\underset{1\leq i\leq1}{%
{\displaystyle\bigwedge}
}\varphi_{i}$ with $\varphi_{1}=\varphi.$ Assume that $\varphi=\psi\wedge\xi$
where $\psi=\underset{1\leq i\leq n}{%
{\displaystyle\bigwedge}
}\psi_{i}\in RULTL\left(  K,AP\right)  ,$ and $\xi=\underset{1\leq j\leq m}{%
{\displaystyle\bigwedge}
}\xi_{j}\in bLTL\left(  K,AP\right)  $ are in form A. Then, $\varphi
=\underset{1\leq i\leq m+n}{%
{\displaystyle\bigwedge}
}\varphi_{i},$ where $\varphi_{i}=\psi_{i},1\leq i\leq n,$ and $\varphi
_{n+j}=\xi_{j},1\leq j\leq m,$ is also in form A, i.e., $\varphi_{i}$ is not a
conjunction $\lambda_{1}\wedge\ldots\wedge\lambda_{k}$ with $k\geq2$ for all
$i\in\left\{  1,\ldots,m+n\right\}  ,$ and since $\xi$ is boolean, there
exists at most an $i\in\left\{  1,\ldots,n+m\right\}  $ with $\varphi
_{i}\notin bLTL\left(  K,AP\right)  .$
\end{remark}

A formula $\varphi\in LTL\left(  K,AP\right)  $ is called \textit{reduced }if
(a) for every subformula of the form $\varphi_{1}\wedge\ldots\wedge\varphi
_{k}$ with $k\geq2$ it holds: $\varphi_{i}\neq true$ for every $1\leq i\leq
k,$ and $\varphi_{i}\neq\varphi_{j}$ whenever $\varphi_{i},\varphi_{j}$ with
$1\leq i<j\leq k$ are boolean and (b) no until operator is in the scope of any
next operator. For every $\varphi\in LTL\left(  K,AP\right)  $ we can
effectively construct an equivalent reduced formula by applying the
equivalences of Proposition \ref{LTL equiv}. We shall denote this formula by
$\varphi_{re}.$

In the sequel, we prove that for every reduced restricted $U$-nesting formula
$\varphi$ there exists an $\varepsilon$-wgBa accepting its semantics.We recall that the value
assigned by $\left\Vert \varphi\right\Vert $ to an infinite word $w$ is
computed by induction on the structure of $\varphi.$ Moreover, in the induction
for the semantics of the $\bigcirc,\square,$ and $U$ operators we compute the
values assigned by the semantics of subformulas of $\varphi$ on suffixes of
$w.$ It is our aim to define $\mathcal{A}_{\varphi}$ in a way that it
simulates the above induction. For this we define the states of the automaton as
sets of formulas, and every non-empty state will contain a maximal (according
to subformulas relation) formula. The weights of the transitions are defined so
that successful paths with non-empty states will simulate the inductive
computation of the semantics of the maximal formula of the first state of the
path. We consider as non-empty initial states of the automaton the ones with
maximal formula $\varphi.$
 We recall from \cite{Wo - Co} that if $\varphi\in LTL\left(  K,AP\right)  ,$ the
closure $cl\left(  \varphi\right)  $ of $\varphi$ is the smallest set $C$ such
that (a) $\varphi\in C,$ (b) if $\psi\wedge\xi\in C,$ or $\psi\vee\xi\in C,$
or $\psi U\xi\in C,$ then $\psi,\xi\in C,$ and (c) if $\bigcirc\psi\in C$ or
$\square\psi\in C,$ then $\psi\in C.$ In fact $cl\left(  \varphi\right)  $
contains $\varphi$ and all its subformulas.

\begin{definition}
\cite{Ma-Co}Let $\varphi\in LTL\left(  K,AP\right)  $. A subset $B$ of
$cl\left(  \varphi\right)  $ will be called $\varphi$-consistent if
$B=\emptyset$, or the following conditions hold.

\begin{itemize}
\item[$\cdot$] For every $a\in AP,$ $a\in B$ implies $\lnot a\notin B,$

\item[$\cdot$] $\varphi\in B,$

\item[$\cdot$] If $\psi\wedge\xi\in B,$ then $\psi,\xi\in B,$

\item[$\cdot$] If $\psi\vee\xi\in B$ or $\psi U\xi\in B,$ then $\psi\in B$ or
$\xi\in B,$

\item[$\cdot$] If $\square\psi\in B,$ then $\psi\in B.$
\end{itemize}
\end{definition}

\begin{example}
Let $%
\mathbb{R}
_{\min}=\left(
\mathbb{R}
_{+}\cup\left\{  \infty\right\}  ,\min,+,\infty,0\right)  $ be the tropical
semiring. It is well known that the tropical semiring is idempotent totally commutative
complete (see Chapter 5 in \cite{Ma-Co}). Let $AP=\left\{  a,b\right\}  ,$ and $\varphi=a\vee
b\in RULTL\left(
\mathbb{R}
_{\min},AP\right)  .$ Then, $\left\{
\emptyset,\left\{  a\vee b,b\right\}  ,\left\{  a\vee b,a\right\}  ,\left\{
a\vee b,a,b\right\}  \right\}  \subset\mathcal{P}\left(  cl\left(
\varphi\right)  \right)  $ is the set of all $\varphi$-consistent sets.
\end{example}

\begin{example}
Let $AP=\left\{  a,b,c\right\}  ,$ and $\varphi=\left(  a\wedge2\right)
\vee\left(  b\wedge3\right)  ,\psi=\varphi U\left(  \bigcirc c\right)  \in
LTL\left(
\mathbb{R}
_{\min},AP\right)  .$ Then, $B_{\varphi}=\left\{  \varphi,a\wedge
2,a,2\right\}  $ is a $\varphi$-consistent set, and $B_{\psi}=\left\{
\psi,\varphi,a\wedge2,a,2\right\}  =\left\{  \psi\right\}  \cup B_{\varphi}$
is a $\psi$-consistent set.
\end{example}

Let $\varphi\in LTL\left(  K,AP\right)  $ and $B\neq\emptyset$ be a $\varphi
$-consistent set. Let also $\varphi^{\prime}\in B,$ and $\Delta,\Gamma$ be
$\varphi^{\prime}$-consistent subsets of $B.$ Then, with standard arguments we
get that $\Delta\cup\Gamma$ is also a $\varphi^{\prime}$-consistent subset of
$B.$ This implies that for every $\varphi\in LTL\left(  K,AP\right)  ,$
$\varphi$-consistent set $B,$ and $\psi\in cl\left(  \varphi\right)  ,$ there
exists the greatest (according to subset relation) $\psi$-consistent subset of
$B.$ Keeping the notations of {\cite{Ma-We}}, we denote this set by
$M_{B,\psi}.$\ Clearly, if $\psi\in cl\left(  \varphi\right)  \backslash B,$
then $M_{B,\psi}=\emptyset.$ Moreover, we shall denote a $\varphi$-consistent
set $B$ by $B_{\varphi}.$

\begin{definition}
\cite{Ma-Co}Let $\varphi\in LTL\left(  K,AP\right)  $ and $B_{\varphi}$ be a
$\varphi$-consistent set. The finite set of formulas $next\left(  B_{\varphi
}\right)  \subseteq LTL\left(  K,AP\right)  $ is defined in the following way.
We set $next\left(  \emptyset\right)  =\left\{  \mathbf{0}\right\}  $ and for
$B_{\varphi}\neq\emptyset,$

\begin{itemize}
\item[$\cdot$] if $\varphi=a,$ or $\varphi=\lnot a,$ or $\varphi=k,a\in
AP,k\in K,$ then $next\left(  B_{\varphi}\right)  =\left\{  true\right\}  ,$

\item[$\cdot$] if $\varphi=\psi\wedge\xi,$ then

$next\left(  B_{\varphi}\right)  =\left\{  \psi^{\prime}\wedge\xi^{\prime}%
\mid\psi^{\prime}\in next\left(  M_{B_{\varphi},\psi}\right)  ,\xi^{\prime}\in
next\left(  M_{B_{\varphi},\xi}\right)  \right\}  ,$

\item[$\cdot$] if $\varphi=\psi\vee\xi,$ then $next\left(  B_{\varphi}\right)
=next\left(  M_{B_{\varphi},\psi}\right)  \cup next\left(  M_{B_{\varphi},\xi
}\right)  ,$

\item[$\cdot$] if $\varphi=\bigcirc\psi,$ then $next\left(  B_{\varphi
}\right)  =\left\{  \psi\right\}  ,$

\item[$\cdot$] if $\varphi=\psi U\xi,$ then \\$next\left(  B_{\varphi}\right)
=next\left(  M_{B_{\varphi},\xi}\right)  \cup\left\{  \varphi\wedge
\psi^{\prime}\mid\psi^{\prime}\in next\left(  M_{B_{\varphi},\psi}\right)
\right\}  ,$

\item[$\cdot$] if $\varphi=\square\psi,$ then $next\left(  B_{\varphi}\right)
=\left\{  \varphi\wedge\psi^{\prime}\mid\psi^{\prime}\in next\left(
M_{B_{\varphi},\psi}\right)  \right\}  .$
\end{itemize}
\end{definition}

The elements of $next\left(  B_{\varphi}\right)  $ will be called \textit{next
formulas }of\textit{ }$B_{\varphi}.$ Clearly, every formula in $next\left(
B_{\varphi}\right)  $ is a finite conjunction of the form $\underset{1\leq
i\leq k}{%
{\displaystyle\bigwedge}
}\psi_{i}$ where for every $1\leq i\leq k,\psi_{i}\in cl\left(  \varphi
\right)  $ or $\psi_{i}=\mathbf{0}$ or $\psi_{i}=true.$ Using induction on the
structure of $\varphi$ we can easily derive that for every $\varphi\in
bLTL\left(  K,AP\right)  $ (resp. $\varphi\in stLTL\left(  K,AP\right)  ,$
$\varphi\in RULTL\left(  K,AP\right)  $) and $\varphi$-consistent set
$B_{\varphi}$, $\ $it holds $next\left(  B_{\varphi}\right)  \subseteq
bLTL\left(  K,AP\right)  $ (resp. $next\left(  B_{\varphi}\right)  \subseteq
bLTL\left(  K,AP\right)  ,next\left(  B_{\varphi}\right)  \subseteq
RULTL\left(  K,AP\right)  $). We recall that the value assigned by $\left\Vert
\varphi\right\Vert $ to an infinite word $w$ is computed by induction on the
structure of $\varphi.$ Moreover, in the induction of the semantics of the
$\bigcirc,\square,$ and $U$ operators we compute the values assigned by the
semantics of subformulas of $\varphi$ on suffixes of $w.$ Next formulas of a
$\varphi$-consistent set indicate the formulas whose semantics should assign a
value to $w_{\geq1}$ so that $\left(  \left\Vert \varphi\right\Vert
,w\right)  $ can be effectively computed.

Next, we define inductively for every formula $\varphi\in LTL\left(
K,AP\right)  $ and every $\varphi$-consistent set $B_{\varphi},$ a mapping
$v_{B_{\varphi}}:next\left(  B_{\varphi}\right)  \rightarrow K$ assigning
values from $K$ to next formulas of $B_{\varphi}.$ We let $v_{\emptyset
}\left(  \mathbf{0}\right)  =\mathbf{0}.$ Now, assume that $B_{\varphi}%
\neq\emptyset$.\\
\begin{itemize}
\item For $\varphi=k\in K,$
$v_{\left\{  \varphi\right\}  }\left(  true\right)  =k.$

\item For $\varphi=a$ or $\varphi=\lnot a$ with $a\in AP,$ we set $v_{\left\{
\varphi\right\}  }\left(  true\right)  =\mathbf{1}.$

\item Let $\varphi=\psi\wedge\xi.$ We let $v_{B_{\varphi}}\left(  \psi^{\prime
}\wedge\xi^{\prime}\right)  =v_{M_{B_{\varphi},\psi}}\left(  \psi^{\prime
}\right)  \cdot v_{M_{B_{\varphi},\xi}}\left(  \xi^{\prime}\right)  $ where
$\psi^{\prime}\in next\left(  M_{B_{\varphi},\psi}\right)  ,\xi^{\prime}\in
next\left(  M_{B_{\varphi},\xi}\right)  .$

\item Next, let $\varphi=\psi\vee\xi.$ For every $\varphi^{\prime}\in next\left(
M_{B_{\varphi},\psi}\right)  \cup next\left(  M_{B_{\varphi},\xi}\right)  $ we
let $v_{B_{\varphi}}\left(  \varphi^{\prime}\right)  =v_{M_{B_{\varphi},\psi}%
}\left(  \varphi^{\prime}\right)  +v_{M_{B_{\varphi},\xi}}\left(
\varphi^{\prime}\right)  $ where abusing notations $v_{M_{B_{\varphi},\psi}%
}\left(  \varphi^{\prime}\right)  $ (resp. $v_{M_{B_{\varphi},\xi}}\left(
\varphi^{\prime}\right)  $) will stand for $\mathbf{0}$ whenever
$\varphi^{\prime}\notin next\left(  M_{B_{\varphi},\psi}\right) \\ (
resp.\text{ }\varphi^{\prime}\notin next\left(  M_{B_{\varphi},\xi}\right)
) .$

\item Assume that $\varphi=\bigcirc\psi.$ Then, for the unique element $\psi$ of
$next\left(  B_{\varphi}\right)  $ we set $v_{B_{\varphi}}\left(  \psi\right)
=\mathbf{1}.$

\item For $\varphi=\psi\mathit{U}\xi,$ we let $v_{B_{\varphi}}\left(
\varphi\wedge\psi^{\prime}\right)  =v_{M_{B_{\varphi},\psi}}\left(
\psi^{\prime}\right)  $ where $\psi^{\prime}\in next\left(  M_{B_{\varphi
},\psi}\right)  ,$ and $v_{B_{\varphi}}\left(  \xi^{\prime}\right)
=v_{M_{B_{\varphi},\xi}}\left(  \xi^{\prime}\right)  $ with $\xi^{\prime}\in
next\left(  M_{B_{\varphi},\xi}\right)  .$

\item For $\varphi=\square\psi,$ we set $v_{B_{\varphi}}\left(  \varphi\wedge
\psi^{\prime}\right)  =v_{M_{B_{\varphi},\psi}}\left(  \psi^{\prime}\right)  $
where $\psi^{\prime}\in next\left(  M_{B_{\varphi},\psi}\right)  .$
\end{itemize}

\begin{example} Let $AP=\left\{ a,b,c\right\} $ and $\varphi=\left(a\wedge2\right)\vee\left(b\wedge3\right)\in LTL\left(\mathbb{R}_{min},AP\right)$.
For $B_{\varphi}=\left\{ \varphi,a\wedge2,b\wedge3,a,2,b,3\right\} $
we have $next\left(B_{\varphi}\right)=\left\{ true\wedge true\right\} $
and $v_{B_{\varphi}}\left(true\wedge true\right)$$=0$$.$

 Now, let $\psi=\varphi U\left(\bigcirc c\right).$ For $B_{\psi}=\left\{ \psi,\varphi,a\wedge2,b\wedge3,a,2,b,3,\bigcirc c\right\} $
we have $next\left(B_{\psi}\right)$$=\left\{ \psi\wedge\left(true\wedge true\right),c\right\} .$ Thus,
$v_{B_{\psi}}\left(\psi\wedge\left(true\wedge true\right)\right)$
\emph{$=0$} and $v_{B_{\psi}}\left(c\right)=0.$ \end{example}

In the sequel, we use next formulas of a set $B_{\varphi}$ and the mapping
$v_{B_{\varphi}}$ to define the non-$\varepsilon$-transitions of the desired
automaton and their weights. More precisely, the states of the automaton will
be consistent sets of formulas. We allow non-$\varepsilon$-transitions with
weight$\neq\mathbf{0}$ only from a set $B_{\varphi}$ to a set $B_{\varphi
^{\prime}}$ with $\varphi^{\prime}\in next\left(  B_{\varphi}\right)  ,$ and
the weight of this transition will be equal to $v_{B_{\varphi}}\left(
\varphi^{\prime}\right)  .$ We will also use $\varepsilon$-transitions with
weight=$\mathbf{1}$ to move from a $\varphi$-consistent set to a $\varphi
_{re}$-consistent set, i.e., through the $\varepsilon$-transitions we will
reduce formulas by erasing from conjunctions multiple copies of identical
boolean formulas and the formula true. This reduction ensures that the state
set of the automaton is finite. It is also crucial that we reduce formulas
only with $\varepsilon$-transitions, since otherwise the reduction should
change the structure of a formula, and thus the computation of $v_{B_{\varphi
}}$ would not be well defined.

Let $\varphi ,\psi \in LTL\left( K,AP\right) $ and $B_{\varphi },B_{\psi },$
$\varphi $-consistent, $\psi $-consistent sets respectively. We say that $%
B_{\psi }$ is reachable by $B_{\varphi },$ if there exists a sequence $%
\varphi _{0},\ldots ,\varphi _{j}$ $\in LTL\left( K,AP\right) ,$ and $%
B_{\varphi _{0}},\ldots ,B_{\varphi _{j}},$ $\varphi _{0}$-consistent,\ldots
,$\varphi _{j}$-consistent set respectively, such that $\left( i\right) $ $%
\varphi _{0}=\varphi ,\ldots ,\varphi _{j}=\psi ,$ and $\left( ii\right) $
for every $0\leq l\leq j-1$ if $\varphi _{l}$ is reduced, then $\varphi
_{l+1}\in next\left( B_{\varphi _{l}}\right) ,$ otherwise $\varphi
_{l+1}=\left( \varphi _{l}\right) _{re}.$
 Observe that, since
$\varphi_{0}$ is reduced, the formulas $\varphi_{l}\left(  1\leq l\leq
j\right)  $ satisfy condition (b) in the definition of reduced
formulas. This implies that reduction, whenever it is applied, reduces only
conjunction. Let $reach\left(  B_{\varphi_{0}}\right)  $ contain all sets of formulas
reachable by $B_{\varphi_{0}}$. As the following remark shows $reach\left(  B_{\varphi_{0}}\right)$ is not finite in general.

\begin{remark}\cite{Ma-Co}
We let $AP=\left\{  a,b,c\right\}  ,$ $\varphi=\square\left(  \square\left(
a\wedge2\right)  \right)  \in LTL\left(
\mathbb{R}
_{\min},AP\right)  ,$ and $B_{\varphi}=\left\{  \varphi,\square\left(
a\wedge2\right)  ,a\wedge2,a,2\right\}  .$ Then, for every $j\geq1,$ every
consistent set of the formula $\varphi\wedge\left(  \underset{1\leq i\leq j}{%
{\displaystyle\bigwedge}
}\psi_{i}\right)  $ with $\psi_{i}=\square\left(  a\wedge2\right)  \left(
1\leq i\leq j\right)  $ belongs to the set $reach\left(  B_{\varphi}\right)
,$ and hence $reach\left(  B_{\varphi}\right)  $ is not finite.

Let now $\varphi=\left(  \left(  a\wedge2\right)  Uc\right)  Ud$ and $B_{\varphi}=\left\{  \varphi,\left(  a\wedge
2\right)  Uc,a\wedge2,a,2\right\}  .$ Then, for every $j\geq1,$ every
consistent set of the formula $\varphi\wedge\left(  \underset{1\leq i\leq j}{%
{\displaystyle\bigwedge}
}\psi_{i}\right)  $ with $\psi_{i}=\left(  a\wedge2\right)  Uc\left(  1\leq
i\leq j\right)  $ belongs to the set $reach\left(  B_{\varphi}\right)  ,$ and
hence $reach\left(  B_{\varphi}\right)  $ is not finite.
\end{remark}
However, the situation is different, if we consider formulas from $RULTL(K,AP)$.
\begin{lemma}
Let $\varphi\in RULTL\left(  K,AP\right)  $ be reduced and $B_{\varphi}$ be a
$\varphi$-consistent set. Then, $reach\left(  B_{\varphi}\right)  $ is finite and effectively computable.
\end{lemma}
The previous lemma is proved with the same arguments as in Lemma
94 in \cite{Ma-Co}.

\begin{definition}
Let $\varphi,\psi\in RULTL\left(  K,AP\right)  .$ For every $\pi\in
\mathcal{P}\left(  AP\right)  $ the triple $\left(  B_{\varphi},\pi,B_{\psi
}\right)  $ is called a \textit{next transition }if the following conditions hold.

\begin{itemize}
\item[$\cdot$] For every $a\in AP,$

$a\in B_{\varphi}$ implies $a\in\pi,$ and $\lnot a\in B_{\varphi}$ implies
$a\notin\pi,$

\item[$\cdot$] $\varphi$ is reduced and $\psi\in next\left(  B_{\varphi
}\right)  .$
\end{itemize}
\end{definition}

\begin{definition}
Let $\varphi\in RULTL\left(  K,AP\right)  .$ Then, for every $B_{\varphi
},B_{\varphi_{re}}$ with $B_{\varphi}\neq\emptyset$ and $B_{\varphi_{re}}%
\neq\emptyset$ the triple $\left(  B_{\varphi},\varepsilon,B_{\varphi_{re}%
}\right)  $ is called an $\varepsilon$-reduction transition.
\end{definition}

Sometimes in the sequel an $\varepsilon$-reduction transition will be called
for simplicity an $\varepsilon$-transition. Next, for every reduced
$RULTL\left(  K,AP\right)  $-formula $\varphi$ we construct an $\varepsilon
$-wgBa $\mathcal{A_{\varphi}}$ and show that $\varphi$ and
$\mathcal{A_{\varphi}}$ are expressively equivalent.

\begin{definition}
Let $\varphi\in RULTL\left(  K,AP\right)  $ be reduced. We define the
$\varepsilon$-wgBa $\mathcal{A_{\varphi}=}\left(  Q,wt,I,\mathcal{F}\right)  $
over $\mathcal{P}\left(  AP\right)  $ and $K$ as follows. We set

\begin{itemize}
\item $Q=\underset{B_{\varphi}}{\bigcup}\left(  \left\{  B_{\varphi}\right\}
\cup reach\left(  B_{\varphi}\right)  \right)  $,

\item $I=\left\{  B_{\varphi}\mid B_{\varphi}:\varphi\text{-consistent
set}\right\}  ,$

\item $wt\left(    B_{\psi},b,B_{\xi}  \right) \\ =%
\begin{cases}%
\begin{array}{ll}
v_{B_{\psi}}\left(  \xi\right)  &
\begin{array}{l}
\text{\text{if }}\left(  B_{\psi},b,B_{\xi}\right)  \text{\text{ \text{is a
next transition }}}\\
\text{\text{\text{and }}}B_{\xi}\neq\emptyset\text{ }%
\end{array}
\\
\mathbf{1} & \text{ if }\left(  B_{\psi},b,B_{\xi}\right)  \text{ }\text{is an
}\varepsilon\text{-reduction transition}\\
\mathbf{0} & \text{ otherwise }%
\end{array}
\end{cases}
$

for every $\left(  B_{\psi},b,B_{\xi}\right)  \in Q\times\left(
\mathcal{P}\left(  AP\right)  \mathcal{\cup\left\{  \varepsilon\right\}
}\right)  \times Q$, and

\item $\mathcal{F=}\left\{  F_{\varphi^{\prime}U\varphi^{\prime\prime}}%
\mid\varphi^{\prime}U\varphi^{\prime\prime}\in cl\left(  \varphi\right)
\right\}  $ $\text{where }$

$F_{\varphi^{\prime}U\varphi^{\prime\prime}}=\left\{  B_{\overline{\varphi}%
}\in Q\mid%
\begin{array}
[c]{c}%
B_{\overline{\varphi}}\neq\emptyset,\overline{\varphi}=\underset{1\leq i\leq
k}{\bigwedge}\varphi_{i}\text{ is of form A with }\\
\varphi_{i}\neq\varphi^{\prime}U\varphi^{\prime\prime},1\leq i\leq k
\end{array}
\right\}  $

for every $\varphi^{\prime}U\varphi^{\prime\prime}\in cl\left(  \varphi
\right)  .$
\end{itemize}
\end{definition}

Observe that for every $\varphi^{\prime}U\varphi^{\prime\prime}\in cl\left(
\varphi\right)  ,$ and every non-empty $B_{\psi},B_{\psi_{re}}\in Q,$ the
relation $B_{\psi}\in F_{\varphi^{\prime}U\varphi^{\prime\prime}}$ implies
that $B_{\psi_{re}}\in F_{\varphi^{\prime}U\varphi^{\prime\prime}},$ and
vice-versa. Thus, the $\varepsilon$-transitions of the automaton are well
defined. We note that if $\varphi$ contains no $U$ operators, then we have no
acceptance conditions, which means that all infinite paths that start with a
$\varphi$-consistent set are successful. Now, let $w\in\left(  \mathcal{P}%
\left(  AP\right)  \right)  ^{\omega}$ and $P_{w}=t_{0}t_{1}t_{2}\ldots$ be a
successful path of $\mathcal{A}_{\varphi}$ over $w.$ If there is an $i\geq0,$
such that $t_{i}$ is not a next transition or an $\varepsilon$-reduction
transition, then $weight_{\mathcal{A}_{\varphi}}\left(  P_{w}\right)
=\mathbf{0}$. We shall denote by $next_{\mathcal{A}_{\varphi}}\left(
w\right)  $ the set of all successful paths of $\mathcal{A}_{\varphi}$ over
$w$ composed of next and $\varepsilon$-transitions only. For the rest of this
section we let $pri_{\mathcal{A}_{\varphi}}\left(  w\right)  =\left\{
weight_{\mathcal{A}_{\varphi}}\left(  P_{w}\right)  \mid P_{w}\in
next_{\mathcal{A}_{\varphi}}\left(  w\right)  \right\}  .$

\begin{remark}
\label{Remark}Let $\varphi\in RULTL\left(  K,AP\right)  $ be reduced and
$\varphi_{1}U\varphi_{2}\in cl\left(  \varphi\right)  $ such that $\varphi
_{1}U\varphi_{2}$ does not appear in the scope of an always operator in
$\varphi.$ Then, for every $w\in\left(  \mathcal{P}\left(  AP\right)  \right)
^{\omega}$ and $P_{w}\in next_{\mathcal{A}_{\varphi}}\left(  w\right)  ,$
there is a state in $P_{w}$ such that all the subsequent states are in
$F_{\varphi_{1}U\varphi_{2}}.$ More precisely, for every $\varphi_{1}%
U\varphi_{2}\in cl\left(  \varphi\right)  ,$ if there is a next
transition in the path where a next formula of the maximal $\varphi
_{1}U\varphi_{2}$-consistent subset of the beginning state (of the transition) appears as a part
of the conjunction defining the maximal formula of the arriving state, then, after
a finite number of next transitions and since $P_{w}$ is successful, there is
a next transition where the next formula of the maximal $\varphi_{1}%
U\varphi_{2}$-consistent maximal subset of the beginning state of the transition is a next formula of the maximal $\varphi_{2}%
$-subset of the beginning state, i.e., a formula that is a conjunction not
containing $\varphi_{1}U\varphi_{2}.$ Since this holds for every appearance of
$\varphi_{1}U\varphi_{2},$ we conclude that there is a state in $P_{w}$ such
that all the subsequent states are in $F_{\varphi_{1}U\varphi_{2}}.$
\end{remark}

Next, we prove by induction on the structure of a reduced $RULTL\left(
K,AP\right)  $-formula $\varphi$ that $\mathcal{A}_{\varphi}$ accepts
$\left\Vert \varphi\right\Vert .$

\begin{lemma}
\label{Lemma atomic}Let $\varphi=a,\varphi=\lnot a,$ or $\varphi=k,$ with
$a\in AP,$ $k\in K.$ Then, $\left\Vert \mathcal{A}_{\varphi}\right\Vert
=\left\Vert \varphi\right\Vert .$
\end{lemma}

\begin{proof}
Let $\varphi=k\in K\backslash\left\{  \mathbf{0},\mathbf{1}\right\}  .$ Then,
the automaton $\mathcal{A}_{\varphi}=\left(  Q,wt,I,\mathcal{F}\right)  $ is
defined in the following way.

\begin{itemize}
\item $Q=\left\{  \emptyset,\left\{  k\right\}  ,\left\{  true\right\}
,\left\{  \mathbf{0}\right\}  \right\}  $

\item $wt\left(    q,b,q^{\prime} \right)  =\left\{
\begin{array}
[c]{ll}%
k & \text{\ if }q=\left\{  k\right\}  ,q^{\prime}=\left\{  true\right\}
,\text{ and }b\in\mathcal{P}\left(  AP\right) \\
\mathbf{1} &
\begin{array}
[c]{l}%
\text{if }q=q^{\prime}\text{, }q\neq\emptyset,\text{ and }b=\varepsilon,\text{
or}\\
\text{if }q=q^{\prime}=true,\text{ and }b\in\mathcal{P}\left(  AP\right)
\end{array}
\\
\mathbf{0} & \text{ \ otherwise}%
\end{array}
\right.  $

\item $I=\left\{  \left\{  k\right\}  ,\emptyset\right\}  $
\end{itemize}

The automaton contains no final subsets, i.e., $\mathcal{F=\emptyset}$. Let
$P_{w}\in next_{\mathcal{A}_{\varphi}}\left(  w\right)  $. Then, the next
transitions appearing in the path either form the sequence
\[
\left(  \left\{  k\right\}  ,w\left(  0\right)  ,\left\{  true\right\}
\right)  \left(  \left(  \left\{  true\right\}  ,w\left(  i\right)  ,\left\{
true\right\}  \right)  \right)  _{i\geq1},
\]
or the sequence%
\[
\left(  \emptyset,w\left(  0\right)  ,\left\{  \mathbf{0}\right\}  \right)
\left(  \left\{  \mathbf{0}\right\}  ,w\left(  1\right)  ,\left\{
true\right\}  \right)  \left(  \left(  \left\{  true\right\}  ,w\left(
i\right)  ,\left\{  true\right\}  \right)  \right)  _{i\geq2}.
\]
In the first case, using Property \ref{Property 3} we get $weight_{\mathcal{A}%
_{\varphi}}\left(  P_{w}\right)  =Val^{\omega}\left(  k,\mathbf{1}%
,\mathbf{1},\mathbf{1},\ldots\right)  \\=k,$ and in the latter case
$weight_{\mathcal{A}_{\varphi}}\left(  P_{w}\right)  =Val^{\omega}\left(
\mathbf{0},\mathbf{0},\mathbf{1},\mathbf{1},\mathbf{1},\ldots\right)
=\mathbf{0.}$ Hence, we get for every $w\in\left(  \mathcal{P}\left(
AP\right)  \right)  ^{\omega}$%
\[
\left(  \left\Vert \mathcal{A}_{\varphi}\right\Vert ,w\right)  =\underset
{k^{\prime}\in pri_{\mathcal{A}_{\varphi}}\left(  w\right)  }{%
{\displaystyle\sum}
}k^{\prime}=k=\left(  \left\Vert k\right\Vert ,w\right)
\]
as wanted. The lemma's claim for $\varphi=a,\varphi=\lnot a,$ $\varphi
=\mathbf{0,}$ $\varphi=\mathbf{1}$ can be proved with similar arguments.
\end{proof}

\begin{lemma}
\label{Lemma disjuction}Let $\psi,\xi\in RULTL\left(  K,AP\right)  $ and
$\varphi=\psi\vee\xi.$ If $\mathcal{A}_{\psi},\mathcal{A}_{\xi}$ recognize
$\left\Vert \psi\right\Vert ,\left\Vert \xi\right\Vert $ respectively, then
$\mathcal{A}_{\varphi}$ recognizes $\left\Vert \varphi\right\Vert .$
\end{lemma}

\begin{proof}
Let $w=\pi_{0}\pi_{1}\pi_{2}\ldots\in\left(  \mathcal{P}\left(  AP\right)
\right)  ^{\omega}$ and $\mathcal{A}_{\psi}=\left(  Q_{1},wt_{1}%
,I_{1},\mathcal{F}_{1}\right)  ,\mathcal{A}_{\xi}=\left(  Q_{2},wt_{2}%
,I_{2},\mathcal{F}_{2}\right)  $ and $\mathcal{A}_{\varphi}=\left(
Q,wt,I,\mathcal{F}\right)  $. First, we prove that
\[
\underset{k\in pri_{\mathcal{A}_{\varphi}}\left(  w\right)  }{%
{\displaystyle\sum}
}k\leq\underset{k\in pri_{\mathcal{A}_{\psi}}\left(  w\right)  }{%
{\displaystyle\sum}
}k+\underset{k\in pri_{\mathcal{A}_{\xi}}\left(  w\right)  }{%
{\displaystyle\sum}
}k.
\]
To this end we show that for every path $P_{w}\in next_{\mathcal{A}_{\varphi}%
}\left(  w\right)  $ there exist paths $P_{w}^{1}\in next_{\mathcal{A}_{\psi}%
}\left(  w\right)  ,P_{w}^{2}\in next_{\mathcal{A}_{\xi}}\left(  w\right)  $
with
\[
weight_{\mathcal{A}_{\varphi}}\left(  P_{w}\right)  =weight_{\mathcal{A}%
_{\psi}}\left(  P_{w}^{1}\right)  +weight_{\mathcal{A}_{\xi}}\left(  P_{w}%
^{2}\right)  .
\]
We let
\[
P_{w}:B_{\varphi}\overset{\ast}{\rightarrow}B_{\varphi}^{\prime}\overset
{\pi_{0}}{\rightarrow}B_{\varphi^{1}}\overset{\ast}{\rightarrow}%
B_{\varphi_{re}^{1}}\overset{\pi_{1}}{\rightarrow}B_{\varphi_{2}}\overset
{\ast}{\rightarrow}B_{\varphi_{re}^{2}}\ldots
\]
be a path in $next_{\mathcal{A}_{\varphi}}\left(  w\right)  $ with
$weight_{\mathcal{A}_{\varphi}}\left(  P_{w}\right)  \neq\mathbf{0}.$ This
implies that
\[
wt\left(    B_{\varphi}^{\prime},\pi_{0},B_{\varphi^{1}}
\right)  \neq\mathbf{0}\text{ and }wt\left(    B_{\varphi_{re}^{i}}%
,\pi_{i},B_{\varphi^{i+1}}  \right)  \neq\mathbf{0}%
\]
for every $i\geqslant1$. Then, by definition\\\\$\left(  a\right)  $ $\varphi^{1}\in next\left(  M_{B_{\varphi}^{\prime},\psi
}\right)  \backslash next\left(  M_{B_{\varphi}^{\prime},\xi}\right)  $ and
$wt\left(   B_{\varphi}^{\prime},\pi_{0},B_{\varphi^{1}}
\right)  =v_{M_{B_{\varphi}^{\prime},\psi}}\left(  \varphi^{1}\right)  ,$\\ \\or\\\\$\left(  b\right)  $ $\varphi^{1}\in next\left(  M_{B_{\varphi}^{\prime},\xi
}\right)  \backslash next\left(  M_{B_{\varphi}^{\prime},\psi}\right)  $ and
$wt\left(    B_{\varphi}^{\prime},\pi_{0},B_{\varphi^{1}}
\right)  =v_{M_{B_{\varphi}^{\prime},\xi}}\left(  \varphi^{1}\right)  ,$ \\\\or\\\\$\left(  c\right)  $ $\varphi^{1}\in next\left(  M_{B_{\varphi}^{\prime},\psi
}\right)
{\displaystyle\bigcap}
  next\left(  M_{B_{\varphi}^{\prime},\xi}\right)  $ and
\[
wt\left(    B_{\varphi}^{\prime},\pi_{0},B_{\varphi^{1}}
\right)  =v_{M_{B_{\varphi}^{\prime},\psi}\left(  \varphi^{1}\right)
}+v_{M_{B_{\varphi}^{\prime},\xi}\left(  \varphi^{1}\right)  }.
\]
If $\left(  a\right)  $ holds, then the path
\[
P_{w}^{1}:M_{B_{\varphi}^{\prime},\psi}\overset{\pi_{0}}{\rightarrow
}B_{\varphi^{1}}\overset{\ast}{\rightarrow}B_{\varphi_{re}^{1}}\overset
{\pi_{1}}{\rightarrow}B_{\varphi^{2}}\overset{\ast}{\rightarrow}%
B_{\varphi_{re}^{2}}\ldots
\]
of $\mathcal{A}_{\psi}$ over $w$ is successful, hence $P_{w}^{1}\in
next_{\mathcal{A}_{\psi}}\left(  w\right)  $ and $weight_{\mathcal{A}%
_{\varphi}}\left(  P_{w}\right)  =weight_{\mathcal{A}_{\psi}}\left(  P_{w}%
^{1}\right)  $, which implies that
\[
weight_{\mathcal{A}_{\varphi}}\left(  P_{w}\right)  =weight_{\mathcal{A}%
_{\psi}}\left(  P_{w}^{1}\right)  +weight_{\mathcal{A}_{\xi}}\left(  P_{w}%
^{2}\right)
\]
for
\[
P_{w}^{2}:\left(  \emptyset,\pi_{0},\left\{  \mathbf{0}\right\}  \right)
\left(  \left\{  \mathbf{0}\right\}  ,\pi_{1},\left\{  true\right\}  \right)
\left(  \left(  \left\{  true\right\}  ,\pi_{i},\left\{  true\right\}
\right)  \right)  _{i\geq0}.
\]
If case $\left(  b\right)  $ holds, then the path
\[
P_{w}^{2}:M_{B_{\varphi}^{\prime},\xi}\overset{\pi_{0}}{\rightarrow}%
B_{\varphi^{1}}\overset{\ast}{\rightarrow}B_{\varphi_{re}^{1}}\overset{\pi
_{1}}{\rightarrow}B_{\varphi^{2}}\overset{\ast}{\rightarrow}B_{\varphi
_{re}^{2}}\ldots.
\]
of $\mathcal{A}_{\xi}$ over $w$ is successful, i.e., $P_{w}^{2}\in
next_{\mathcal{A}_{\xi}}\left(  w\right)  ,$ and
\[
weight_{\mathcal{A}_{\varphi}}\left(  P_{w}\right)  =weight_{\mathcal{A}%
_{\psi}}\left(  P_{w}^{1}\right)  +weight_{\mathcal{A}_{\xi}}\left(  P_{w}%
^{2}\right)
\]
for
\[
P_{w}^{1}:\left(  \emptyset,\pi_{0},\left\{  \mathbf{0}\right\}  \right)
\left(  \left\{  \mathbf{0}\right\}  ,\pi_{1},\left\{  true\right\}  \right)
\left(  \left(  \left\{  true\right\}  ,\pi_{i},\left\{  true\right\}
\right)  \right)  _{i\geq0}.
\]
If case $\left(  c\right)  $ holds, then
\[
v_{M_{B_{\varphi}^{\prime},\psi}}\left(  \varphi^{1}\right)  +v_{M_{B_{\varphi
}^{\prime},\xi}}\left(  \varphi^{1}\right)  =v_{B_{\varphi}^{\prime}}\left(
\varphi^{1}\right)
\]
and for the paths $P_{w}^{1},P_{w}^{2}$ of $\mathcal{A}_{\psi},\mathcal{A}%
_{\xi}$ respectively, defined as in cases $\left(  a\right)  $ and $\left(
b\right)  $ respectively, we get
\[
weight_{\mathcal{A}_{\varphi}}\left(  P_{w}\right)  =weight_{\mathcal{A}%
_{\psi}}\left(  P_{w}^{1}\right)  +weight_{\mathcal{A}_{\xi}}\left(  P_{w}%
^{2}\right)  .
\]
More precisely, for $k_{1}=v_{M_{B_{\varphi}^{\prime},\psi}}\left(
\varphi^{1}\right)  ,k_{2}=v_{M_{B_{\varphi}^{\prime},\xi}}\left(  \varphi
^{1}\right)  $ it holds
\begin{align*}
&  weight_{\mathcal{A}_{\varphi}}\left(  P_{w}\right) \\
&  =Val^{\omega}\left(  wt\left(   B_{\varphi}^{\prime},\pi
_{0},B_{\varphi^{1}}\right)    ,wt\left(    B_{\varphi_{re}^{1}%
},\pi_{1},B_{\varphi^{2}}\right)   ,wt\left(    B_{\varphi
_{re}^{2}},\pi_{2},B_{\varphi^{3}}\right)    ,\ldots\right) \\
&  =Val^{\omega}\left(  k_{1}+k_{2},wt\left(  B_{\varphi_{re}^{1}}%
,\pi_{1},B_{\varphi^{2}}\right)    ,wt\left(    B_{\varphi
_{re}^{2}},\pi_{2},B_{\varphi^{3}}  \right)  ,\ldots\right) \\
&  =Val^{\omega}\left(  k_{1},wt\left(    B_{\varphi_{re}^{1}},\pi
_{1},B_{\varphi^{2}}\right)    ,wt\left(    B_{\varphi_{re}^{2}%
},\pi_{2},B_{\varphi^{3}}\right)    ,\ldots\right) \\
&  +Val^{\omega}\left(  k_{2},wt\left(   B_{\varphi_{re}^{1}},\pi
_{1},B_{\varphi^{2}}\right)    ,wt\left(    B_{\varphi_{re}^{2}%
},\pi_{2},B_{\varphi^{3}}\right)    ,\ldots\right) \\
&  =weight_{\mathcal{A}_{\psi}}\left(  P_{w}\right)  +weight_{\mathcal{A}%
_{\xi}}\left(  P_{w}\right)  ,
\end{align*}
where the third equality holds by the distributivity of $Val^{\omega}$ over finite sums.

We have thus shown that for every $k\in pri_{\mathcal{A}_{\varphi}}\left(
w\right)  $ there exist $k^{\prime}\in pri_{\mathcal{A}_{\psi}}\left(
w\right)  +pri_{\mathcal{A}_{\xi}}\left(  w\right)  $ such that $k\leq
k^{\prime}$ and by Lemmas \ref{Properties Sum}iv, \ref{Lemma inequality sum}%
\ we get
\begin{align}
\underset{k\in pri_{\mathcal{A}_{\varphi}}\left(  w\right)  }{%
{\displaystyle\sum}
}k  &  \leq\underset{k\in pri_{\mathcal{A}_{\psi}}\left(  w\right)
+pri_{\mathcal{A}_{\xi}}\left(  w\right)  }{%
{\displaystyle\sum}
}k\label{inequality}\\
&  =\underset{k\in pri_{\mathcal{A}_{\psi}}\left(  w\right)  }{%
{\displaystyle\sum}
}k+\underset{k\in pri_{\mathcal{A}_{\xi}}\left(  w\right)  }{%
{\displaystyle\sum}
}k.\nonumber
\end{align}

We now show that
\[
\underset{k\in pri_{\mathcal{A}_{\psi}}\left(  w\right)  }{%
{\displaystyle\sum}
}k+\underset{k\in pri_{\mathcal{A}_{\xi}}\left(  w\right)  }{%
{\displaystyle\sum}
}k\leq\underset{k\in pri_{\mathcal{A}_{\varphi}}\left(  w\right)  }{%
{\displaystyle\sum}
}k.
\]
Assume that
\[
P_{w}^{1}:B_{\psi}\overset{\ast}{\rightarrow}B_{\psi}^{\prime}\overset{\pi
_{0}}{\rightarrow}B_{\psi^{1}}\overset{\ast}{\rightarrow}B_{\psi_{re}^{1}%
}\overset{\pi_{1}}{\rightarrow}B_{\psi^{2}}\overset{\ast}{\rightarrow}%
B_{\psi_{re}^{2}}\ldots
\]
is a path in $next_{\mathcal{A}_{\psi}}\left(  w\right)  $ with
$weight_{\mathcal{A}_{\psi}}\left(  P_{w}^{1}\right)  \neq\mathbf{0.}$ We set
$B_{\varphi}^{\prime}=B_{\psi}^{\prime}\cup\left\{  \psi\vee\xi\right\}  .$
Then, the path
\[
P_{w}:B_{\varphi}^{\prime}\overset{\pi_{0}}{\rightarrow}B_{\psi^{1}}%
\overset{\ast}{\rightarrow}B_{\psi_{re}^{1}}\overset{\pi_{1}}{\rightarrow
}B_{\psi^{2}}\overset{\ast}{\rightarrow}B_{\psi_{re}^{2}}\ldots
\]
is a path of $\mathcal{A}_{\varphi}$ over $w$ in $next_{\mathcal{A}_{\varphi}%
}\left(  w\right)  $ and we claim that $weight_{\mathcal{A}_{\varphi}}\left(
P_{w}\right)  \geq weight_{\mathcal{A}_{\psi}}\left(  P_{w}^{1}\right)  .$ It
suffices to prove that%
\[
wt\left(    B_{\varphi}^{\prime},\pi_{0},B_{\psi^{1}}  \right)
\geq wt_{1}\left(    B_{\psi}^{\prime},\pi_{0},B_{\psi^{1}}
\right)  ,
\]
then our claim is derived by Lemma \ref{Valuation inequality}. If $\psi
^{1}\notin next\left(  M_{B_{\varphi}^{\prime},\xi}\right)  ,$ then the
equality holds by definition. Otherwise, $wt\left(   B_{\varphi
}^{\prime},\pi_{0},B_{\psi^{1}}  \right)  =v_{B_{\varphi}^{\prime}%
}\left(  \psi^{1}\right)  =v_{M_{B_{\varphi}^{\prime},\psi}}\left(  \psi
^{1}\right)  +v_{M_{B_{\varphi}^{\prime},\xi}}\left(  \psi^{1}\right)
=v_{B_{\psi}^{\prime}}\left(  \psi^{1}\right)  +v_{M_{B_{\varphi}^{\prime}%
,\xi}}\left(  \psi^{1}\right)  =wt_{1}\left(   B_{\psi}^{\prime}%
,\pi_{0},B_{\psi^{1}}  \right)  +v_{M_{B_{\varphi}^{\prime},\xi}%
}\left(  \psi^{1}\right)  ,$ i.e., $wt\left(    B_{\varphi}^{\prime}%
,\pi_{0},B_{\psi^{1}}  \right)  \geq wt_{1}\left(   B_{\psi
}^{\prime},\pi_{0},B_{\psi^{1}}  \right)  $ as wanted. For $P_{w}%
^{1}\in next_{\mathcal{A}_{\psi}}\left(  w\right)  $ with $weight_{\mathcal{A}%
_{\psi}}\left(  P_{w}^{1}\right)  =\mathbf{0}$ it trivially holds
$weight_{\mathcal{A}_{\varphi}}\left(  P_{w}\right)  \geq weight_{\mathcal{A}%
_{\psi}}\left(  P_{w}^{1}\right)  $ for every $P_{w}\in next_{\mathcal{A}%
_{\varphi}}\left(  w\right)  .$

Similarly, for every path $P_{w}^{2}\in next_{\mathcal{A}_{\xi}}\left(
w\right)  $ there exist a path $P_{w}\in next_{\mathcal{A}_{\varphi}}\left(
P_{w}\right)  $ with $weight_{\mathcal{A}_{\varphi}}\left(  P_{w}\right)  \geq
weight_{\mathcal{A}_{\xi}}\left(  P_{w}^{2}\right)  .$ Hence, for every $k\in
pri_{\mathcal{A}_{\psi}}\left(  P_{w}\right)  $ (resp. $k\in pri_{\mathcal{A}%
_{\xi}}\left(  w\right)  $) there exists a $k^{\prime}\in pri_{\mathcal{A}%
_{\varphi}}\left(  w\right)  $ such that $k\leq k^{\prime}$. This implies that
$\underset{k\in pri_{\mathcal{A}_{\psi}}\left(  w\right)  }{%
{\displaystyle\sum}
}k\leq\underset{k\in pri_{\mathcal{A}_{\varphi}}\left(  w\right)  }{%
{\displaystyle\sum}
}k$ \ and $\underset{k\in pri_{\mathcal{A}_{\xi}}\left(  w\right)  }{%
{\displaystyle\sum}
}k\leq\underset{k\in pri_{\mathcal{A}_{\varphi}}\left(  w\right)  }{%
{\displaystyle\sum}
}k.$ Again by Lemma \ref{Lemma inequality sum} and idempotency we get,
\begin{equation}
\underset{k\in pri_{\mathcal{A}_{\psi}}\left(  w\right)  }{%
{\displaystyle\sum}
}k+\underset{k\in pri_{\mathcal{A}_{\xi}}\left(  w\right)  }{%
{\displaystyle\sum}
}k\leq\underset{k\in pri_{\mathcal{A}_{\varphi}}\left(  w\right)  }{%
{\displaystyle\sum}
}k \label{disjuction 1}%
\end{equation}
as wanted.

Therefore, for every $w\in\left(  \mathcal{P}\left(  AP\right)  \right)
^{\omega}$ we get
\begin{align*}
\left(  \left\Vert \mathcal{A}_{\varphi}\right\Vert ,w\right)   &
=\underset{k\in pri_{\mathcal{A}_{\varphi}}\left(  w\right)  }{%
{\displaystyle\sum}
}k\\
&  =\underset{k\in pri_{\mathcal{A}_{\psi}}\left(  w\right)  }{%
{\displaystyle\sum}
}k+\underset{k\in pri_{\mathcal{A}_{\xi}}\left(  w\right)  }{%
{\displaystyle\sum}
}k\\
&  =\left(  \left\Vert \mathcal{A}_{\psi}\right\Vert ,w\right)  +\left(
\left\Vert \mathcal{A}_{\xi}\right\Vert ,w\right) \\
&  =\left(  \left\Vert \psi\right\Vert ,w\right)  +\left(  \left\Vert
\xi\right\Vert ,w\right) \\
&  =\left(  \left\Vert \psi\vee\xi\right\Vert ,w\right)
\end{align*}
where the second equality holds by \ref{inequality} and \ref{disjuction 1}.
\end{proof}

\begin{lemma}
\label{Lemma next}Let $\psi\in RULTL\left(  K,AP\right)  $ and $\varphi
=\bigcirc\psi.$ If $\mathcal{A}_{\psi}$ recognizes $\left\Vert \psi\right\Vert
,$ then $\mathcal{A}_{\varphi}$ recognizes $\left\Vert \varphi\right\Vert .$
\end{lemma}

\begin{proof}
Let $\mathcal{A}_{\psi}=\left(  Q^{\prime},wt^{\prime},I^{\prime}%
,\mathcal{F}^{\prime}\right)  ,\mathcal{A}_{\varphi}=\left(
Q,wt,I,\mathcal{F}\right)  ,$ and $w=\pi_{0}\pi_{1}\ldots\in\left(
\mathcal{P}\left(  AP\right)  \right)  ^{\omega}.$ We show that
$pri_{\mathcal{A}_{\varphi}}\left(  w\right)  =pri_{\mathcal{A}_{\psi}}\left(
w_{\geq1}\right)  .$ It suffices to prove that for every $P_{w}\in
next_{\mathcal{A}_{\varphi}}\left(  w\right)  ,$ there exists a $P_{w_{\geq1}%
}^{\prime}\in next_{\mathcal{A}_{\psi}}\left(  w_{\geq1}\right)  $ with
$weight_{\mathcal{A}_{\varphi}}\left(  P_{w}\right)  $
$=weight_{\mathcal{A_{\psi}}}\left(  P_{w_{\geq1}}^{\prime}\right)  $ and
vice-versa. For the straight implication, our claim clearly holds if the empty
state appears in $P_{w}.$ \ Now, let $P_{w}\in next_{\mathcal{A}_{\varphi}%
}\left(  w\right)  $ be a path with non-empty states that starts with a next
transition, i.e., it is of the form
\[
P_{w}:B_{\varphi}\overset{\pi_{0}}{\rightarrow}B_{\varphi^{1}}\overset{\ast
}{\rightarrow}B_{\varphi_{re}^{1}}\overset{\pi_{1}}{\rightarrow}B_{\varphi
^{2}}\overset{\ast}{\rightarrow}B_{\varphi_{re}^{2}}\ldots
\]
We have $\varphi^{1}=\psi$, and $wt\left(    B_{\varphi},\pi
_{0},B_{\varphi^{1}}  \right)  =v_{B_{\varphi}}\left(  \psi\right) =\mathbf{1},$ where the last equality holds by the definition of $v_{B_{\varphi}}.$ Then, the sequence
\[
P_{w_{\geq1}}^{\prime}:B_{\varphi_{re}^{1}}\overset{\pi_{1}}{\rightarrow
}B_{\varphi^{2}}\overset{\ast}{\rightarrow}B_{\varphi_{re}^{2}}\ldots
\]
is a path of $\mathcal{A}_{\psi}$ over $w_{\geq1}$ with $P_{w_{\geq1}}%
^{\prime}\in next_{\mathcal{A}_{\psi}}\left(  w_{\geq1}\right)  ,$ and
\begin{align*}
weight_{\mathcal{A}_{\varphi}}\left(  P_{w}\right)   &  =Val^{\omega}\left(
\mathbf{1},wt\left(    B_{\varphi_{re}^{1}},\pi_{1},B_{\varphi^{2}%
}\right)    ,wt\left(    B_{\varphi_{re}^{2}},\pi_{2}%
,B_{\varphi^{3}}\right)    ,\ldots\right) \\
&  =Val^{\omega}\left(  wt^{\prime}\left(   B_{\varphi_{re}^{1}}%
,\pi_{1},B_{\varphi^{2}}  \right)  ,wt^{\prime}\left(
B_{\varphi_{re}^{2}},\pi_{2},B_{\varphi^{3}}  \right)  ,\ldots\right)
\\
&  =weight_{\mathcal{A}_{\psi}}\left(  P_{w_{\geq1}}^{\prime}\right).
\end{align*}
The second equality holds since
$wt\left(   B_{\varphi_{re}^{i}},\pi_{i},B_{\varphi^{i+1}}
\right)  =wt^{\prime}\left(    B_{\varphi_{re}^{i}},\pi_{i}%
,B_{\varphi^{i+1}}  \right)  $ \ for every $i\geq1,$ and by Property \ref{Property 1}.

Conversely, let $P_{w_{\geq1}}^{\prime}:B_{\psi}\overset{\pi_{1}}{\rightarrow
}B_{\psi^{2}}\overset{\ast}{\rightarrow}B_{\psi_{re}^{2}}\ldots$ be a path in
$next_{\mathcal{A}_{\psi}}\left(  w_{\geq1}\right)  $ with non-empty states.
Then, the sequence
\[
P_{w}:B_{\psi}\cup\left\{  \bigcirc\psi\right\}  \overset{\pi_{0}}%
{\rightarrow}B_{\psi}\overset{\pi_{1}}{\rightarrow}B_{\psi^{1}}\overset{\ast
}{\rightarrow}B_{\psi_{re}^{1}}\ldots
\]
is a path of $\mathcal{A}_{\varphi}$ over $w$ with $P_{w}\in next_{\mathcal{A}%
_{\varphi}}\left(  w\right)  ,$ and
\[
weight_{\mathcal{A}_{\varphi}}\left(  P_{w}\right)  =weight_{\mathcal{A}%
_{\psi}}\left(  P_{w_{\geq1}}^{\prime}\right)  .
\]
If the empty state occurs in $P_{w_{\geq1}}^{\prime}$ it is obvious.

Thus, for every $w\in\left(  \mathcal{P}\left(  AP\right)  \right)  ^{\omega}$
we have
\[
\left(  \left\Vert \mathcal{A}_{\varphi}\right\Vert ,w\right)  =\underset{k\in
pri_{\mathcal{A}_{\varphi}}\left(  w\right)  }{%
{\displaystyle\sum}
}k=\underset{k\in pri_{\mathcal{A}_{\psi}}\left(  w_{\geq 1}\right)  }{%
{\displaystyle\sum}
}k=\left(  \left\Vert \psi\right\Vert ,w_{\geq1}\right)  =\left(  \left\Vert
\bigcirc\psi\right\Vert ,w\right)  ,
\]
as wanted.
\end{proof}

For every $\varphi\in LTL\left(  K,AP\right)  $, and every $B_{\varphi},$ we
let $\widehat{next}\left(  B_{\varphi}\right)  $ be the subset of $next\left(
B_{\varphi}\right)  $ containing all $\underset{1\leq i\leq k}{%
{\displaystyle\bigwedge}
}\psi_{i}\in next\left(  B_{\varphi}\right)  $ of Form A where $\psi_{i}\neq\mathbf{0}$
for all $i\in\left\{  1,\ldots,k\right\}  .$ The subsequent three lemmas will
contribute to the proof of the remaining induction steps.

\begin{lemma}
\label{boithitiko lemma 1}Let $\varphi\in RULTL\left(  K,AP\right)  ,$ and
$B_{\varphi},B_{\varphi}^{\prime}\neq\emptyset$ be $\varphi$-consistent sets
with $B_{\varphi}\subseteq B_{\varphi}^{\prime}$. Then, $\widehat{next}\left(
B_{\varphi}\right)  \subseteq next\left(  B_{\varphi}^{\prime}\right)  $ and
for every $\psi\in\widehat{next}\left(  B_{\varphi}\right)  $ it holds
$v_{B_{\varphi}^{\prime}}\left(  \psi\right)  \geq v_{B_{\varphi}}\left(
\psi\right)  .$
\end{lemma}

\begin{proof}
For atomic propositions $a,\lnot a\in AP,$ and for $k\in K$ our claim is
obvious. Let $\varphi=\lambda\vee\xi.$ If $M_{B_{\varphi,\lambda}}%
\neq\emptyset$ and $M_{B_{\varphi},\xi}\neq\emptyset,$ then
\begin{align*}
\widehat{next}\left(  B_{\varphi}\right)   &  =\widehat{next}\left(
M_{B_{\varphi},\lambda}\right)  \cup\widehat{next}\left(  M_{B_{\varphi},\xi
}\right) \\
&  \subseteq next\left(  M_{B_{\varphi}^{\prime},\lambda}\right)  \cup
next\left(  M_{B_{\varphi}^{\prime},\xi}\right) \\
&  =next\left(  B_{\varphi}^{\prime}\right)
\end{align*}
where the inclusion holds by the induction hypothesis since $M_{B_{\varphi
},\lambda}\subseteq M_{B_{\varphi}^{\prime},\lambda}$\ and these two sets are $\lambda$-consistent, and similarly $M_{B_{\varphi}%
,\xi}\subseteq M_{B_{\varphi}^{\prime},\xi}$ and these two sets are $\xi$-consistent.

Moreover, for every $\psi\in\widehat{next}\left(  M_{B_{\varphi},\lambda
}\right)  \cap\widehat{next}\left(  M_{B_{\varphi},\xi}\right)  ,$ we have
\begin{align*}
v_{B_{\varphi}^{\prime}}\left(  \psi\right)   &  =v_{M_{B_{\varphi}^{\prime
},\lambda}}\left(  \psi\right)  +v_{M_{B_{\varphi}^{\prime},\xi}}\left(
\psi\right) \\
&  \geq v_{M_{B_{\varphi},\lambda}}\left(  \psi\right)  +v_{M_{B_{\varphi}%
,\xi}}\left(  \psi\right) \\
&  =v_{B_{\varphi}}\left(  \psi\right)  ,
\end{align*}
whereas for $\psi\in\widehat{next}\left(  M_{B_{\varphi},\lambda}\right)
\backslash\widehat{next}\left(  M_{B_{\varphi},\xi}\right)  $ we have
\[
v_{B_{\varphi}^{\prime}\left(  \psi\right)  }\geq v_{M_{B_{\varphi}^{\prime
},\lambda}}\left(  \psi\right)  \geq v_{M_{B_{\varphi},\lambda}}\left(
\psi\right)  =v_{B_{\varphi}}\left(  \psi\right)  .
\]
In the same way, for $\psi\in\widehat{next}\left(  M_{B_{\varphi},\xi}\right)
\diagdown\widehat{next}\left(  M_{B_{\varphi},\lambda}\right)  $ we get
$v_{B_{\varphi}^{\prime}}\left(  \psi\right)  \geq v_{B_{\varphi}}\left(
\psi\right)  .$ Now, if $M_{B_{\varphi,\lambda}}\neq\emptyset$ and
$M_{B_{\varphi},\xi}=\emptyset,$ then%

\[
\widehat{next}\left(  B_{\varphi}\right)  =\widehat{next}\left(
M_{B_{\varphi},\lambda}\right)  \subseteq next\left(  M_{B_{\varphi}^{\prime
},\lambda}\right)  \subseteq next\left(  B_{\varphi}^{\prime}\right)
\]
and for every $\psi\in\widehat{next}\left(  B_{\varphi}\right)  ,$ we have
\[
v_{B_{\varphi}^{\prime}\left(  \psi\right)  }\geq v_{M_{B_{\varphi}^{\prime
},\lambda}}\left(  \psi\right)  \geq v_{M_{B_{\varphi},\lambda}}\left(
\psi\right)  =v_{B_{\varphi}}\left(  \psi\right)  .
\]
The case $M_{B_{\varphi},\lambda}=\emptyset$ and $M_{B_{\varphi},\xi}%
\neq\emptyset$ is treated similarly.

Let $\varphi=\lambda\wedge\xi$ such that $\lambda\in RULTL\left(  K,AP\right)
$, $\xi\in bLTL\left(  K,AP\right)  $. Then,
\begin{align*}
\widehat{next}\left(  B_{\varphi}\right)   &  =\left\{  \lambda^{\prime}%
\wedge\xi^{\prime}\mid\lambda^{\prime}\in\widehat{next}\left(  M_{B_{\varphi
},\lambda}\right)  ,\xi^{\prime}\in\widehat{next}\left(  M_{B_{\varphi},\xi
}\right)  \right\} \\
&  \subseteq\left\{  \lambda^{\prime}\wedge\xi^{\prime}\mid\lambda^{\prime}\in
next\left(  M_{B_{\varphi}^{\prime},\lambda}\right)  ,\xi^{\prime}\in
next\left(  M_{B_{\varphi}^{\prime},\xi}\right)  \right\} \\
&  =next\left(  B_{\varphi}^{\prime}\right)
\end{align*}
where the inclusion again holds by the induction hypothesis since
$M_{B_{\varphi},\lambda}\subseteq M_{B_{\varphi}^{\prime},\lambda}$ and
$M_{B_{\varphi},\xi}\subseteq M_{B_{\varphi}^{\prime},\xi}.$ Moreover, for
every $\psi=\lambda^{\prime}\wedge\xi^{\prime}$ with $\lambda^{\prime}%
\in\widehat{next}\left(  M_{B_{\varphi},\lambda}\right)  ,\xi^{\prime}%
\in\widehat{next}\left(  M_{B_{\varphi},\xi}\right)  $ we have
\begin{align*}
v_{B_{\varphi}^{\prime}}\left(  \psi\right)   &  =v_{M_{B_{\varphi}^{\prime
},\lambda}}\left(  \lambda^{\prime}\right)  \cdot v_{M_{B_{\varphi}^{\prime
},\xi}}\left(  \xi^{\prime}\right) \\
&  \geq v_{M_{B_{\varphi},\lambda}}\left(  \lambda^{\prime}\right)  \cdot
v_{M_{B_{\varphi}},\xi}\left(  \xi^{\prime}\right) \\
&  =v_{B_{\varphi}}\left(  \psi\right)
\end{align*}
where the inequality holds by induction hypothesis and the fact that $\xi$ is boolean. More precisely, by induction hypothesis it holds $v_{M_{B_{\varphi}^{\prime
},\lambda}}\left(  \lambda^{\prime}\right)
\geq v_{M_{B_{\varphi},\lambda}}\left(  \lambda^{\prime}\right)$, and $v_{M_{B_{\varphi}^{\prime
},\xi}}\left(  \xi^{\prime}\right) \geq
v_{M_{B_{\varphi}},\xi}\left(  \xi^{\prime}\right)
 $, and since $\xi$ is boolean
$v_{M_{B^{\prime}_{\varphi},\xi}}\left(\xi^{\prime}\right)=\mathbf{0}$ implies $v_{M_{B_{\varphi},\xi}}\left(\xi^{\prime}\right)=\mathbf{0}$, and $v_{M_{B_{\varphi},\xi}}\left(\xi^{\prime}\right)=\mathbf{1}$ implies $v_{M_{B^{\prime}_{\varphi},\xi}}\left(\xi^{\prime}\right)=\mathbf{1}$, and thus we conclude the inequality.

Assume now that $\varphi=\bigcirc\xi$. Then, $\widehat{next}\left(  B_{\varphi
}\right)  =\left\{  \xi\right\}  $ or $\widehat{next}\left(  B_{\varphi
}\right)  =\emptyset.$ In both cases $\widehat{next}\left(  B_{\varphi
}\right)  \subseteq next\left(  B_{\varphi}^{\prime}\right)  =\left\{
\xi\right\}  .$ In addition if $next\left(  B_{\varphi}\right)  =\left\{
\xi\right\}  ,$ then $v_{B_{\varphi}^{\prime}}\left(  \xi\right)
=\mathbf{1=}v_{B_{\varphi}}\left(  \xi\right)  .$

Next, let $\varphi=\lambda U\xi$ where $\lambda,\xi\in stLTL\left(
K,AP\right)  $%
\begin{align*}
&  \widehat{next}\left(  B_{\varphi}\right) \\
&  =\left\{  \varphi\wedge\lambda^{\prime}\mid\lambda^{\prime}\in
\widehat{next}\left(  M_{B_{\varphi},\lambda}\right)  \right\} \\
&  \cup\left\{  \xi^{\prime}\mid\xi^{\prime}\in\widehat{next}\left(
M_{B_{\varphi},\xi}\right)  \right\} \\
&  \subseteq\left\{  \varphi\wedge\lambda^{\prime}\mid\lambda^{\prime}\in
next\left(  M_{B_{\varphi}^{\prime},\lambda}\right)  \right\} \\
&  \cup\left\{  \xi^{\prime}\mid\xi^{\prime}\in next\left(  M_{B_{\varphi
}^{\prime},\xi}\right)  \right\} \\
&  =next\left(  B_{\varphi}^{\prime}\right)  .
\end{align*}
For $\psi\in\widehat{next}\left(  B_{\varphi}\right)  $ with $\psi
=\varphi\wedge\lambda^{\prime}$ and $\lambda^{\prime}\in\widehat{next}\left(
M_{B_{\varphi},\lambda}\right)  ,$ we \ have $v_{B_{\varphi}^{\prime}}\left(
\psi\right)  =v_{M_{B_{\varphi}^{\prime},\lambda}}\left(  \lambda^{\prime
}\right)  \geq v_{M_{B_{\varphi},\lambda}}\left(  \lambda^{\prime}\right)
=v_{B_{\varphi}}\left(  \psi\right)  .$ For $\psi\in\widehat{next}\left(
B_{\varphi}\right)  $ with $\psi=\xi^{\prime}\in\widehat{next}\left(
M_{B_{\varphi},\xi}\right)  $ we get $v_{B_{\varphi}^{\prime}}\left(
\psi\right)  =v_{M_{B_{\varphi}^{\prime},\xi}}\left(  \xi^{\prime}\right)
\geq v_{M_{B_{\varphi},\xi}}\left(  \xi^{\prime}\right)  =v_{B_{\varphi}%
}\left(  \psi\right)  .$

Finally, if $\varphi=\square\xi$ where $\xi\in stLTL\left(  K,AP\right)  ,$
then
\begin{align*}
\widehat{next}\left(  B_{\varphi}\right)   &  =\left\{  \varphi\wedge
\xi^{\prime}\mid\xi^{\prime}\in\widehat{next}\left(  M_{B_{\varphi},\xi
}\right)  \right\} \\
&  \subseteq\left\{  \varphi\wedge\xi^{\prime}\mid\xi^{\prime}\in next\left(
M_{B_{\varphi}^{\prime},\xi}\right)  \right\} \\
&  =next\left(  B_{\varphi}^{\prime}\right)  ,
\end{align*}
and for every $\psi=\varphi\wedge\xi^{\prime},$ with $\xi^{\prime}\in
\widehat{next}\left(  M_{B_{\varphi},\xi}\right)  ,$ we have $v_{B_{\varphi
}^{\prime}}\left(  \psi\right)  =v_{M_{B_{\varphi}^{\prime},\xi}}\left(
\xi^{\prime}\right)  \geq v_{M_{B_{\varphi},\xi}}\left(  \xi^{\prime}\right)
=v_{B_{\varphi}}\left(  \psi\right)  .$
\end{proof}
\begin{lemma}
\label{boithitiko lemma 2}Let $\psi\in bLTL\left(  K,AP\right)  $ and
$\xi\in RULTL\left(  K,AP\right)  .$ If $\psi,\xi$ are reduced and
$\varphi=\left(  \psi\wedge\xi\right)  _{re}$, then for every $B_{\varphi}%
\neq\emptyset$ and $\varphi^{\prime}\in next\left(  B_{\varphi}\right)  $ there exist
a $\psi$-consistent set $B_{\psi}$ and a $\xi$-consistent set $B_{\xi
}$ such that for some $\psi^{\prime}\in next\left(  B_{\psi}\right)  $,
$\xi^{\prime}\in next\left(  B_{\xi}\right)  $ it holds $\varphi_{re}^{\prime
}=\left(  \psi_{re}^{\prime}\wedge\xi_{re}^{\prime}\right)  _{re}$ and
$v_{B_{\varphi}}\left(  \varphi^{\prime}\right)  =v_{B_{\psi}}\left(
\psi^{\prime}\right)  \cdot v_{B_{\xi}}\left(  \xi^{\prime}\right)  .$
\end{lemma}

\begin{proof}
Assume first that $\psi\neq true$ and $\xi\neq true.$ We point out the
following cases.

\begin{itemize}
\item[(a)] $\psi=\underset{1\leq i\leq m_{1}}{%
{\displaystyle\bigwedge}
}\psi_{i},\xi=\underset{1\leq j\leq m_{2}}{%
{\displaystyle\bigwedge}
}\xi_{j}$ and there exist $i_{1},\ldots,i_{h}\in\left\{  1,\ldots
,m_{1}\right\}  $ and $j_{1},\ldots,j_{h}\in\left\{  1,\ldots,m_{2}\right\}  $
such that $\psi_{i_{1}}=\xi_{j_{1}},\ldots,\psi_{i_{h}}=\xi_{j_{h}}.$
Then,
\[
\varphi=\left(  \underset{1\leq i\leq m_{1}}{%
{\displaystyle\bigwedge}
}\psi_{i}\right)  \wedge\left(  \underset{j\neq j_{1},\ldots,j\neq j_{h}%
}{\underset{1\leq j\leq m_{2}}{%
{\displaystyle\bigwedge}
}}\xi_{j}\right)
\]

and%

\[
\varphi^{\prime}=\left(  \underset{1\leq i\leq m_{1}}{%
{\displaystyle\bigwedge}
}\psi_{i}^{\prime}\right)  \wedge\left(  \underset{j\neq j_{1},\ldots,j\neq
j_{h}}{\underset{1\leq j\leq m_{2}}{%
{\displaystyle\bigwedge}
}}\xi_{j}^{\prime}\right)
\]
where $\psi_{i}^{\prime}\in next\left(  M_{B_{\varphi},\psi_{i}}\right)
,i\in\left\{  1,\ldots,m_{1}\right\}  ,$ and $\xi_{j}^{\prime}\in next\left(
M_{B_{\varphi},\xi_{j}}\right)  ,j\in\left\{  1,\ldots,m_{2}\right\}
\backslash\left\{  j_{1},\ldots,j_{h}\right\}  .$ We let
\[
B_{\psi}=\left\{  \psi\right\}  \cup\left(  \underset{1\leq i\leq m_{1}%
}{%
{\displaystyle\bigcup}
}M_{B_{\varphi},\psi_{i}}\right)
\]
and
\[
B_{\xi}=\left\{  \xi\right\}  \cup\left(  \underset{1\leq j\leq m_{2}}{%
{\displaystyle\bigcup}
}M_{B_{\varphi},\xi_{j}}\right)  .
\]

Then, $\psi^{\prime}=\underset{1\leq i\leq m_{1}}{%
{\displaystyle\bigwedge}
}\psi_{i}^{\prime}\in next\left(  B_{\psi}\right)  $ and $\xi^{\prime
}=\xi^{\prime\prime}\wedge\left(  \psi_{i_{1}}^{\prime}\wedge\ldots
\wedge\psi_{i_{h}}^{\prime}\right)  \in next\left(  B_{\xi}\right)  $
where $\xi^{\prime\prime}=\underset{j\neq j_{1},\ldots,j\neq j_{h}}%
{\underset{1\leq j\leq m_{2}}{%
{\displaystyle\bigwedge}
}}\xi_{j}^{\prime}.$ It follows that $\psi_{re}^{\prime}=\left(  \psi
_{re}^{\prime}\wedge\xi_{re}^{\prime}\right)  _{re},$ and since
$M_{B_{\psi},\psi_{i}}=M_{B_{\varphi},\psi_{i}},M_{B_{\xi},\xi_{j}%
}=M_{B_{\varphi},\xi_{j}}\left(  1\leq i\leq m_{1},1\leq j\leq m_{2}\right)  $ we
get that
\begin{align*}
v_{B_{\varphi}}\left(  \varphi^{\prime}\right)   &  =\underset{1\leq i\leq m_{1}}{%
{\displaystyle\prod}
}v_{M_{B_{\varphi}},\psi_{i}}\left(  \psi_{i}^{\prime}\right)
\cdot\underset{j\neq j_{1},\ldots,j\neq j_{h}}{\underset{1\leq j\leq m_{2}}{%
{\displaystyle\prod}
}}v_{M_{B_{\varphi}},\xi_{j}}\left(  \xi_{j}^{\prime}\right) \\
&  =\underset{1\leq i\leq m_{1}}{%
{\displaystyle\prod}
}v_{M_{B_{\varphi}},\psi_{i}}\left(  \psi_{i}^{\prime}\right)
\cdot\underset{j\neq j_{1},\ldots,j\neq j_{h}}{\underset{1\leq j\leq m_{2}}{%
{\displaystyle\prod}
}}v_{M_{B_{\xi}},\xi_{j}}\left(  \xi_{j}^{\prime}\right) \\
&  \cdot\underset{1\leq k\leq h}{%
{\displaystyle\prod}
}v_{M_{B_{\xi},\psi_{i_{k}}}}\left(  \psi_{i_{k}}^{\prime}\right) \\
&  =v_{B_{\psi}}\left(  \psi^{\prime}\right)  \cdot v_{B_{\xi}}\left(
\xi^{\prime}\right)  .
\end{align*}

\item[(b)] $\left(  \psi\wedge\xi\right)  _{re}=\psi\wedge\xi,$ and
our claim follows by definition.
\end{itemize}

Now, let $\psi=true,$ and $\xi\neq true.$ Then, $\varphi=\left(  \psi
\wedge\xi\right)  _{re}=\xi$, and for $B_{\xi}=B_{\varphi},B_{\psi}=\left\{
true\right\}  ,\xi^{\prime}=\varphi^{\prime},$ and $\psi^{\prime}=true$ our
claim obviously holds. For the case where $\psi\neq true,$ and $\xi=true,$
and the case $\psi=\xi=true,$ we act similarly. \
\end{proof}

\begin{lemma}
\label{boithitiko lemma 3}Let $\psi\in bLTL\left(  K,AP\right)  $ and
$\xi\in RULTL\left(  K,AP\right)  $ be reduced and $\pi\in\mathcal{P}\left(
AP\right)  .$ If $\left(  B_{\psi},\pi,B_{\psi^{\prime}}\right)
,\left(  B_{\xi},\pi,B_{\xi^{\prime}}\right)  $ are next transitions with
$\psi^{\prime}\in\widehat{next}\left(  B_{\psi}\right)  ,$ $\xi
^{\prime}\in\widehat{next}\left(  B_{\xi}\right)  $, and $v_{B_{\psi}%
}\left(  \psi^{\prime}\right)  \neq\mathbf{0},v_{B_{\xi}}\left(
\xi^{\prime}\right)  \neq\mathbf{0},$ then for $\varphi=\left(  \psi\wedge
\xi\right)  _{re}$ there exist $B_{\varphi}\neq\emptyset,\varphi^{\prime}\in
\widehat{next}\left(  B_{\varphi}\right)  ,$ and $\psi^{\prime\prime}\in
RULTL\left(  K,AP\right)  $ such that

\begin{itemize}
\item[(i)] $\left(  B_{\varphi},\pi,B_{\varphi^{\prime}}\right)  $ is a next
transition for every $B_{\varphi^{\prime}}$ and $v_{B_{\varphi}}\left(  \varphi^{\prime
}\right)  \geq v_{B_{\psi}}\left(  \psi^{\prime}\right)  \cdot
v_{B_{\xi}}\left(  \xi^{\prime}\right)  ,$

\item[(ii)] $\varphi_{re}^{\prime}=\left(  \psi_{re}^{\prime}\wedge\psi
_{re}^{\prime\prime}\right)  _{re},$ and for every infinite sequence of next
and $\varepsilon$-reduction transitions
\[
B_{\xi^{0}}\overset{\pi_{0}}{\longrightarrow}B_{\xi^{1}}\overset
{\varepsilon}{\longrightarrow}B_{\xi_{re}^{1}}\overset{\pi_{1}}%
{\longrightarrow}B_{\xi^{2}}\overset{\varepsilon}{\longrightarrow}%
B_{\xi_{re}^{2}}\ldots
\]
with $\xi^{0}=\xi_{re}^{\prime}$ and $v_{B_{\xi_{re}^{i}}}\left(
\xi^{i+1}\right)  \neq\mathbf{0}$ $\left(  i\geq0\right)  $, there exist an
infinite sequence of next and $\varepsilon$-reduction transitions
\[
B_{\lambda^{0}}\overset{\pi_{0}}{\longrightarrow}B_{\lambda^{1}}%
\overset{\varepsilon}{\longrightarrow}B_{\lambda_{re}^{1}}\overset{\pi_{1}%
}{\longrightarrow}B_{\lambda^{2}}\overset{\varepsilon}{\longrightarrow
}B_{\lambda_{re}^{2}}\ldots
\]
with $\lambda^{0}=\psi_{re}^{\prime\prime}$ and $v_{B_{\lambda_{re}^{i}}%
}\left(  \lambda^{i+1}\right)  =v_{B_{\xi_{re}^{i}}}\left(  \xi
^{i+1}\right)  $ for every $i\geq0.$
\end{itemize}
\end{lemma}

\begin{proof}
First, we assume that both $\psi,\xi$ are different from $true$ and we
point out the following cases.

(a) $\psi=\underset{1\leq i\leq m_{1}}{%
{\displaystyle\bigwedge}
}\psi_{i},\xi=\left(  \underset{1\leq j\leq m_{2}+1}{%
{\displaystyle\bigwedge}
}\xi_{j}\right)  $ with $\xi_{m_{2}+1}\in RULTL\left(  K,AP\right)  $,
$\xi_{j}\in bLTL\left(  K,AP\right)  $ for every $j\in\left\{  1,\ldots
,m_{2}\right\}  ,$ and there exist $i_{1},\ldots,i_{h}\in\left\{
1,\ldots,m_{1}\right\}  $, $j_{1},\ldots,j_{h}\in\left\{  1,\ldots
,m_{2}\right\}  $ such that $\psi_{i_{1}}=\xi_{j_{1}},\ldots
,\psi_{i_{h}}=\xi_{j_{h}}$. Then, $\psi^{\prime}=\underset{1\leq i\leq
m_{1}}{%
{\displaystyle\bigwedge}
}\psi_{i}^{\prime},\xi^{\prime}=\underset{1\leq j\leq m_{2}+1}{%
{\displaystyle\bigwedge}
}\xi_{j}^{^{\prime}}$ where $\psi_{i}^{^{\prime}}\in\widehat{next}\left(
M_{B_{\psi},\psi_{i}}\right)  $ for every $i\in\left\{  1,\ldots
,m_{1}\right\}  ,$ $\xi_{j}^{\prime}\in\widehat{next}\left(  M_{B_{\xi}%
,\xi_{j}}\right)  $ for every $\left\{  1,\ldots,m_{2}+1\right\}  $. Clearly,
$\varphi=\left(  \underset{1\leq i\leq m_{1}}{%
{\displaystyle\bigwedge}
}\psi_{i}\right)  \wedge\left(  \underset{j\neq j_{1},\ldots,j\neq j_{h}%
}{\underset{1\leq j\leq m_{2}+1}{%
{\displaystyle\bigwedge}
}}\xi_{j}\right)  .$ Let $B_{\varphi}=\left\{  \varphi\right\}  \cup\left(
\underset{1\leq i\leq m_{1}}{%
{\displaystyle\bigcup}
}M_{B_{\psi},\psi_{i}}\right)  \cup\left(  \underset{1\leq j\leq
m_{2}+1}{%
{\displaystyle\bigcup}
}M_{B_{\xi},\xi_{j}}\right)  .$ We can prove that $B_{\varphi}$ is a $\varphi
$-consistent set (see proof of Lemma 107 in \cite{Ma-Co}). Moreover,
$M_{B_{\psi},\psi_{i}}\subseteq M_{B_{\varphi},\psi_{i}}$ and
$M_{B_{\xi},\xi_{j}}\subseteq M_{B_{\varphi},\xi_{j}}$ $\left(  1\leq i\leq
m_{1},1\leq j\leq m_{2}\right)  $, which by Lemma \ref{boithitiko lemma 1}
implies that $\varphi^{\prime}=\left(  \underset{1\leq i\leq m_{1}}{%
{\displaystyle\bigwedge}
}\psi_{i}^{\prime}\right) \\ \wedge\left(  \underset{j\neq j_{1},\ldots,j\neq
j_{h}}{\underset{1\leq j\leq m_{2}+1}{%
{\displaystyle\bigwedge}
}}\xi_{j}^{\prime}\right)  \in\widehat{next}\left(  B_{\varphi}\right)  .$
\ Therefore, $\left(  B_{\varphi},\pi,B_{\varphi^{\prime}}\right)  $ is a next
transition and
\begin{align*}
v_{B_{\varphi}}\left(  \varphi^{\prime}\right)   &  =\underset{1\leq i\leq m_{1}}{%
{\displaystyle\prod}
}v_{M_{B_{\varphi}},\psi_{i}}\left(  \psi_{i}^{\prime}\right)  \cdot\left(
\underset{j\neq j_{1},\ldots,j\neq j_{h}}{\underset{1\leq j\leq m_{2}+1}{%
{\displaystyle\prod}
}}v_{M_{B_{\varphi}},\xi_{j}}\left(  \xi_{j}^{\prime}\right)  \right) \\
&  \geq\underset{1\leq i\leq m_{1}}{%
{\displaystyle\prod}
}v_{M_{B_{\psi},\psi_{i}}}\left(  \psi_{i}^{\prime}\right)
\cdot\left(  \underset{j\neq j_{1},\ldots,j\neq j_{h}}{\underset{1\leq j\leq
m_{2}+1}{%
{\displaystyle\prod}
}}v_{M_{B_{\xi}},\xi_{j}}\left(  \xi_{j}^{\prime}\right)  \right) \\
&  \cdot\underset{1\leq k\leq h}{{\prod}}v_{M_{B_{\xi},\xi_{j_{k}}}}\left(
\xi_{j_{k}}^{\prime}\right) \\
&  =v_{B_{\psi}}\left(  \psi^{\prime}\right)  \cdot v_{B_{\xi}}\left(
\xi^{\prime}\right)
\end{align*}
where the inequality is concluded
due to the following: By Lemma \ref{boithitiko lemma 1}, we get $v_{M_{B_{\varphi}},\psi_{i}}\left(
\psi_{i}^{\prime}\right)  \geq v_{M_{B_{\psi}},\psi_{i}}\left(
\psi_{i}^{\prime}\right)  $, and $v_{M_{B_{\varphi}},\xi_{j}}\left(  \xi_{j}^{\prime}\right)  \geq v_{M_{B_{\xi
}},\xi_{j}}\left(  \xi_{j}^{\prime}\right)  $ for every $1\leq i\leq m_{1},$ $1\leq j\leq
m_{2}+1.$ It holds $v_{M_{B_{\psi}},\psi_{i}}\left(  \psi_{i}^{\prime
}\right)  =\mathbf{1,}$ $v_{M_{B_{\xi}},\xi_{j}}\left(  \xi_{j}^{\prime
}\right)  =\mathbf{1}$ $\left(  1\leq i\leq m_{1},1\leq j\leq m_{2}\right)  ,$
which implies $v_{M_{B_{\varphi}},\psi_{i}}\left(  \psi_{i}^{\prime
}\right)  =\mathbf{1,}v_{M_{B_{\varphi}},\xi_{j}}\left(  \xi_{j}^{\prime
}\right)  =\mathbf{1}$ $\left(  1\leq i\leq m_{1},1\leq j\leq m_{2}\right)  .$ Then, taking into account Remark \ref{remark_commutative}, we conclude the inequality.

We have completed the proof of (i). In order to prove (ii) we set
$\psi^{\prime\prime}=\underset{j\neq j_{1},\ldots,j\neq j_{h}}{\underset{1\leq
j\leq m_{2}+1}{%
{\displaystyle\bigwedge}
}}\xi_{j}^{\prime}.$ It holds $\varphi_{re}^{\prime}=\left(  \psi_{re}%
^{\prime}\wedge\psi_{re}^{\prime\prime}\right)  _{re}$ . We consider now the
infinite sequence of next and $\varepsilon$-reduction transitions
\[
B_{\xi^{0}}\overset{\pi_{0}}{\rightarrow}B_{\xi^{1}}\overset{\varepsilon
}{\rightarrow}B_{\xi_{re}^{1}}\overset{\pi_{1}}{\rightarrow}B_{\xi^{2}%
}\overset{\varepsilon}{\rightarrow}B_{\xi_{re}^{2}}\ldots
\]
with $\xi^{0}=\xi_{re}^{\prime}$ and $v_{B_{\xi_{re}^{i}}}\left(
\xi^{i+1}\right)  \neq\mathbf{0}$ $\left(  i\geq0\right)  .$ Clearly,
\[
\xi^{0}=\xi_{re}^{\prime}=\left(  \psi_{re}^{\prime\prime}\wedge\left(
\psi_{j_{1}}^{\prime}\wedge\ldots\wedge\psi_{j_{h}}^{\prime}\right)
_{re}\right)  _{re}.
\]
Then, for $\lambda^{0}=\psi_{re}^{\prime\prime}$, and $\zeta^{0}=\left(
\xi_{j_{1}}^{\prime}\wedge\ldots\wedge\xi_{j_{h}}^{\prime}\right)  _{re}$,
by induction on $i$ and Lemma \ref{boithitiko lemma 2}, we obtain that for every $i\geq0$,
there exist a $\lambda_{re}^{i}$-consistent set $B_{\lambda_{re}^{i}}$, and a
$\zeta_{re}^{i}$-consistent set $B_{\zeta_{re}^{i}}$, and formulas
$\lambda^{i+1}\in next\left(  B_{\lambda_{re}^{i}}\right)  ,$ $\zeta^{i+1}\in
next\left(  B_{\zeta_{re}^{i}}\right)  $ such that $\xi_{re}^{i+1}=\left(
\lambda_{re}^{i+1}\wedge\zeta_{re}^{i+1}\right)  _{re}$, and $v_{B_{\xi
_{re}^{i}}}\left(  \xi^{i+1}\right)  =v_{B_{\lambda_{re}^{i}}}\left(
\lambda^{i+1}\right)  \cdot v_{B_{\zeta_{re}^{i}}}\left(  \zeta^{i+1}\right)
$. For every $i\geq0$, $v_{B_{\xi_{re}^{i}}}\left(  \xi^{i+1}\right)
\neq\mathbf{0}$ and $\zeta_{re}^{i}$ is boolean, hence $v_{B_{\zeta_{re}^{i}}%
}\left(  \zeta^{i+1}\right)  =\mathbf{1}$, i.e., $v_{B_{\xi_{re}^{i}}}\left(
\xi^{i+1}\right)  =v_{B_{\lambda_{re}^{i}}}\left(  \lambda^{i+1}\right)  $
for every $i\geq0$. So, the sequence
\[
B_{\lambda^{0}}\overset{\pi_{0}}{\rightarrow}B_{\lambda^{1}}\overset{\ast
}{\rightarrow}B_{\lambda_{re}^{1}}\overset{\pi_{1}}{\rightarrow}B_{\lambda
^{2}}\overset{\ast}{\rightarrow}B_{\lambda_{re}^{2}}\ldots
\]
satisfies the lemma's claim.

(b) If $\left(  \psi\wedge\xi\right)  _{re}=\psi\wedge\xi,$ we set
$B_{\varphi}=\left\{  \varphi\right\}  \cup B_{\psi}\cup B_{\xi}$, and we proceed
in the same way. Finally, it is trivial to prove our claim in the cases where
at least one of $\psi,\xi$ equals to $true.$
\end{proof}

\begin{lemma}
\label{Lemma conjuction}Let $\varphi=\psi\wedge\xi$ with $\psi\in bLTL\left(
K,AP\right)  $ and $\xi\in RULTL\left(  K,AP\right)  $. If $\mathcal{A_{\psi}%
},\mathcal{A_{\xi}}$ recognize $\left\Vert \psi\right\Vert ,\left\Vert
\mathcal{\xi}\right\Vert $ respectively, then $\mathcal{A_{\varphi}}$
recognizes $\left\Vert \varphi\right\Vert $.
\end{lemma}

\begin{proof}
Let $\mathcal{A_{\psi}=}\left(  Q_{1},wt_{1},I_{1},\mathcal{F}_{1}\right)  $,
$\mathcal{A_{\xi}=}\left(  Q_{2},wt_{2},I_{2},\mathcal{F}_{2}\right)  $,
$\mathcal{A_{\varphi}=}\left(  Q,wt,I,\mathcal{F}\right)  ,$ and $w=\pi_{0}%
\pi_{1}\pi_{2}\ldots\in\left(  \mathcal{P}\left(  AP\right)  \right)
^{\omega}$. First, we show that $\left(  \left\Vert \mathcal{A_{\varphi}%
}\right\Vert ,w\right)  \leq\left(  \left\Vert \varphi\right\Vert ,w\right)
.$ In order to do this, it is necessary to prove that for every $P_{w}\in
next_{\mathcal{A}_{\varphi}}\left(  w\right)  $, there exist $P_{w}^{1}\in
next_{\mathcal{A_{\psi}}}\left(  w\right)  $ and $P_{w}^{2}\in
next_{\mathcal{A_{\xi}}}\left(  w\right)  $ such that
$weight_{\mathcal{A_{\varphi}}}\left(  P_{w}\right)  =weight_{\mathcal{A}%
_{\psi}}\left(  P_{w}^{1}\right)  \cdot weight_{\mathcal{A_{\xi}}}\left(
P_{w}^{2}\right)  $. If $weight_{\mathcal{A_{\varphi}}}\left(  P_{w}\right)
=\mathbf{0}$, then the paths $P_{w}^{1}$, $P_{w}^{2}$ can be defined in the
obvious way. Otherwise, it is possible to define the paths $P_{w}^{1}$,
$P_{w}^{2}$ due to the following. At every next transition of $P_{w}$ the
automaton $\mathcal{A}_{\varphi}$ simulates two next transitions, one of
$\mathcal{A_{\psi}}$ and one of $\mathcal{A}_{\xi},$ and multiplies their
weights. Since $\varphi$ is reduced there are two possibilities. Either
$P_{w}$ starts with a next transition, or if not, before realizing the first
next transition the automaton realizes a finite number of $\varepsilon
$-transitions connecting $\varphi$-consistent sets. In the second case, the
weight of the path coincides with the weight of the suffix path starting with
the first next transition. So it suffices to prove our claim for paths
$P_{w}\in next_{\mathcal{A}_{\varphi}}\left(  w\right)  $ with non-zero weight
of the form
\[
P_{w}:B_{\varphi^{0}}\overset{\pi_{0}}{\rightarrow}B_{\varphi^{1}}%
\overset{\ast}{\rightarrow}B_{\varphi_{re}^{1}}\overset{\pi_{1}}{\rightarrow
}B_{\varphi^{2}}\overset{\ast}{\rightarrow}B_{\varphi_{re}^{2}}\ldots
\]
where $\varphi^{0}=\left(  \psi\wedge\xi\right)  _{re}=\psi\wedge\xi.$

We let $\psi^{0}=\psi$ and $\xi^{0}=\xi$. By induction and Lemma
\ref{boithitiko lemma 2}, we get that for every $i\geq0,$ there exist a
$\psi_{re}^{i}$-consistent set $B_{\psi_{re}^{i}}$, a $\xi_{re}^{i}%
$-consistent set $B_{\xi_{re}^{i}}$, and formulas $\psi^{i+1}\in next\left(
B_{\psi_{re}^{i}}\right)  ,$ $\xi^{i+1}\in next\left(  B_{\xi_{re}^{i}%
}\right)  $ such that $\varphi_{re}^{i+1}=\left(  \psi_{re}^{i+1}\wedge
\xi_{re}^{i+1}\right)  _{re}$ and $wt\left(    B_{\varphi_{re}^{i}}%
,\pi_{i},B_{\varphi^{i+1}} \right)  =v_{B_{\psi_{re}^{i}}}\left(
\psi^{i+1}\right)  \cdot v_{B_{\xi_{re}^{i}}}\left(  \xi^{i+1}\right)  .$\footnote{We shall call this inductive procedure, Procedure 1.}

So, the sequences
\[
P_{w}^{1}:B_{\psi^{0}}\overset{\pi_{0}}{\rightarrow}B_{\psi^{1}}%
\overset{\varepsilon}{\rightarrow}B_{\psi_{re}^{1}}\overset{\pi_{1}%
}{\rightarrow}B_{\psi^{2}}\overset{\varepsilon}{\rightarrow}B_{\psi_{re}^{2}%
}\ldots
\]

and%

\[
P_{w}^{2}:B_{\xi^{0}}\overset{\pi_{0}}{\rightarrow}B_{\xi^{1}}\overset
{\varepsilon}{\rightarrow}B_{\xi_{re}^{1}}\overset{\pi_{1}}{\rightarrow}%
B_{\xi^{2}}\overset{\varepsilon}{\rightarrow}B_{\xi_{re}^{2}}\ldots
\]
form successful paths of next and $\varepsilon$-transitions of
$\mathcal{A_{\psi}}$ and $\mathcal{A_{\xi}}$, respectively. We note that in
the above paths for every $i\geq0$, $B_{\psi^{i}},B_{\xi^{i}}$ are non-empty sets and
$\psi^{i}$-consistent, $\xi^{i}$-consistent respectively. It holds%

\[%
\begin{array}
[c]{cl}%
weight_{\mathcal{A_{\varphi}}}\left(  P_{w}\right)  & =Val^{\omega}\left(
wt\left(   B_{\varphi_{re}^{i}},\pi_{i},B_{\varphi^{i+1}}\right)
 \right)  _{i\geq0}\\
& =Val^{\omega}\left(  wt_{1}\left(   B_{\psi_{re}^{i}},\pi_{i}%
,B_{\psi^{i+1}}\right)    \cdot wt_{2}\left(    B_{\xi_{re}^{i}%
},\pi_{i},B_{\xi^{i+1}}\right)    \right)  _{i\geq0}\\
& =Val^{\omega}\left(  wt_{2}\left(   B_{\xi_{re}^{i}},\pi_{i}%
,B_{\xi^{i+1}}\right)   \right)  _{i\geq0}\\
& =Val^{\omega}\left(  wt_{1}\left(    B_{\psi_{re}^{i}},\pi_{i}%
,B_{\psi^{i}}\right)    \right)  _{i\geq0}\cdot Val^{\omega}\left(
wt_{2}\left(   B_{\xi_{re}^{i}},\pi_{i},B_{\xi^{i}}  \right)
\right)  _{i\geq0}\\
& =weight_{\mathcal{A_{\psi}}}\left(  P_{w}^{1}\right)  \cdot
weight_{\mathcal{A_{\xi}}}\left(  P_{w}^{2}\right)
\end{array}
\]
where the third and fourth equality hold by the fact that $P_{w}^{1}$ is a
path on the boolean formula $\psi$ with $weight_{\mathcal{A}_{\psi}}\left(
P_{w}^{1}\right)  =\boldsymbol{1}$, i.e., the weight of each next transition
appearing in the path is equal to $\mathbf{1}$. We thus conclude that for
every $k\in pri_{\mathcal{A_{\varphi}}}\left(  w\right)  \setminus\left\{
\mathbf{0}\right\}  $ there exist $k_{1}\in pri_{\mathcal{A_{\psi}}}\left(
w\right)  \setminus\left\{  \mathbf{0}\right\}  =\left\{  \mathbf{1}\right\}
,k_{2}\in pri_{\mathcal{A_{\xi}}}\left(  w\right)  \backslash\left\{
\mathbf{0}\right\}  $ such that $k\leq k_{1}\cdot k_{2}=k_{2}$ which implies%
\[%
\begin{array}
[c]{ll}%
\left(  \left\Vert \mathcal{A}_{\varphi}\right\Vert ,w\right)  &
=\underset{k\in pri_{\mathcal{A}_{\varphi}}\left(  w\right)  \setminus\left\{
\mathbf{0}\right\}  }{\sum}k\\
& \leq\underset{k_{2}\in pri_{\mathcal{A}_{\xi}}\left(  w\right)
\backslash\left\{  \mathbf{0}\right\}  }{\sum}k_{2}\\
& =\left(  \underset{k_{1}\in pri_{\mathcal{A_{\psi}}}\left(  w\right)
\backslash\left\{  \mathbf{0}\right\}  }{\sum}k_{1}\right)  \cdot\left(
\underset{k_{2}\in pri_{\mathcal{A_{\xi}}}\left(  w\right)  \backslash\left\{
\mathbf{0}\right\}  }{\sum}k_{2}\right) \\
& =\left(  \left\Vert \psi\right\Vert ,w\right)  \cdot\left(  \left\Vert
\xi\right\Vert ,w\right) \\
& =\left(  \left\Vert \varphi\right\Vert ,w\right).
\end{array}
\]
Clearly, if $\left(  \left\Vert \mathcal{A}_{\varphi}\right\Vert ,w\right)
=\mathbf{0}$, it holds $\left(  \left\Vert \mathcal{A}_{\varphi}\right\Vert
,w\right)  \leq\left(  \left\Vert \varphi\right\Vert ,w\right)  .$

Now, we prove that $\left(  \left\Vert \varphi\right\Vert ,w\right)
\leqslant\left(  \left\Vert \mathcal{A}_{\varphi}\right\Vert ,w\right)  $. For
this, we first prove that for every $P_{w}^{1}\in next_{\mathcal{A_{\psi}}%
}\left(  w\right)  $ with $weight_{\mathcal{A_{\psi}}}\left(  w\right)
=\mathbf{1}$ and every $P_{w}^{2}\in next_{\mathcal{A_{\xi}}}\left(  w\right)
$ with $weight_{\mathcal{A_{\xi}}}\left(  w\right)  \neq\mathbf{0}$ there is a
$P_{w}\in next_{\mathcal{A_{\varphi}}}\left(  w\right)  $ such that
$weight_{\mathcal{A_{\psi}}}\left(  P_{w}^{1}\right)  \cdot
weight_{\mathcal{A_{\xi}}}\left(  P_{w}^{2}\right)  =weight_{\mathcal{A_{\xi}%
}}\left(  P_{w}^{2}\right)  \leq weight_{\mathcal{A_{\varphi}}}\left(
P_{w}\right)  $. We let
\[
P_{w}^{1}:B_{\psi^{0}}\overset{\pi_{0}}{\rightarrow}B_{\psi^{1}}%
\overset{\varepsilon}{\rightarrow}B_{\psi_{re}^{1}}\overset{\pi_{1}%
}{\rightarrow}B_{\psi^{2}}\overset{\varepsilon}{\rightarrow}B_{\psi_{re}^{2}%
}\ldots
\]

and%

\[
P_{w}^{2}:B_{\xi^{0}}\overset{\pi_{0}}{\rightarrow}B_{\xi^{1}}\overset
{\varepsilon}{\rightarrow}B_{\xi_{re}^{1}}\overset{\pi_{1}}{\rightarrow}%
B_{\xi^{2}}\overset{\varepsilon}{\rightarrow}B_{\xi_{re}^{2}}\ldots
\]

Clearly, $P_{w}^{1}$,$P_{w}^{2}$ contain no empty states and $\xi^{i+1}%
\in\widehat{next}\left(  B_{\xi_{re}^{i}}\right)  $, $\psi^{i+1}\in
\widehat{next}\left(  B_{\psi_{re}^{i}}\right)  $ for every $i\geq0$. Taking
into account Remark \ref{Remark}, we distinguish the following cases.

$\mathbf{(a)}$ The set $cl\left(  \psi\right)  \bigcap cl\left(  \xi\right)  $
contains no subformulas of the form $\varphi_{1}U\varphi_{2}$.

$\mathbf{(b)}$
For every $\varphi_{1}U\varphi_{2}\in cl\left(  \psi\right)  \bigcap cl\left(
\xi\right)  $, $\varphi_{1}U\varphi_{2}$ does not appear in the scope of an
always operator $\square$ in at least one of $\psi$, $\xi$.

$\mathbf{(c)}$ For
every $\varphi_{1}U\varphi_{2}\in cl\left(  \psi\right)  \bigcap cl\left(
\xi\right)  $ that is in the scope of an always operator $\square$ in both
$\psi$, $\xi$, there is an $n\geq0$, such that the acceptance condition from
$\varphi_{1}U\varphi_{2}$ is satisfied for every position $n^{\prime}\geqslant
n$ in at least one of $P_{w}^{1}$,$P_{w}^{2}$.

$\mathbf{(d)}$ There is at
least one $\varphi_{1}U\varphi_{2}$$\in cl\left(  \psi\right)  \bigcap
cl\left(  \xi\right)  $ that is in the scope of an always operator $\square$
in both $\psi,\xi$, and in both $P_{w}^{1}$,$P_{w}^{2}$ the acceptance
condition from $\varphi_{1}U\varphi_{2}$ is satisfied for infinitely many
positions, and not satisfied for infinitely many positions, too.

If case (a), or (b), or (c) holds we act as follows. Inductively, we can determine a path
$P_{w}$ of next and $\varepsilon$-transitions of $\mathcal{A}_{\varphi}$ over $w$
\[
P_{w}:B_{\varphi^{0}}\overset{\pi_{0}}{\rightarrow}B_{\varphi^{1}}%
\overset{\ast}{\rightarrow}B_{\varphi_{re}^{1}}\overset{\pi_{1}}{\rightarrow
}B_{\varphi^{2}}\overset{\ast}{\rightarrow}B_{\varphi_{re}^{2}}\ldots
\]
in the following way\footnote{We shall call this inductive procedure, Procedure 2.}: For $\varphi=\varphi^{0}=\left(  \psi^{0}\wedge\xi
^{0}\right)  _{re}$, and for the next transitions $\left(  B_{\psi^{0}}%
,\pi_{0},B_{\psi^{1}}\right)  ,\left(  B_{\xi^{0}},\pi_{0},B_{\xi^{1}}\right)
$ we apply Lemma \ref{boithitiko lemma 3}, and we obtain $B_{\varphi^{0}}%
\neq\emptyset$, $\varphi^{1}\in next\left(  B_{\varphi^{0}}\right)  $ such
that for every $B_{\varphi^{1}}$ the triple $\left(  B_{\varphi^{0}},\pi
_{0},B_{\varphi^{1}}\right)  $ is a next transition with%
\[
wt_{1}\left(  B_{\psi_{re}^{0}},\pi_{0},B_{\psi^{1}}\right)  \cdot
wt_{2}\left(  B_{\xi_{re}^{0}},\pi_{0},B_{\xi^{1}}\right)  \leq wt\left(
B_{\varphi^{0}},\pi_{0},B_{\varphi^{1}}\right)  .
\]
We also get $\overline{\xi}^{\left(  1,1\right)  }\in RULTL\left(
K,AP\right)  $ with $\varphi_{re}^{1}=\left(  \psi_{re}^{1}\wedge\overline
{\xi}_{re}^{\left(  1,1\right)  }\right)  _{re}$, and an infinite sequence of
next and $\varepsilon$-transitions $B_{\overline{\xi}_{re}^{\left(
1,1\right)  }}\overset{\pi_{1}}{\rightarrow}B_{\overline{\xi}^{\left(
1,2\right)  }}\overset{\ast}{\rightarrow}B_{\overline{\xi}_{re}^{\left(
1,2\right)  }}\ldots$ with $v_{B_{\overline{\xi}_{re}^{\left(  1,i\right)  }}%
}\left(  \overline{\xi}^{\left(  1,i+1\right)  }\right)  =wt_{2}$$\left(
B_{\xi_{re}^{i}},\pi_{i},B_{\xi^{i+1}}\right)  $ for every $i\geq1$. \ Assume
now that $B_{\varphi^{j-1}}$ are built with the previous procedure for every
$j\leq m,$ which implies that there exists $\overline{\xi}^{\left(
m,1\right)  }\in RULTL\left(  K,AP\right)  $ such that $\varphi_{re}%
^{m}=\left(  \psi_{re}^{m}\wedge\overline{\xi}_{re}^{\left(  m,1\right)
}\right)  _{re}$, and an infinite sequence of next and $\varepsilon$-transitions
$B_{\overline{\xi}_{re}^{\left(  m,1\right)  }}\overset{\pi_{m}}{\rightarrow
}B_{\overline{\xi}^{\left(  m,2\right)  }}\overset{\ast}{\rightarrow
}B_{\overline{\xi}_{re}^{\left(  m,2\right)  }}\ldots$ with $v_{B_{\overline
{\xi}_{re}^{\left(  m,i\right)  }}}\left(  \overline{\xi}^{\left(  m,i+1\right)
}\right)  =wt_{2}\left(  B_{\xi_{re}^{m-1+i}},\pi_{m-1+i},B_{\xi^{m+i}}\right)
$ for all $i\geq1.$ We apply Lemma  \ref{boithitiko lemma 3} for $\varphi
_{re}^{m}=\left(  \psi_{re}^{m}\wedge\overline{\xi}_{re}^{\left(
m,1\right)  }\right)  _{re}$ and the next transitions $\left(  B_{\psi
_{re}^{m}},\pi_{m},B_{\psi^{m+1}}\right)  ,\left(  B_{\overline{\xi}%
_{re}^{\left(  m,1\right)  }},\pi_{m},B_{\overline{\xi}^{\left(  m,2\right)
}}\right)  .$ We get $B_{\varphi_{re}^{m}}\neq\emptyset,\varphi^{m+1}\in
next\left(  B_{\varphi_{re}^{m}}\right)  $ such that for every
$B_{\varphi^{m+1}}$, $\left(  B_{\varphi_{re}^{m}},\pi
_{m},B_{\varphi^{m+1}}\right)  $ is a next transition with

\begin{align*}
& wt_{1}\left(  B_{\psi_{re}^{m}},\pi_{m},B_{\psi^{m+1}}\right)  \cdot
v_{B_{_{\overline{\xi}_{re}^{\left(  m,1\right)  }}}}\left(  \overline{\xi
}^{\left(  m,2\right)  }\right)  \\
& =wt_{1}\left(  B_{\psi_{re}^{m}},\pi_{m},B_{\psi^{m+1}}\right)  \cdot
wt_{2}\left(  B_{\xi_{re}^{m}},\pi_{m},B_{\xi^{m+1}}\right)  \\
& \leq wt\left(  B_{\varphi_{re}^{m}},\pi_{m},B_{\varphi^{m+1}}\right)  .
\end{align*}

Hence, for every $i\geq0$ it holds $wt_{1}\left(
B_{\psi_{re}^{i}},\pi_{i},B_{\psi^{i+1}}  \right)  \cdot wt_{2}\left(
  B_{\xi_{re}^{i}},\pi_{i},B_{\xi^{i+1}} \right)  \leq wt\left(
  B_{\varphi_{re}^{i}},\pi_{i},B_{\varphi^{i+1}}  \right)  ,$
which implies%

\begin{align*}
&  Val^{\omega}\left(  wt_{1}\left(   B_{\psi_{re}^{i}},\pi_{i}%
,B_{\psi^{i+1}}  \right)  \cdot wt_{2}\left(    B_{\xi_{re}^{i}%
},\pi_{i},B_{\xi^{i+1}}\right)    \right)  _{i\geq0}\\
&  \leq Val^{\omega}\left(  wt\left(    B_{\varphi_{re}^{i}},\pi
_{i},B_{\varphi^{i+1}}\right)   \right)  _{i\geq0}\\
&  \Longrightarrow\\
&  Val^{\omega}\left(  wt_{2}\left(    B_{\xi_{re}^{i}},\pi_{i}%
,B_{\xi^{i+1}}\right)    \right)  _{i\geq0}\\
&  \leq Val^{\omega}\left(  wt\left(    B_{\varphi_{re}^{i}},\pi
_{i},B_{\varphi^{i+1}}\right)  \right)  _{i\geq0}\\
&  \Longrightarrow\\
&  Val^{\omega}\left(  wt_{1}\left(    B_{\psi_{re}^{i}},\pi_{i}%
,B_{\psi^{i+1}}\right)    \right)  _{i\geq0}\cdot Val^{\omega}\left(
wt_{2}\left(   B_{\xi_{re}^{i}},\pi_{i},B_{\xi^{i+1}}  \right)
\right)  _{i\geq0}\\
&  \leq Val^{\omega}\left(  wt\left(   B_{\varphi_{re}^{i}},\pi
_{i},B_{\varphi^{i+1}}\right)   \right)  _{i\geq0}%
\end{align*}
where the first inequality holds by Lemma \ref{Valuation inequality}, and the
third and fourth inequality are derived by the fact that $wt_{1}
\left(  B_{\psi_{re}^{i}},\pi_{i},B_{\psi^{i+1}}  \right)  =\mathbf{1}$
for every $i\geq0.$ Thus,
\[
weight_{\mathcal{A_{\psi}}}\left(  P_{w}^{1}\right)  \cdot
weight_{\mathcal{A_{\xi}}}\left(  P_{w}^{2}\right)  \leq
weight_{\mathcal{A_{\varphi}}}\left(  P_{w}\right)  .
\]
Following the constructive proof of Lemma \ref{boithitiko lemma 3}, and since either (a), or (b), or (c) holds, we get that for all $\varphi_{1}U\varphi_{2}$$\in cl\left(  \psi\right)  \bigcap cl\left(  \xi\right)  $, $P_{w}$ satisfies the acceptance condition for $\varphi_{1}U\varphi_{2}$ for infinitely many $i\geq 0$, i.e., $P_{w}$ is successful.

Assume now that case (d) holds, and let $\varphi_{1}U\varphi_{2}\in cl\left(
\psi\right)  \cap cl\left(  \xi\right)  $ with the property of case (d).
Clearly, $\varphi_{1}U\varphi_{2}$ is boolean. Let $i_{1}<i_{2}<\ldots$ be the
sequence of positions with $B_{\psi^{i_{p}}}\in F_{\varphi_{1}U\varphi_{2}%
},p\geq1,$ and with the additional property that in positions $i_{1}%
-1<i_{2}-1<\ldots$ the acceptance condition from $\varphi_{1}U\varphi_{2}$ is
not satisfied. Then, due to the fact that $\varphi_{1}U\varphi_{2}$ is not in
the scope of a next operator, we can determine a path $\widehat{P_{w}^{2}}$
of next and $\varepsilon$-reduction transitions of $\mathcal{A}_{\xi}$ over
$w$ such that for every position $i_{p},p\geq1,$ the acceptance condition of
$\varphi_{1}U\varphi_{2}$ is satisfied, and $weight_{\mathcal{A}_{\psi}%
}\left(  P_{w}^{1}\right)  \cdot weight_{\mathcal{A}_{\xi}}\left(  P_{w}%
^{2}\right)  \leq weight_{\mathcal{A}_{\psi}}\left(  P_{w}^{1}\right)  \cdot
weight_{\mathcal{A}_{\xi}}\left(  \widehat{P_{w}^{2}}\right)  .$ Also,
$\widehat{P_{w}^{2}}$ can be chosen in such a way that the above statement is
satisfied for every $\varphi_{1}U\varphi_{2}\in cl\left(  \psi\right)  \cap
cl\left(  \xi\right)  $ with the property of case (d).\footnote{The existence of this path is determined by Procedure 3, which is presented in the Appendix of this paper.} We construct the path
$P_{w}\in next_{\mathcal{A}_{\varphi}}\left(  w\right)  $ by $P_{w}^{1}$ and
$\widehat{P_{w}^{2}}$, in the same way that $P_{w}$ was constructed by
$P_{w}^{1}$ and $P_{w}^{2}$ in cases (a), (b), (c). Then, $P_{w}$ is
successful and we get $weight_{\mathcal{A_{\psi}}}\left(  P_{w}^{1}\right)
\cdot weight_{\mathcal{A_{\xi}}}\left(  P_{w}^{2}\right)  \leq
weight_{\mathcal{A}_{\psi}}\left(  P_{w}^{1}\right)  \cdot
weight_{\mathcal{A_{\xi}}}\left(  \widehat{P_{w}^{2}}\right)  \leq
weight_{\mathcal{A_{\varphi}}}\left(  P_{w}\right)  .$

We have shown that for every $P_{w}^{1}\in next_{\mathcal{A_{\psi}}}\left(
w\right)  $ with $weight_{\mathcal{A_{\psi}}}\left(  w\right)  =\mathbf{1}$
and every $P_{w}^{2}\in next_{\mathcal{A_{\xi}}}\left(  w\right)  $ with
$weight_{\mathcal{A_{\xi}}}\left(  w\right)  \neq\mathbf{0}$ there is a
$P_{w}\in next_{\mathcal{A_{\varphi}}}\left(  w\right)  $ with
\[weight_{\mathcal{A_{\psi}}}\left(  P_{w}^{1}\right)  \cdot
weight_{\mathcal{A_{\xi}}}\left(  P_{w}^{2}\right)  =weight_{\mathcal{A_{\xi}%
}}\left(  P_{w}^{2}\right)  \leq weight_{\mathcal{A_{\varphi}}}\left(
P_{w}\right) . \] This implies that for every $k_{1}\in pri_{\mathcal{A}_{\psi
}}\left(  w\right)  \setminus\left\{  \mathbf{0}\right\}  ,k_{2}\in
pri_{\mathcal{A}_{\xi}}\left(  w\right)  \setminus\left\{  \mathbf{0}\right\}
$ there exists $k\in pri_{\mathcal{A}_{\varphi}}\left(  w\right)  $ such that
$k_{1}\cdot k_{2}=k_{2}\leq k,$ i.e.,
\begin{align*}
\left(  \left\Vert \varphi\right\Vert ,w\right)   &  =\left(  \left\Vert
\psi\right\Vert ,w\right)  \cdot\left(  \left\Vert \xi\right\Vert ,w\right) \\
&  =\left(  \underset{k_{1}\in pri_{\mathcal{A}_{\psi}}\left(  w\right)
\setminus\left\{  \mathbf{0}\right\}  }{%
{\displaystyle\sum}
}k_{1}\right)  \cdot\left(  \underset{k_{2}\in pri_{\mathcal{A}_{\xi}}\left(
w\right)  \setminus\left\{  \mathbf{0}\right\}  }{%
{\displaystyle\sum}
}k_{2}\right) \\
&  =\left(  \underset{k_{2}\in pri_{\mathcal{A}_{\xi}}\left(  w\right)
\setminus\left\{  \mathbf{0}\right\}  }{%
{\displaystyle\sum}
}k_{2}\right) \\
&  \leq\left(  \underset{k\in pri_{\mathcal{A}_{\varphi}}\left(  w\right)
\setminus\left\{  \mathbf{0}\right\}  }{%
{\displaystyle\sum}
}k\right) \\
&  =\left(  \left\Vert \mathcal{A}_{\varphi}\right\Vert ,w\right)
\end{align*}
as wanted. Hence, we have shown that $\left(  \left\Vert \varphi\right\Vert
,w\right)  \leq\left(  \left\Vert \mathcal{A}_{\varphi}\right\Vert ,w\right)
,$ and $\left(  \left\Vert \mathcal{A}_{\varphi}\right\Vert ,w\right)\\
\leq\left(  \left\Vert \varphi\right\Vert ,w\right)  $ for every $w\in\left(
\mathcal{P}\left(  AP\right)  \right)  ^{\omega},$ which implies that $\left(
\left\Vert \varphi\right\Vert ,w\right)  =\left(  \left\Vert \mathcal{A}%
_{\varphi}\right\Vert ,w\right)  $, and the proof is completed.
\end{proof}

The proof of the Lemma \ref{boithitiko until 1} can be found in the Appendix. Then, Lemma \ref{boithitiko until 2} can be proved with the same arguments with ones we used in the proof of Lemma \ref{boithitiko until 1}.
\begin{lemma}
\label{boithitiko until 1}Let $\psi U\xi\in RULTL\left(  K,AP\right)  $ with
$\psi,\xi\in stLTL\left(  K,AP\right)  ,$ and $\xi_{j}\in bLTL\left(
K,AP\right)  $ be reduced formulas $\left(  1\leq j\leq k,k\geq1\right)  $ and
$\pi\in\mathcal{P}\left(  AP\right)  .$ Let \ $\left(  B_{\psi},\pi
,B_{\psi^{\prime}}\right)  ,\left(  B_{\xi_{j}},\pi,B_{\xi_{j}^{\prime}%
}\right)  $ be next transitions with $B_{\psi}\neq\emptyset,\psi^{\prime}%
\in\widehat{next}\left(  B_{\psi}\right)  ,\\B_{\xi_{j}}\neq\emptyset,\xi
_{j}^{\prime}\in\widehat{next}\left(  B_{\xi_{j}}\right)  \left(  1\leq j\leq
k\right)  ,$ and $v_{B_{\psi}}\left(  \psi^{\prime}\right)  \cdot
\underset{1\leq j\leq k}{%
{\displaystyle\prod}
}v_{B_{\xi_{j}}}\left(  \xi_{j}^{\prime}\right)  \neq\mathbf{0.}$\footnote{Since $\xi_{j}\left(1\leq j \leq k\right)$ are boolean, by Remark \ref{remark_commutative}, and the fact that $k\cdot\mathbf{1}=\mathbf{1}\cdot k=k, $ $k\cdot\mathbf{0}=\mathbf{0}\cdot k=\mathbf{0} $ for every $k\in K$, we conclude that the product $v_{B_{\psi}}\left(  \psi^{\prime}\right)  \cdot
\underset{1\leq j\leq k}{%
{\displaystyle\prod}
}v_{B_{\xi_{j}}}\left(  \xi_{j}^{\prime}\right)$ is well defined.}

Then, for $\varphi=\left(  \left(  \psi U\xi\right)  \wedge\left(
\underset{1\leq j\leq k}{%
{\displaystyle\bigwedge}
}\xi_{j}\right)  \right)  _{re}$ there exist $B_{\varphi}\neq\emptyset
,\varphi^{\prime}\in\widehat{next}\left(  B_{\varphi}\right)  ,$ $\psi
^{\prime\prime},$ $\xi_{j}^{\prime\prime}\in bLTL\left(  K,AP\right)  \left(
1\leq j\leq k\right)  $ such that

(i) $\left(  B_{\varphi},\pi,B_{\varphi^{\prime}}\right)  $ is a next
transition and $v_{B_{\varphi}}\left(  \varphi^{\prime}\right)  \geq
v_{B_{\psi}}\left(  \psi^{\prime}\right)  \cdot\underset{1\leq j\leq k}{%
{\displaystyle\prod}
}v_{B_{\xi_{j}}}\left(  \xi_{j}^{\prime}\right)  ,$

(ii) $\varphi_{re}^{\prime}=\left(  \left(  \psi U\xi\right)  \wedge\psi
_{re}^{\prime\prime}\wedge\left(  \underset{1\leq j\leq k}{%
{\displaystyle\bigwedge}
}\left(  \xi_{j}^{\prime\prime}\right)  _{re}\right)  \right)  _{re}$ and for
every infinite sequence of next and $\varepsilon$-reduction transitions%
\[
B_{\psi^{0}}\overset{\pi_{0}}{\rightarrow}B_{\psi^{1}}\overset{\varepsilon
}{\rightarrow}B_{\psi_{re}^{1}}\overset{\pi_{1}}{\rightarrow}B_{\psi^{2}%
}\overset{\varepsilon}{\rightarrow}B_{\psi_{re}^{2}}\ldots
\]
with $\psi^{0}=\psi_{re}^{\prime}\left(  \text{resp. }\psi^{0}=\left(  \xi
_{j}^{\prime}\right)  _{re},1\leq j\leq k\right)  $ and $v_{B_{\psi_{re}^{i}}%
}\left(  \psi^{i+1}\right)  \neq\mathbf{0}$ $\left(  i\geq0\right)  ,$ there
exists an infinite sequence of next and $\varepsilon$-reduction transitions
\[
B_{\lambda^{0}}\overset{\pi_{0}}{\rightarrow}B_{\lambda^{1}}\overset
{\varepsilon}{\rightarrow}B_{\lambda_{re}^{1}}\overset{\pi_{1}}{\rightarrow
}B_{\lambda^{2}}\overset{\varepsilon}{\rightarrow}B_{\lambda_{re}^{2}}\ldots
\]
with $\lambda^{0}=\psi_{re}^{\prime\prime}$ \ $\left(  \text{resp. }%
\lambda^{0}=\left(  \xi_{j}^{\prime\prime}\right)  _{re},1\leq j\leq k\right)
$ and $v_{B_{\lambda_{re}^{i}}}\left(  \lambda^{i+1}\right)  =v_{B_{\psi
_{re}^{i}}}\left(  \psi^{i+1}\right)  $ for every $i\geq0.$
\end{lemma}

\begin{lemma}
\label{boithitiko until 2}Let $\psi U\xi\in RULTL\left(  K,AP\right)  $ with
$\psi,\xi\in stLTL\left(  K,AP\right)  ,$ and $\xi_{j}\in bLTL\left(
K,AP\right)  $ be reduced formulas $\left(  1\leq j\leq k,k\geq1\right)  $ and
$\pi\in\mathcal{P}\left(  AP\right)  .$ Let \ $\left(  B_{\xi},\pi
,B_{\xi^{\prime}}\right)  ,\left(  B_{\xi_{j}},\pi,B_{\xi_{j}^{\prime}%
}\right)  $ be next transitions with $B_{\xi}\neq\emptyset,\xi^{\prime}%
\in\widehat{next}\left(  B_{\xi}\right)  ,\\B_{\xi_{j}}\neq\emptyset,\xi
_{j}^{\prime}\in\widehat{next}\left(  B_{\xi_{j}}\right)  \left(  1\leq j\leq
k\right)  ,$ and $v_{B_{\xi}}\left(  \xi^{\prime}\right)  \cdot\underset{1\leq
j\leq k}{%
{\displaystyle\prod}
}v_{B_{\xi_{j}}}\left(  \xi_{j}^{\prime}\right)  \neq\mathbf{0.}$

Then, for $\varphi=\left(  \left(  \psi U\xi\right)  \wedge\left(
\underset{1\leq j\leq k}{%
{\displaystyle\bigwedge}
}\xi_{j}\right)  \right)  _{re}$ there exist $B_{\varphi}\neq\emptyset
,\varphi^{\prime}\in\widehat{next}\left(  B_{\varphi}\right)  ,$ $\xi
^{\prime\prime},$ $\xi_{j}^{\prime\prime}\in bLTL\left(  K,AP\right)  \left(
1\leq j\leq k\right)  $ such that

(i) $\left(  B_{\varphi},\pi,B_{\varphi^{\prime}}\right)  $ is a next
transition and $v_{B_{\varphi}}\left(  \varphi^{\prime}\right)  \geq
v_{B_{\xi}}\left(  \xi^{\prime}\right)  \cdot\underset{1\leq j\leq k}{%
{\displaystyle\prod}
}v_{B_{\xi_{j}}}\left(  \xi_{j}^{\prime}\right)  ,$

(ii) $\varphi_{re}^{\prime}=\left(  \xi_{re}^{\prime\prime}\wedge\left(
\underset{1\leq j\leq k}{%
{\displaystyle\bigwedge}
}\left(  \xi_{j}^{\prime\prime}\right)  _{re}\right)  \right)  _{re}$ and for
every infinite sequence of next and $\varepsilon$-reduction transitions%
\[
B_{\psi^{0}}\overset{\pi_{0}}{\rightarrow}B_{\psi^{1}}\overset{\varepsilon
}{\rightarrow}B_{\psi_{re}^{1}}\overset{\pi_{1}}{\rightarrow}B_{\psi^{2}%
}\overset{\varepsilon}{\rightarrow}B_{\psi_{re}^{2}}\ldots
\]
with $\psi^{0}=\xi_{re}^{\prime}\left(  \text{resp. }\psi^{0}=\left(  \xi
_{j}^{\prime}\right)  _{re},1\leq j\leq k\right)  $ and $v_{B_{\psi_{re}^{i}}%
}\left(  \psi^{i+1}\right)  \neq\mathbf{0}$ $\left(  i\geq0\right)  ,$ there
exists an infinite sequence of next and $\varepsilon$-reduction transitions
\[
B_{\lambda^{0}}\overset{\pi_{0}}{\rightarrow}B_{\lambda^{1}}\overset
{\varepsilon}{\rightarrow}B_{\lambda_{re}^{1}}\overset{\pi_{1}}{\rightarrow
}B_{\lambda^{2}}\overset{\varepsilon}{\rightarrow}B_{\lambda_{re}^{2}}\ldots
\]
with $\lambda^{0}=\xi_{re}^{\prime\prime}$ \ $\left(  \text{resp. }\lambda
^{0}=\left(  \xi_{j}^{\prime\prime}\right)  _{re},1\leq j\leq k\right)  $ and
$v_{B_{\lambda_{re}^{i}}}\left(  \lambda^{i+1}\right)  =v_{B_{\psi_{re}^{i}}%
}\left(  \psi^{i+1}\right)  $ for every $i\geq0.$
\end{lemma}

\begin{lemma}
\label{Lemma until}Let $\varphi=\psi U\xi$ with $\psi,\xi\in stLTL\left(
K,AP\right)  $. If $\mathcal{A}_{\psi},\mathcal{A}_{\xi}$ recognize $\Vert
\psi\Vert,\Vert\xi\Vert$ respectively, then $\mathcal{A_{\varphi}}$ recognizes
$\Vert\varphi\Vert.$
\end{lemma}

\begin{proof}
Let $w=\pi_{0}\pi_{1}\pi_{2}\ldots\in\left(  \mathcal{P}\left(  AP\right)
\right)  ^{\omega}$ and $\mathcal{A_{\psi}=}\left(  Q_{1},wt_{1}%
,I_{1},\mathcal{F}_{1}\right)  ,$ $\mathcal{A_{\xi}}=\left(  Q_{2}%
,wt_{2},I_{2},\mathcal{F}_{2}\right)  $, and $\mathcal{A_{\varphi}=}\left(
Q,wt,I,\mathcal{F}\right)  $. Let also
\[
P_{w}:B_{\varphi^{0}}\overset{\pi_{0}}{\rightarrow}B_{\varphi^{1}}%
\overset{\ast}{\rightarrow}B_{\varphi_{re}^{1}}\overset{\pi_{1}}{\rightarrow
}B_{\varphi^{2}}\overset{\ast}{\rightarrow}B_{\varphi_{re}^{2}}\ldots
\]
be a path in $next_{\mathcal{A_{\varphi}}}\left(  w\right)  $ with
$weight_{\mathcal{A_{\varphi}}}\left(  w\right)  \neq\mathbf{0}.$ Since
$P_{w}$ is successful there is an $l>0$ such that $B_{\varphi^{l}}\in
F_{\varphi}$ for the first time. We claim that there are paths $P_{w_{\geq j}%
}^{1}$ of $\mathcal{A}_{\psi}$ over $w_{\geq j}$, $0\leq j\leq l-1$, and a
path $P_{w_{\geq l-1}}^{1}$ of $\mathcal{A}_{\xi}$ over $w_{\geq l-1}$, that
are simultaneously simulated while $\mathcal{A_{\varphi}}$ runs $P_{w}$. This
is due to the following. Until the $l$th next transition the automaton moves
between states that are consistent sets of conjunctions containing $\varphi$.
After the $lth$ next transition the automaton moves between states that are
consistent sets of conjunctions not containing $\varphi$. For every $0<j\leq
l-1$, at the $j$th next transition of $P_{w}$ the choice of the next formula
of the maximal $\varphi$-consistent subset of the state indicates a next
transition of $\mathcal{A_{\psi}}$ that can be considered as the first
transition of a path of $\mathcal{A_{\psi}}$ over the suffix of $w$ starting
at this point. At the $l$th next transition this choice indicates a next
transition of $\mathcal{A_{\xi}}$ that can be considered as the first of a
path of $\mathcal{A_{\xi}}$ over $w_{\geq l-1}$.

Now, formally for $P_{w}$ we have that $\varphi^{0}=\varphi$, and the
following hold.

$\bullet$ For every $0<m<l$ there exist boolean formulas $\varphi^{\left(
m,1\right)  },\ldots,\varphi^{\left(  m,m\right)  }$ such that $\varphi
_{re}^{m}=\left(  \varphi\wedge\varphi_{re}^{\left(  m,1\right)  }\wedge
\ldots\wedge\varphi_{re}^{\left(  m,m\right)  }\right)  _{re}$ with

$\left(  i\right)  $ $\varphi\wedge\varphi^{\left(  m,1\right)  }\in
next\left(  M_{B_{\varphi_{re}^{m-1}},\varphi}\right)  $, and $\varphi
^{\left(  m,1\right)  }\in next\left(  M_{B_{\varphi_{re}^{m-1}},\psi}\right) $, and

$\left(  ii\right)  $ $\varphi^{\left(  m,p\right)  }\in next\left(
A^{\left(  m-1,p-1\right)  }\right)  $ for some $\varphi_{re}^{\left(
m-1,p-1\right)  }$-consistent set $A^{\left(  m-1,p-1\right)  }$ $\left(
2\leq p\leq m\right)  $. Moreover,%
\begin{align*}
&  wt\left(    B_{\varphi_{re}^{m-1}},\pi_{m-1},B_{\varphi^{m}}
\right) \\
&  =v_{M_{B_{\varphi_{re}^{m-1}},\psi}}\left(  \varphi^{\left(  m,1\right)
}\right)  \cdot v_{A^{\left(  m-1,1\right)  }}\left(  \varphi^{\left(
m,2\right)  }\right)  \cdot\ldots\cdot v_{A^{\left(  m-1,m-1\right)  }}\left(
\varphi^{\left(  m,m\right)  }\right) \\
&  =v_{M_{B_{\varphi_{re}^{m-1}},\psi}}\left(  \varphi^{\left(  m,1\right)
}\right)
\end{align*}
where last equality holds since $weight_{\mathcal{A_{\varphi}}}\left(
P_{w}\right)  \neq\mathbf{0}$, which implies that $wt\left(  B_{\varphi_{re}^{m-1}%
},\pi_{m-1},B_{\varphi^{m}}  \right)  \neq\mathbf{0}$, i.e.,
$v_{A^{\left(  m-1,j\right)  }}\left(  \varphi^{\left(  m,j+1\right)
}\right)  =\mathbf{1}$\textbf{ }for every $1\leq j\leq m-1.$ \footnote{Recall that $\varphi^{\left(m,1\right)}$ is boolean, as it is a next formula of an $LTL$-step formula. Then, $\varphi^{\left(m,j\right)}$ $\left(1\leq j\leq m-1\right)$ are boolean, since they are next formulas of boolean formulas.}

$\bullet$ $\varphi_{re}^{l}=\left(
\varphi_{re}^{\left(  l,1\right)  }\wedge\ldots\wedge\varphi_{re}^{\left(
l.l\right)  }\right)  _{re}$ for boolean formulas $\varphi^{\left(  l,1\right)  }%
,\ldots,\varphi^{\left(  l,l\right)  }$ with

$\left(  i\right)  $ $\varphi^{\left(  l,1\right)  }\in next\left(
M_{B_{\varphi_{re}^{l-1}},\xi}\right)  ,$ and

$\left(  ii\right)  $ $\varphi^{\left(  l,p\right)  }\in next\left(
A^{\left(  l-1,p-1\right)  }\right)  $ for some $\varphi_{re}^{\left(
l-1,p-1\right)  }$-consistent set $A^{\left(  l-1,p-1\right)  }$ $\left(
2\leq p\leq l\right)  $. Furthermore,%
\begin{align*}
&  wt\left(    B_{\varphi_{re}^{l-1}},\pi_{l-1},B_{\varphi^{l}}
\right) \\
&  =v_{M_{B_{\varphi_{re}^{l-1}},\psi}}\left(  \varphi^{\left(  l,1\right)
}\right)  \cdot v_{A^{\left(  l-1,1\right)  }}\left(  \varphi^{\left(
l,2\right)  }\right)  \cdot\ldots\cdot v_{A^{\left(  l-1,l-1\right)  }}\left(
\varphi^{\left(  l,l\right)  }\right) \\
&  =v_{M_{B_{\varphi_{re}^{l-1}},\psi}}\left(  \varphi^{\left(  l,1\right)
}\right)  .
\end{align*}

$\bullet$ Last, by induction on $n$ and the same arguments used in the proof
of Lemma \ref{boithitiko lemma 2} we get that for every $n>l$ there exist
boolean $\varphi^{\left(  n,1\right)  },\ldots,\varphi^{\left(  n,l\right)  }$
such that $\varphi_{re}^{n}=\left(  \varphi_{re}^{\left(  n,1\right)  }%
\wedge\ldots\wedge\varphi_{re}^{\left(  n,l\right)  }\right)  _{re}$ and
$\varphi^{\left(  n,p\right)  }\in next\left(  A^{\left(  n-1,p\right)
}\right)  $ for some $\varphi_{re}^{\left(  n-1,p\right)  }$-consistent set
$A^{\left(  n-1,p\right)  }$ $\left(  1\leq p\leq l\right)  ,$ and%
\begin{align*}
&  wt\left(   B_{\varphi_{re}^{n-1}},\pi_{n-1},B_{\varphi^{n}}
\right) \\
&  =v_{A^{\left(  n-1,1\right)  }}\left(  \varphi^{\left(  n,1\right)
}\right)  \cdot v_{A^{\left(  n-1,1\right)  }}\left(  \varphi^{\left(
n,2\right)  }\right)  \cdot\ldots\cdot v_{A^{\left(  n-1,l-1\right)  }}\left(
\varphi^{\left(  n,l\right)  }\right) \\
&  =\mathbf{1},
\end{align*}
where the last equality is concluded by the fact that
$weight_{\mathcal{A_{\varphi}}}\left(  P_{w}\right)  \neq\mathbf{0}$, i.e.,
$v_{A^{\left(  n-1,j\right)  }}\left(  \varphi^{\left(  n,j\right)  }\right)
=\mathbf{1}$\textbf{ }for every $1\leq j\leq l.$

For every $0<m<l$ the sequence%

\[
P_{w_{\geq m-1}}^{1}:M_{B_{\varphi_{re}^{m-1,\psi}}}\overset{\pi_{m-1}%
}{\rightarrow}B_{\varphi^{\left(  m,1\right)  }}\overset{\varepsilon
}{\rightarrow}A^{\left(  m,1\right)  }\overset{\pi_{m}}{\rightarrow}%
B_{\varphi^{\left(  m+1,2\right)  }}\overset{\varepsilon}{\rightarrow
}A^{\left(  m+1,2\right)  }\ldots
\]%
\[
A^{\left(  l-1,l-m\right)  }\overset{\pi_{l-1}}{\rightarrow}B_{\varphi
^{\left(  l,l-m+1\right)  }}\overset{\varepsilon}{\rightarrow}A^{\left(
l,l-m+1\right)  }\overset{\pi_{l}}{\rightarrow}B_{\varphi^{\left(
l+1,l-m+1\right)  }}\ldots
\]
is a path in $next_{\mathcal{A_{\psi}}}\left(  w_{\geq m-1}\right)  $ with%
\[
weight_{\mathcal{A_{\psi}}}\left(  P_{w_{\geq m-1}}^{1}\right)  =Val^{\omega
}\left(  v_{B_{\varphi_{re}^{m-1}},\psi}\left(  \varphi^{\left(  m,1\right)
}\right)  ,\mathbf{1},\mathbf{1},\ldots\right)
\]
and the sequence%

\[
P_{w_{\geq l-1}}^{2}:M_{B_{\varphi_{re}^{l-1},\xi}}\overset{\pi_{l-1}%
}{\rightarrow}B_{\varphi^{\left(  l,1\right)  }}\overset{\varepsilon
}{\rightarrow}A^{\left(  l,1\right)  }\overset{\pi_{l}}{\rightarrow}%
B_{\varphi^{\left(  l+1,1\right)  }}\overset{\varepsilon}{\rightarrow
}A^{\left(  l+1,1\right)  }\ldots
\]
is a path of $next_{ \mathcal{A}_{\xi}  }\left(  w_{\geq
l-1}\right)  $ with%
\[
weight_{  \mathcal{A}_{\xi}  }\left(  P_{w_{\geq l-1}}%
^{2}\right)  =Val^{\omega}\left(  v_{B_{\varphi_{re}^{l-1}},\xi}\left(
\varphi^{\left(  l,1\right)  }\right)  ,\mathbf{1},\mathbf{1},\mathbf{1}%
,\ldots\right)  .
\]
We note that for every $0\leq j\leq l-1$, and every $i\geq1$ the state
$B_{\varphi^{\left(  j+i,i\right)  }}$ appearing in the above paths could be
any non-empty $\varphi^{\left(  j+i,i\right)  }$-consistent set. We show that
$P_{w_{\geq m}}^{1},$ $0\leq m\leq l-1$, and $P_{w_{\geq l-1}}^{2}$ are
successful. Let us assume the contrary. Then, there exists a subformula of
$\varphi$ of the form $\xi^{\prime}U\xi^{\prime\prime}$ and an $n\geq l,$ such
that for every $r>n$, there is an $1\leq h\leq r$ such that $\varphi^{\left(  r,h\right)  }$ does not satisfy the acceptance
condition of $\mathcal{A}_{\xi}$ corresponding to $\xi^{\prime}U\xi
^{\prime\prime}$, or it does not satisfy the acceptance condition of $\mathcal{A}_{\psi}$ corresponding to $\xi^{\prime}U\xi
^{\prime\prime}$. But then $P_{w}$ would not be successful, which is a
contradiction. It holds

$weight_{\mathcal{A}_{\varphi}}\left(  P_{w}\right)  $\\\\

$=$$Val^{\omega}\left(
\begin{array}
[c]{c}%
v_{M_{B_{\varphi_{re}^{0}},\psi}}\left(  \varphi^{\left(  0,1\right)
}\right)  ,\ldots,v_{M_{B_{\varphi_{re}^{l-2}},\psi}}\left(  \varphi^{\left(
l-1,1\right)  }\right)  ,\\\\

v_{M_{B_{re}^{l-1},\xi}}\left(  \varphi^{\left(  l,1\right)  }\right),\mathbf{1},\mathbf{1},\ldots
\end{array}
\right)  $\\\\

$=Val^{\omega}\left(
\begin{array}
[c]{c}%
weight_{\mathcal{A_{\psi}}}\left(  P_{w_{\geq0}}^{1}\right)  ,\ldots
,weight_{\mathcal{A_{\psi}}}\left(  P_{w_{\geq l-2}}^{1}\right)  ,\\
weight_{\mathcal{A_{\xi}}}\left(  P_{w_{\geq l-1}}^{2}\right)  ,\mathbf{1}%
,\mathbf{1},\ldots
\end{array}
\right)  $\\\\

$\leq$$\underset{i\geq0}{\sum}$$\left(  \underset{k_{i}\in
pri_{\mathcal{A_{\xi}}}\left(  w_{\geq i}\right)  }{\underset{k_{j}\in
pri_{\mathcal{A_{\psi}}}\left(  w_{\geq j}\right)  \left(  0\leq j<i\right)
}{\sum}}Val^{\omega}\left(  k_{0},\ldots,k_{i-1},k_{i},\mathbf{1}%
,\mathbf{1},\ldots\right)  \right)  $\\\\

$=\underset{i\geq0}{\sum}\left(  Val^{\omega}\left(
\begin{array}
[c]{c}%
\underset{k_{0}\in pri_{\mathcal{A_{\psi}}}\left(  w_{\geq0}\right)  }{\sum
}k_{0},\ldots,\underset{k_{i-1}\in pri_{\mathcal{A_{\psi}}}\left(  w_{\geq
i-1}\right)  }{\sum}k_{i-1},\\
\underset{k_{i}\in pri_{\mathcal{A_{\xi}}}\left(  w_{\geq i}\right)  }{\sum
}k_{i},\mathbf{1},\mathbf{1},\ldots
\end{array}
\right)  \right)  $\\\\

$=\underset{i\geq0}{\sum}\left(  Val^{\omega}\left(  \left(  \left(
\Vert\mathcal{A_{\psi}\Vert},w_{\geq0}\right)  ,\ldots,\left(  \Vert
\mathcal{A_{\psi}\Vert},w_{\geq i-1}\right)  ,\left(  \Vert\mathcal{A}_{\xi
}\Vert,w_{\geq i}\right)  ,\mathbf{1},\mathbf{1},\ldots\right)  \right)
\right)  $\\\\

$=\underset{i\geq0}{\sum}$$\left(  Val^{\omega}\left(  \left(  \left(
\Vert\mathcal{\psi\Vert},w_{\geq0}\right)  ,\ldots,\left(  \Vert\psi
\Vert,w_{\geq i-1}\right)  ,\left(  \Vert\xi\Vert,w_{\geq i}\right)
,\mathbf{1},\mathbf{1},\ldots\right)  \right)  \right)  $\\\\

$=\left(  \Vert\varphi\Vert,w\right) . $\\

The second equality holds by Property \ref{Property 3}, the inequality by
Lemmas \ref{Properties Sum}iii, \ref{Lemma inequality sum}, and the third equality by the distributivity of $Val^{\omega}$ over finite sums, and the fact that $pri_{\mathcal{A_{\psi}}}\left(  w_{\geq j}\right)$ $\left (0\leq j \leq i-1\right)$, $pri_{\mathcal{A_{\xi}}}\left(  w_{\geq i}\right)$ are finite for all $i\geq0$. For every path
$P_{w}\in next_{\mathcal{A_{\varphi}}}\left(  w\right)  $ with
$weight_{\mathcal{A_{\varphi}}}\left(  P_{w}\right)  =\mathbf{0}$\textbf{ }it
trivially holds $weight_{\mathcal{A_{\varphi}}}\left(  P_{w}\right)
\leq\left(  \Vert\varphi\Vert,w\right)  $. Thus, for every $k\in
pri_{\mathcal{A_{\varphi}}}\left(  P_{w}\right)  ,k\leq\left(  \Vert
\varphi\Vert,w\right)  $, and so by Lemmas \ref{Properties Sum}ii,
\ref{Lemma inequality sum} we get%
\[
\left(  \Vert\mathcal{A_{\varphi}\Vert},w\right)  =\underset{k\in
pri_{\mathcal{A_{\varphi}}}\left(  w\right)  }{\sum}k\leq\left(  \Vert
\varphi\Vert,w\right)  .
\]

We show now that $\left(  \Vert\varphi\Vert,w\right)  \leq\left(
\Vert\mathcal{A}_{\varphi}\Vert,w\right)  $. To this end, we fix an $l\geq0$,
and we let $P_{w_{\geq m}}^{1}\in next$$_{A_{\psi}}\left(  w_{\geq m}\right)
$ for every $0\leq m<l$, and $P_{w_{\geq l}}^{2}\in next_{\mathcal{A_{\xi}}%
}\left(  w_{\geq l}\right)  $. We further assume that $weight_{\mathcal{A}%
_{\psi}}\left(  w_{\geq m}\right)  \neq\mathbf{0}$ $\left(  0\leq m<l\right)
\mathbf{,}$ $weight_{\mathcal{A}_{\xi}}\left(  w_{\geq l}\right)
\neq\mathbf{0.}$ We prove that there exists a path $P_{w}\in
next_{\mathcal{A_{\varphi}}}\left(  w\right)  $ such that%
\begin{align*}
&  Val^{\omega}\left(  weight_{\mathcal{A_{\psi}}}\left(  P_{w_{\geq0}}%
^{1}\right)  ,\ldots,weight_{\mathcal{A_{\psi}}}\left(  P_{w_{\geq l-1}}%
^{1}\right)  ,weight_{\mathcal{A_{\xi}}}\left(  P_{w_{\geq l}}^{2}\right)
,\mathbf{1},\mathbf{1},\ldots\right) \\
&  \leq weight_{\mathcal{A_{\varphi}}}\left(  P_{_{w}}\right)  .
\end{align*}
We set%
\[
P_{w_{\geq m}}^{1}:B_{\psi^{\left(  m,0\right)  }}\overset{\pi_{m}%
}{\rightarrow}B_{\psi^{\left(  m,1\right)  }}\overset{\ast}{\rightarrow
}B_{\psi_{re}^{\left(  m,1\right)  }}\overset{\pi_{m+1}}{\rightarrow}%
B_{\psi^{\left(  m,2\right)  }}\ldots
\]
and%
\[
P_{w_{\geq l}}^{2}:B_{\xi^{0}}\overset{\pi_{l}}{\rightarrow}B_{\xi^{1}%
}\overset{\ast}{\rightarrow}B_{\xi_{re}^{1}}\overset{\pi_{l+1}}{\rightarrow
}B_{\xi^{2}}\ldots
\]

For every $j\geq1,$ it holds $\psi^{\left(  m,j\right)  }$$\in\widehat
{next}\left(  B_{\psi_{re}^{\left(  m,j-1\right)  }}\right)  $ $\left(  0\leq
m\leq l-1\right)  $, and $\xi^{j}$$\in\widehat{next}$$\left(  B_{\xi
_{re}^{j-1}}\right)  $. We point out the following cases: (a) There is at
least one subformula $\varphi_{1}U\varphi_{2}\in cl\left(  \psi\right)  $ that is in
the scope of an always operator in $\psi$, and for at least two of the paths
$P_{w_{\geq m}}^{1}$$\left(  0\leq m\leq l-1\right)  $ the acceptance
condition from $\varphi_{1}U\varphi_{2}$ is satisfied for infinitely many
positions, and not satisfied for infinitely many positions, too. (b) There is
at least one $\varphi_{1}U\varphi_{2}\in cl\left(  \psi\right)  \cap cl\left(
\xi\right)  $ that is in the scope of an always operator in both $\psi,\xi$
and for at least two of the above $l+1$ paths the acceptance condition from
$\varphi_{1}U\varphi_{2}$ is satisfied for infinitely many positions, and not
satisfied for infinitely many positions, too.

First assume that cases (a) and (b) do not hold. Then, we set $\varphi
^{0}=\varphi$ and the following is true.

\begin{itemize}
\item For every $0\leq m<l,$ with the use of Lemma \ref{boithitiko until 1},
we obtain $B_{\varphi_{re}^{m}}\neq\emptyset,$ and $\varphi^{m+1}\in
next\left(  B_{\varphi_{re}^{m}}\right)  $ such that%
\begin{align*}
B_{\varphi_{re}^{m}}\left(  \varphi^{m+1}\right)   &  \geq\underset{0\leq
j\leq m}{\prod}v_{B_{\psi_{re}^{\left(  j,m-j\right)  }}}\left(  \psi^{\left(
j,m-j+1\right)  }\right) \\
&  =v_{B_{\psi_{re}^{\left(  m,0\right)  }}}\left(  \psi^{\left(  m,1\right)
}\right)
\end{align*}
and $\left(  B_{\varphi_{re}^{m}},\pi_{m},B_{\varphi^{m+1}}\right)  $ is a
next transition. More precisely, $\varphi_{re}^{m}=\left(  \varphi
\wedge\left(  \underset{0\leq j\leq m-1}{%
{\displaystyle\bigwedge}
}\psi^{\left(  j,m-j\right)  }\right)  \right)  _{re},$ and we apply Lemma
\ref{boithitiko until 1}, for the next transitions
$\left(  B_{\psi_{re}^{\left(  m,0\right)  }},\pi_{m},B_{\psi^{\left(
m,1\right)  }}\right)  $, and $\left(  B_{\psi
_{re}^{\left(  j,m-j\right)  }},\pi_{m},B_{\psi_{re}^{\left(  j,m-j+1\right)
}}\right) \\ \left(  0\leq j\leq m-1\right)  .$

\item By Lemma \ref{boithitiko until 2} we obtain $B_{\varphi_{re}^{l}}%
\neq\emptyset$ and $\varphi^{l+1}\in next\left(  B_{\varphi_{re}^{l}}\right)
$ such that $\left(  B_{\varphi_{re}^{l}},\pi_{l},B_{\varphi^{l+1}}\right)  $
is a next transition and%
\begin{align*}
v_{B_{\varphi_{re}^{l}}}\left(  \varphi^{l+1}\right)   &  \geq v_{B_{\xi^{0}}%
}\left(  \xi^{1}\right)  \cdot\underset{0\leq j\leq l-1}{\prod}v_{B_{\psi
_{re}^{\left(  j,l-j\right)  }}}\left(  \psi^{\left(  j,l-j+1\right)  }\right)
\\
&  =v_{B_{\xi^{0}}}\left(  \xi^{1}\right)  .
\end{align*}

It holds, $\varphi_{re}^{l}=\left(  \varphi\wedge\left(  \underset{0\leq
i\leq l-1}{%
{\displaystyle\bigwedge}
}\psi^{\left(  j,l-j\right)  }\right)  \right)  _{re},$ and we apply Lemma
\ref{boithitiko until 2}, for the next transitions $\left(  B_{\psi
_{re}^{\left(  j,l-j\right)  }},\pi_{m},B_{\psi_{re}^{\left(  j,l-j+1\right)
}}\right)  $ $\left(  0\leq j\leq l-1\right)  ,$ and $\left(  B_{\xi^{0}}%
,\pi_{l},B_{\xi^{1}}\right)  .$

\item Last, with the same arguments used in Lemma \ref{boithitiko lemma 3} we
obtain, for every $k>l,$ $B_{\varphi_{re}^{k}}\neq\emptyset$ and
$\varphi^{k+1}\in next$$\left(  B_{\varphi_{re}^{k}}\right)  $ such that
$\left(  B_{\varphi_{re}^{k}},\pi_{k},B_{\varphi^{k+1}}\right)  $ is a next
transition and%
\begin{align*}
v_{B_{\varphi_{re}^{k}}}\left(  \varphi^{k+1}\right)   &  \geq v_{B_{\xi
_{re}^{k-l}}}\left(  \xi^{k-l+1}\right)  \cdot\underset{0\leq j\leq l-1}%
{\prod}v_{B_{\psi_{re}^{\left(  j,k-j\right)  }}\left(  \psi^{\left(
j,k-j+1\right)  }\right)  }\\
&  =\mathbf{1}  ,
\end{align*}
where the last equality is obtained by the following. It holds $\psi,\xi\in
stLTL\left(  K,AP\right)  ,$ which implies that $\psi_{re}^{\left(
m,j\right)  },\xi_{re}^{j}\in bLTL\left(  K,AP\right)  $ for all $j\geq1, 0\leq m\leq
l-1$. Since $P_{w_{\geq m}}^{1}\left(  0\leq m\leq
l-1\right)  ,$ and $P_{w_{\geq l}}^{2}$ have non-zero weight, all but
the first next transitions appearing in $P_{w_{\geq m}}^{1}\left(  0\leq m\leq
l-1\right)  ,$ and $P_{w_{\geq l}}^{2}$ have weight $\mathbf{1.}$
\end{itemize}

Clearly, the path $P_{w}:B_{\varphi^{0}}\overset{\pi_{0}}{\rightarrow
}B_{\varphi^{1}}$$\overset{\varepsilon}{\rightarrow}B_{\varphi_{re}^{1}%
}\overset{\pi_{1}}{\rightarrow}B_{\varphi^{2}}\ldots$ (where we let
$B_{\varphi^{i}}$ be any non-empty $\varphi^{i}$-consistent set $\left(
i\geq0\right)  $) is a successful path of $\mathcal{A}_{\varphi}$ over $w$. This is concluded by the constructive proofs of Lemmas \ref{boithitiko until 1}, \ref{boithitiko until 2}, \ref{boithitiko lemma 3} and by the fact that (a), and (b) do not hold, which imply that for all $\varphi_{1}U\varphi_{2}$$\in cl\left(  \psi\right)  \bigcap cl\left(  \xi\right)  $, $P_{w}$ satisfies the acceptance condition for $\varphi_{1}U\varphi_{2}$ for infinitely many $i\geq 0$. It
holds%
\[
wt_{1}\left(    B_{\psi_{re}^{\left(  m,0\right)  }},\pi_{m}%
,B_{\psi^{\left(  m,1\right)  }}  \right)  \leq wt\left(
B_{\varphi_{re}^{m}},\pi_{m},B_{\varphi^{m+1}}\right)
\]
for every $0\leq m<l$, and%
\[
wt_{2}\left(   B_{\xi_{re}^{0}},\pi_{l},B_{\xi^{1}}
\right)  \leq wt\left(   B_{\varphi_{re}^{l}},\pi_{l},B_{\varphi^{l+1}%
}\right)
\]
and%
\[
\mathbf{1} \leq wt\left(   B_{\varphi_{re}^{k}},\pi_{k},B_{\varphi^{k+1}%
}\right)
\]
for every $k> l$. Hence, by the above relations, Lemma
\ref{Valuation inequality}, and Property \ref{Property 3} we get%
\begin{align*}
&  Val^{\omega}\left(  weight_{\mathcal{A_{\psi}}}\left(  P_{w_{\geq0}}%
^{1}\right)  ,\ldots,weight_{\mathcal{A_{\psi}}}\left(  P_{w_{\geq l-1}}%
^{1}\right)  ,weight_{\mathcal{A_{\xi}}}\left(  P_{w_{\geq l}}^{2}\right)
,\mathbf{1},\mathbf{1},\ldots\right) \\
&  \leq weight_{\mathcal{A_{\varphi}}}\left(  P_{w}\right)  .
\end{align*}

Now, if case (a) or (b) holds, we can prove our claim following the same
arguments used in the proof of Lemma \ref{Lemma conjuction}. Thus, for every
$l\geq0$, every $P_{w_{\geq m}}^{1}\in next_{\mathcal{A_{\psi}}}\left(
w_{\geq m}\right)  $, where $0\leq m<l,$ and every $P_{w_{\geq l}}^{2}\in
next_{\mathcal{A_{\xi}}}\left(  w_{\geq l}\right)  $, there exists a $P_{w}\in
next_{\mathcal{A_{\varphi}}}\left(  w\right)  $ such that%
\begin{align*}
&  Val^{\omega}\left(  weight_{\mathcal{A_{\psi}}}\left(  P_{w_{\geq0}}%
^{1}\right)  ,\ldots,weight_{\mathcal{A_{\psi}}}\left(  P_{w_{\geq l-1}}%
^{1}\right)  ,weight_{\mathcal{A_{\xi}}}\left(  P_{w_{\geq l}}^{2}\right)
,\mathbf{1},\mathbf{1},\ldots\right) \\
&  \leq weight_{\mathcal{A_{\varphi}}}\left(  P_{w}\right)  .
\end{align*}
Thus, it holds

\begin{align*}
\underset{k_{l}\in pri_{\mathcal{A_{\xi}}}\left(  w_{\geq l}\right)
}{\underset{k_{m}\in pri_{\mathcal{A_{\psi}}}\left(  w_{\geq m}\right)
}{\underset{l\geq0,0\leq m<l}{\sum}}}Val^{\omega}\left(  k_{0},\ldots
,k_{l-1},k_{l},\mathbf{1},\mathbf{1},\ldots\right)   &  \leq\underset{k\in
pri_{\mathcal{A}_{\varphi}}\left(  w\right)  }{\sum}k
\end{align*}
\begin{align*}
&  \Longrightarrow
\end{align*}
\begin{align*}
\underset{l\geq0}{\sum}Val^{\omega}\left(
\begin{array}
[c]{c}%
\underset{k_{0}\in pri_{\mathcal{A_{\psi}}}\left(  w_{\geq0}\right)  }{\sum
}k_{0},\ldots,\\
\underset{k_{l-1}\in pri_{\mathcal{A_{\psi}}}\left(  w_{\geq l-1}\right)
}{\sum}k_{l-1},\underset{k_{l}\in pri_{\mathcal{A_{\xi}}}\left(  w_{\geq
l}\right)  }{\sum}k_{l},\\
\mathbf{1},\mathbf{1},\mathbf{1},\ldots
\end{array}
\right)   &  \leq\underset{k\in pri_{\mathcal{A_{\varphi}}}\left(  w\right)
}{\sum}k
\end{align*}
\begin{align*}
&  \Longrightarrow
\end{align*}
\begin{align*}
\underset{l\geq0}{\sum}Val^{\omega}\left(
\begin{array}
[c]{c}%
\left(  \Vert\psi\Vert,w_{\geq0}\right)  ,\ldots,\left(  \Vert\psi
\Vert,w_{\geq l-1}\right)  ,\left(  \Vert\varphi\Vert,w_{\geq l}\right)  ,\\
\mathbf{1},\mathbf{1},\mathbf{1},\ldots
\end{array}
\right)   &  \leq\underset{k\in pri_{\mathcal{A_{\varphi}}}\left(  w\right)
}{\sum}k
\end{align*}
\begin{align*}
&  \Longrightarrow
\end{align*}
\begin{align*}
\left(  \Vert\varphi\Vert,w\right)   &  \leq\left(  \mathcal{\Vert A_{\varphi
}\Vert},w\right)  ,
\end{align*}
where the second inequality is obtained by the distributivity of $Val^{\omega}$ over finite sums, and the fact that $pri_{\mathcal{A_{\psi}}}\left(  w_{\geq m}\right)  $ $\left(0\leq m<l\right) $, and $pri_{\mathcal{A_{\xi}}}\left(  w_{\geq l}\right)  $ are finite, and this concludes our proof.
\end{proof}

\begin{lemma}
\label{Lemma always}Let $\psi\in stLTL\left(  K,AP\right)  $ such that
$\varphi=\square\psi.$ If $\mathcal{A_{\psi}}$ recognizes $\Vert\psi\Vert,$
then $\mathcal{A_{\varphi}}$ recognizes $\Vert\varphi\Vert$.
\end{lemma}

\begin{proof}
Let $\mathcal{A_{\psi}=}$$\left(  Q^{\prime},wt^{\prime},I^{\prime
},\mathcal{F^{\prime}}\right)  $, $\mathcal{A_{\varphi}=}\left(
Q,wt,I,\mathcal{F}\right)  $. First, we prove that $\left(  \Vert
\mathcal{A_{\varphi}\Vert},w\right)  \leq\left(  \Vert\varphi\Vert,w\right)  $
for every $w\in\left(  \mathcal{P}\left(  AP\right)  \right)  ^{\omega}$. To
this end, let $w=\pi_{0}\pi_{1}\pi_{2}\ldots$ and $P_{w}\in
next_{\mathcal{A_{\varphi}}}\left(  w\right)  $ be a path with
$weight_{\mathcal{A_{\varphi}}}\left(  P_{w}\right)  \neq\mathbf{0}$. We show
that there exist paths $P_{w_{\geq i}}^{\prime}\in next_{\mathcal{A_{\psi}}%
}\left(  w_{\geq i}\right)  $ $\left(  i\geq0\right)  $ such that
\[
weight_{\mathcal{A_{\varphi}}}\left(  P_{w}\right)  \leq Val^{\omega}\left(
weight_{\mathcal{A_{\psi}}}\left(  P_{w_{\geq i}}^{\prime}\right)  \right)
_{i\geq0}%
\]

Without any loss we may assume that $P_{w}$ starts with a next transition. So
we let
\[
P_{w}:B_{\varphi^{0}}\overset{\pi_{0}}{\rightarrow}B_{\varphi^{1}}%
\overset{\ast}{\rightarrow}B_{\varphi_{re}^{1}}\overset{\pi_{1}}{\rightarrow
}B_{\varphi^{2}}\overset{\ast}{\rightarrow}B_{\varphi_{re}^{2}}\overset
{}{\rightarrow}\ldots
\]

It holds $\varphi^{0}=\varphi_{re}^{0}=\varphi$ and for every $i\geq1$, we can
prove by induction on $i$ and the same arguments used in Lemma
\ref{boithitiko lemma 3}, that there exist boolean formulas $\varphi^{\left(
i,1\right)  },\ldots,\varphi^{\left(  i,i\right)  }$ such that $\varphi
_{re}^{i}=\left(  \varphi\wedge\varphi_{re}^{\left(  i,1\right)  }\wedge
\ldots\wedge\varphi_{re}^{\left(  i,i\right)  }\right)  _{re}$, and
$\varphi_{re}^{\left(  i,1\right)  }\in next\left(  M_{B_{\varphi_{re}^{i-1}%
},\psi}\right)  ,\varphi^{\left(  i,p\right)  }\in next\left(  A^{\left(
i-1,p-1\right)  }\right)  $ where $A^{\left(  i-1,p-1\right)  }$ is a
$\varphi_{re}^{\left(  i-1,p-1\right)  }$-consistent set $\left(  2\leq p\leq
i\right)  $, and
\begin{align}
wt\left(    B_{\varphi_{re}^{i-1}},\pi_{i-1},B_{\varphi^{i}}
\right)   &  =v_{M_{B_{\varphi_{re}^{i-1},\psi}}}\left(  \varphi^{\left(
i,1\right)  }\right)  \cdot\underset{2\leq p\leq i}{\prod}v_{A^{\left(
i-1,p-1\right)  }}\left(  \varphi^{\left(  i,p\right)  }\right)
\label{always 1}\\
&  =v_{M_{B_{\varphi_{re}^{i-1},\psi}}}\left(  \varphi^{\left(  i,1\right)
}\right) \nonumber
\end{align}
where the last equality holds since $v_{A^{\left(  i-1,1\right)  }}\left(
\varphi^{\left(  i,2\right)  }\right)  =\ldots=v_{A^{\left(  i-1,i-1\right)
}}\left(  \varphi^{\left(  i,i\right)  }\right)  =\mathbf{1}.$

Hence, for every $i\geq0$ we can define the path $P_{w_{\geq i}}^{\prime}\in
next_{\mathcal{A_{\psi}}}\left(  w_{\geq i}\right)  $ as follows.
\[
P_{w_{\geq i}}^{\prime}:M_{B_{\varphi_{re}^{i},\psi}}\overset{\pi_{i}%
}{\rightarrow}B_{\varphi^{\left(  i+1,1\right)  }}\overset{\varepsilon
}{\rightarrow}A^{\left(  i+1,1\right)  }\overset{\pi_{i+1}}{\rightarrow
}B_{\varphi^{\left(  i+2,2\right)  }}\overset{\varepsilon}{\rightarrow
}A^{\left(  i+2,2\right)  }\ldots
\]
where for every $j\geq1,$ we let $B_{\varphi^{\left(  i+j,j\right)  }}$ be any
non-empty $\varphi^{\left(  i+j,j\right)  }$-consistent set. We show that
$P_{w_{\geq i}}^{\prime}$ is successful. Let us assume the contrary. This
means that there exists a boolean subformula of $\psi$ of the form $\xi
U\xi^{\prime}$ and an $n\geq0$, such that for every $r>n$, there is an $1\leq
h\leq r$ such that $\varphi^{\left(  r,h\right)  }$ does not satisfy the
acceptance condition of $\mathcal{A_{\psi}}$ corresponding to $\xi^{\prime
}U\xi^{\prime\prime}$. But then $P_{w}$ would not be successful, which is a
contradiction. Moreover, it holds%
\begin{align*}
&  weight_{\mathcal{A_{\psi}}}\left(  P_{w_{\geq i}}^{\prime}\right) \\
&  =Val^{\omega}\left(  v_{M_{B_{\varphi_{re}^{i}},\psi}}\left(
\varphi^{\left(  i+1,1\right)  }\right)  ,v_{A^{\left(  i+1,1\right)  }%
}\left(  \varphi^{\left(  i+2,2\right)  }\right)  ,v_{A^{\left(  i+2,2\right)
}}\left(  \varphi^{\left(  i+3,3\right)  }\right)  ,\ldots\right) \\
&  =Val^{\omega}\left(  v_{M_{B_{\varphi_{re}^{i},\psi}}}\left(
\varphi^{i+1,1}\right)  ,\mathbf{1},\mathbf{1},\ldots\right) \\
&  =v_{M_{B_{\varphi_{re}^{i},\psi}}}\left(  \varphi^{i+1,1}\right)  ,
\end{align*}
where the last equality holds by Property \ref{Property 3}. Then,%
\begin{align*}
weight_{\mathcal{A_{\varphi}}}\left(  P_{w}\right)   &  =Val^{\omega}\left(
weight_{\mathcal{A_{\psi}}}\left(  P_{w_{\geq0}}^{\prime}\right)
,weight_{\mathcal{A_{\psi}}}\left(  P_{w_{\geq1}}^{\prime}\right)
,\ldots\right) \\
&  \leq Val^{\omega}\left(  \underset{k_{0}\in pri_{\mathcal{A_{\psi}}}\left(
w_{\geq0}\right)  }{\sum}k_{0},\ldots\right) \\
&  =Val^{\omega}\left(  \left(  \Vert\psi\Vert,w_{\geq0}\right)  ,\left(
\Vert\psi\Vert,w_{\geq1}\right)  ,\ldots\right) \\
&  =\left(  \Vert\varphi\Vert,w\right)  ,
\end{align*}
where the first inequality is concluded by Lemmas \ref{Properties Sum}iii,
\ref{Valuation inequality}.

Hence, for every $k\in pri_{\mathcal{A_{\varphi}}}\left(  w\right)  $ it holds
$k\leq\left(  \Vert\varphi\Vert,w\right)  $, and thus using Lemmas
\ref{Properties Sum}ii, \ref{Valuation inequality} we get $\left(
\Vert\mathcal{A}_{\varphi}\Vert,w\right)  =\underset{k\in
pri_{\mathcal{A_{\varphi}}}\left(  w\right)  }{\sum}k$$\leq$$\left(
\Vert\varphi\Vert,w\right)  $. We show now that $\left(  \Vert\varphi
\Vert,w\right)  \leq\left(  \Vert\mathcal{A}_{\varphi}\Vert,w\right)  $. To
this end, we let $P_{w_{\geq i}}^{\prime}\in next_{\mathcal{A_{\psi}}}\left(
w_{\geq i}\right)  \left(  i\geq0\right)  .$ We will prove that there exists a
$P_{w}\in next_{\mathcal{A_{\varphi}}}\left(  w\right)  $ with
\[
Val^{\omega}\left(  weight_{\mathcal{A_{\psi}}}\left(  P_{w_{\geq0}}^{\prime
}\right)  ,weight_{\mathcal{A_{\psi}}}\left(  P_{w_{\geq1}}^{\prime}\right)
,\ldots\right)  \leq weight_{\mathcal{A_{\varphi}}}\left(  P_{w}\right)  .
\]

If $Val^{\omega}\left(  weight_{\mathcal{A_{\psi}}}\left(  P_{w_{\geq0}%
}^{\prime}\right)  ,weight_{\mathcal{A_{\psi}}}\left(  P_{w_{\geq1}}^{\prime
}\right)  ,\ldots\right)  =\mathbf{0}$, then the inequality holds for every
$P_{w}\in next_{\mathcal{A_{\varphi}}}\left(  w\right)  $. Otherwise, no empty
states appear in $P_{w_{\geq i}}^{\prime}\left(  i\geq0\right)  $ and the
subsequent hold.

Let $\varphi_{1}U\varphi_{2}\in cl\left(  \psi\right)  $. There exist
paths $\widehat{P}_{w_{\geq i}}\in next_{\mathcal{A_{\psi}}}\left(  w_{\geq
i}\right)  $$\left(  i\geq1\right)  $\footnote{We can prove for every $i\geq1$ the existence of the path $\widehat{P}_{w_{\geq i}}$ following the constructive arguments of Procedure 3.} with the following properties. (a) There
are infinitely many $j\geq1$ such that at the next transition of $\widehat
{P}_{w_{\geq k}}\left(  1\leq k\leq j\right)  $ that processes the letter
$\pi_{j}$, and at the corresponding next transition of $P_{w_{\geq0}}^{\prime
}$, the automaton moves to a state that satisfies the acceptance condition of
$\varphi_{1}U\varphi_{2}$, and (b)\\

$Val^{\omega}\left(  weight_{\mathcal{A_{\psi}}}\left(  P_{w_{\geq0}}^{\prime
}\right)  ,weight_{\mathcal{A_{\psi}}}\left(  P_{w_{\geq1}}^{\prime}\right)
,\ldots\right)  $

$\leq Val^{\omega}\left(  weight_{\mathcal{A_{\psi}}}\left(  P_{w_{\geq0}%
}^{\prime}\right)  ,weight_{\mathcal{A_{\psi}}}\left(  \widehat{P}_{w_{\geq1}%
}\right)  ,weight_{\mathcal{A_{\psi}}}\left(  \widehat{P}_{w_{\geq2}}\right)
,\ldots\right)  $\\

Moreover, the paths $\widehat{P}_{w_{\geq i}}$$\left(  i\geq1\right)  $ can be
chosen so that condition (a) is satisfied for every $\varphi_{1}U\varphi
_{2}\in cl\left(  \psi\right)  .$ We set $\widehat{P}_{w_{\geq0}}=P_{w_{\geq
0}}^{\prime}$, and for every $i\geq0$ we let

$\widehat{P}_{w_{\geq i}}:B_{\psi^{\left(  i,0\right)  }}\overset{\pi_{i}%
}{\rightarrow}B_{\psi^{\left(  i,1\right)  }}$$\overset{\ast}{\rightarrow
}B_{\psi_{re}^{\left(  i,1\right)  }}\overset{\pi_{i+1}}{\rightarrow}%
B_{\psi^{\left(  i,2\right)  }}\ldots$.

Clearly, for every $j\geq1$ it holds $\psi^{\left(  i,j\right)  }\in
\widehat{next}\left(  B_{\psi_{re}^{\left(  i,j-1\right)  }}\right)  $$\left(
i\geq0\right)  $. Then, we set $\varphi^{0}=\varphi$ and with the same
procedure used in Lemma \ref{boithitiko until 1} we obtain for every $m\geq0$,
a $B_{\varphi_{re}^{m}}\neq\emptyset$ and $\varphi^{m+1}\in next\left(
B_{\varphi_{re}^{m}}\right)  $, such that%
\[%
\begin{array}
[c]{ll}%
v_{B_{\varphi_{re}^{m}}}\left(  \varphi^{m+1}\right)  & \geq\underset{0\leq
j\leq m}{%
{\displaystyle\prod}
}v_{B_{\psi_{re}^{\left(  j,m-j\right)  }}}\left(  \psi^{\left(
j,m-j+1\right)  }\right) \\
& =\underset{0\leq j\leq m}{%
{\displaystyle\prod}
}wt^{\prime}\left(   B_{\psi_{re}^{\left(  j,m-j\right)  }},\pi
_{m},B_{\psi^{\left(  j,m-j+1\right)  }}\right)   \\
& =wt^{\prime} \left(  B_{\psi_{re}^{\left(  m,0\right)  }},\pi
_{m},B_{\psi^{\left(  m,1\right)  }}\right)   \\
& =weight_{\mathcal{A}_{\psi}}\left(  \widehat{P}_{w_{\geq m}}\right)  .
\end{array}
\]

Then, the path
\[
P_{w}:B_{\varphi^{0}}\overset{\pi_{0}}{\rightarrow}B_{\varphi^{1}}%
\overset{\varepsilon}{\rightarrow}B_{\varphi_{re}^{1}}\overset{\pi_{1}%
}{\rightarrow}B_{\varphi^{2}}\ldots
\]

(where we let $B_{\varphi^{1}}$ be any non-empty $\varphi^{i}$-consistent set
$\left(  i\geq0\right)  $) is a successful path of next and $\varepsilon
$-reduction transitions of $\mathcal{A_{\varphi}}$ over $w$ and it holds
\begin{align*}
Val^{\omega}\left(  weight_{\mathcal{A}_{\psi}}\left(  \widehat{P}_{w_{\geq
i}}\right)  \right)  _{i\geq0}  &  \leq Val^{\omega}\left(  v_{B_{\varphi
_{re}^{i}}}\left(  \varphi^{i+1}\right)  \right)  _{i\geq0}\\
&  =weight_{\mathcal{A}_{\varphi}}\left(  P_{w}\right)  .
\end{align*}
Thus, for every family $P_{w_{\geq i}}^{\prime}\in next_{\mathcal{A}_{\psi}%
}\left(  w_{\geq i}\right)  $ $\left(  i\geq0\right)  $ there exists a
$P_{w}\in next_{\mathcal{A}_{\varphi}}\left(  w\right)  $ such that
\[
Val^{\omega}\left(  weight_{\mathcal{A}_{\psi}}\left(  P_{w_{\geq i}}^{\prime
}\right)  \right)  _{i\geq0}\leq weight_{\mathcal{A}_{\varphi}}\left(
P_{w}\right)
\]
i.e., for every family $k_{i}\in pri_{\mathcal{A}_{\psi}}\left(  w_{\geq
i}\right)  \left(  i\geq0\right)  $ there exists a $k\in pri_{\mathcal{A}%
_{\varphi}}\left(  w\right)  $ such that $Val^{\omega}\left(  k_{i}\right)
_{i\geq0}\leq k$. Then, by Lemma \ref{Lemma inequality sum}\ we get
\[
\underset{i\geq0}{\underset{k_{i}\in pri_{\mathcal{A}_{\psi}}\left(  w_{\geq
i}\right)  }{\sum}}Val^{\omega}\left(  k_{i}\right)  _{i\geq0}\leq
\underset{pri_{\mathcal{A}_{\varphi}}\left(  w\right)  }{%
{\displaystyle\sum}
}k=\left(  \left\Vert \mathcal{A}_{\varphi}\right\Vert ,w\right)
\]
Moreover,
\begin{align*}
\underset{i\geq0}{\underset{k_{i}\in pri_{\mathcal{A}_{\psi}}\left(  w_{\geq
i}\right)  }{\sum}}Val^{\omega}\left(  k_{i}\right)  _{i\geq0}  &
=Val^{\omega}\left(  \underset{k_{i}\in pri_{\mathcal{A}_{\psi}}\left(
w_{\geq i}\right)  }{\sum}k_{i}\right)  _{i\geq0}\\
&  =Val^{\omega}\left(  \left\Vert \psi\right\Vert ,w_{\geq i}\right)
_{i\geq0}\\
&  =\left(  \left\Vert \varphi\right\Vert ,w\right)
\end{align*}
where the first equality holds by the distributivity of $Val^{\omega}$ over finite sums and the fact that for every
$i\geq0$ the set $pri_{\mathcal{A}_{\psi}}\left(  w_{\geq i}\right)  $ is
finite. We conclude that $\left(  \left\Vert \varphi\right\Vert ,w\right)
\leq\left(  \left\Vert \mathcal{A}_{\varphi}\right\Vert ,w\right)  .$ Hence,
for every $w\in\left(  \mathcal{P}\left(  AP\right)  \right)  ^{\omega}$ we
have $\left(  \left\Vert \mathcal{A}_{\varphi}\right\Vert ,w\right)  =\left(
\left\Vert \varphi\right\Vert ,w\right)  ,$ and the proof is completed.
\end{proof}

\begin{lemma}
\label{Lemma boolean}Let $\varphi\in bLTL\left(  K,AP\right)  .$ Then,
$\mathcal{A}_{\varphi}$ recognizes $\left\Vert \varphi\right\Vert .$
\end{lemma}

\begin{proof}
We prove our claim by induction on the structure of $bLTL\left(  K,AP\right)
$-formulas and using the same arguments as in Lemmas \ref{Lemma atomic},
\ref{Lemma disjuction}, \ref{Lemma next}, \ref{Lemma conjuction},
\ref{Lemma until}, \ref{Lemma always}.
\end{proof}

\begin{lemma}
\label{Lemma step}Let $\varphi\in stLTL\left(  K,AP\right)  .$ Then,
$\mathcal{A}_{\varphi}$ recognizes $\left\Vert \varphi\right\Vert .$
\end{lemma}

\begin{proof}
Our claim is derived by Lemmas \ref{Lemma boolean}, \ref{Lemma atomic},
\ref{Lemma disjuction}.
\end{proof}

\begin{theorem}
Let $\varphi\in RULTL\left(  K,AP\right)  .$ Then, $\mathcal{A}_{\varphi}$
recognizes $\left\Vert \varphi\right\Vert .$
\end{theorem}

\begin{proof}
By Lemmas \ref{Lemma atomic}, \ref{Lemma next}, \ref{Lemma disjuction},
\ref{Lemma conjuction}, \ref{Lemma step}, \ \ref{Lemma until},
\ref{Lemma always}, \ref{Lemma boolean} we get that $\left\Vert \mathcal{A}%
_{\varphi}\right\Vert =\left\Vert \varphi\right\Vert $ for every $\varphi\in
RULTL\left(  K,AP\right)  .$
\end{proof}

\begin{corollary}
Let $\varphi\in RULTL\left(  K,AP\right)  .$ Then, we can effectively
construct a wBa over $\mathcal{P}\left(  AP\right)  $ and $K$ recognizing
$\left\Vert \varphi\right\Vert .$
\end{corollary}

\section{Weighted LTL over generalized product $\omega$-valuation monoids}

\bigskip We let $AP$ be a finite set of atomic propositions and $K=\left(
K,+,\cdot,Val^{\omega},\mathbf{0},\mathbf{1}\right)  $ be an indempotent
generalized product $\omega$-valuation monoid. The syntax and semantics over
the weighted $LTL$ over $AP$ and $K$ is defined as in the previous section, as
well as the class $LTL(K,AP)$, and the fragment of $bLTL\left(  K,AP\right)  $. We let a \textit{restricted}
$LTL$\textit{-step formula }be an $LTL\left(  K,AP\right)  $-formula of the
form $\underset{1\leq i\leq n}{%
{\displaystyle\bigvee}
}\left(  k_{i}\wedge\varphi_{i}\right)  $ with $k_{i}\in K\backslash\left\{
\mathbf{0},\mathbf{1}\right\}  $ and $\varphi_{i}\in bLTL\left(  K,AP\right)
$ for every $1\leq i\leq n.$ We denote by $r$-$stLTL\left(  K,AP\right)  $ the
class of restricted $LTL$-step formulas over $AP$ and $K.$ We introduce now
the syntactic fragment of \textit{totally restricted} $U$-\textit{nesting}
$LTL$-formulas.

\begin{definition}
The fragment of totally \textit{restricted} $U$-\textit{nesting}
$LTL$-formulas over $AP$ and $K$, denoted by $t$-$RULTL\left(  K,AP\right)  $,
is the least class of formulas in $LTL\left(  K,AP\right)  $ which is defined
inductively in the following way.

\begin{itemize}
\item[$\cdot$] $k\in t$-$RULTL\left(  K,AP\right)  $ for every $k\in K.$

\item[$\cdot$] $bLTL\left(  K,AP\right)  \subseteq t$-$RULTL\left(
K,AP\right)  $

\item[$\cdot$] If $\varphi\in t$-$RULTL\left(  K,AP\right)  ,$ then
$\bigcirc\varphi\in t$-$RULTL\left(  K,AP\right)  .$\

\item[$\cdot$] If $\varphi,\psi\in t$-$RULTL\left(  K,AP\right)  ,$ then
$\varphi\vee\psi\in t$-$RULTL\left(  K,AP\right)  .$

\item[$\cdot$] If $\varphi\in bLTL\left(  K,AP\right)  $ and $\psi\in
r$-$stLTL\left(  K,AP\right)  ,$

or $\psi\in bLTL\left(  K,AP\right)  ,$ or $\psi=\xi U\lambda,$ or $\psi
=\square\xi$

where $\xi,\lambda\in r$-$stLTL\left(  K,AP\right)  ,$

then $\varphi\wedge\psi,\psi\wedge\varphi\in t$-$RULTL\left(  K,AP\right)  .$

\item[$\cdot$] If $\varphi,\psi\in r$-$stLTL\left(  K,AP\right)  ,$ then
$\varphi U\psi\in t$-$RULTL\left(  K,AP\right)  .$

\item[$\cdot$] If $\varphi\in r$-$stLTL\left(  K,AP\right)  ,$ then
$\square\varphi\in t$-$RULTL\left(  K,AP\right)  .$
\end{itemize}
\end{definition}

We adopt the theory of the previous section (observe that using induction on the
structure of $\varphi$ we can derive that for every $\varphi\in
t$-$RULTL\left(  K,AP\right)  $ and $\varphi$-consistent set
$B_{\varphi}$, $\ $it holds $next\left(  B_{\varphi}\right)  \subseteq
t$-$RULTL\left(  K,AP\right)  $). The following theorem is obtained by induction on the structure of
$t$-$RULTL\left(  K,AP\right)  $-formulas and using the same arguments as the ones used
in Lemmas \ref{Lemma atomic}, \ref{Lemma next}, \ref{Lemma disjuction},
\ref{Lemma conjuction}, \ref{Lemma step}, \ref{Lemma until},
\ref{Lemma always}, \ref{Lemma boolean} of the previous section. The stronger syntactical restriction that we impose on the fragment of \textit{totally restricted} $U$-\textit{nesting}
$LTL$-formulas allows us the use of the distributivity of $Val^{\omega}$ over finite sums for generalized product $\omega$-valuation monoids and Lemma 3 whenever necessary.

\begin{theorem}
Let $\varphi\in t$-$RULTL\left(  K,AP\right)  .$ Then, $\mathcal{A}_{\varphi}$
recognizes $\left\Vert \varphi\right\Vert .$
\end{theorem}

\begin{example}
Let $AP=\left\{  a,b\right\}  ,$ and $\varphi
=\square\left(  a\wedge2\right)  $$\in t$-$RULTL\left(  K,AP\right) $ where $K$ is the generalized product $\omega$-valuation monoid of Example \ref{example2}. Then, $\mathcal{A}_{\varphi}=\left(
Q,wt,in,\mathcal{F}\right)  $ is defined below, where $\pi$
ranges over $\mathcal{P}\left(  AP\right)  $, and by $\pi_{a}$ we denote any
letter in $\mathcal{P}\left(  AP\right)  $ that contains $a.$

\begin{itemize}
\item[$\bullet$] $Q=\left\{  q_{1},\ldots,q_{5}\right\}  $ with $q_{1}%
=\emptyset,$ $q_{2}=\left\{  -\infty\right\}  ,$ $q_{3}=\left\{  true\right\}
,$$q_{4}=\left\{  \varphi,a\wedge2,a,2\right\}  ,$ $q_{5}=\left\{
\varphi\wedge\left(  true\wedge true\right)  ,\varphi,a\wedge2,2,a,true\wedge
true,true\right\}  $

\item[$\bullet$] The states with initial weight $\infty$ are the sets $q_{1}%
,q_{4}.$

\item[$\bullet$] The transitions with weight$\neq-\infty$ are the following:

$wt\left(  q_{3},\pi,q_{3} \right)  =wt\left(
q_{5},\varepsilon,q_{4}\right)   =wt\left(    q_{k}%
,\varepsilon,q_{k}\right)    =\infty$ where $k=2,3,4$, and

$wt\left(    q_{4},\pi_{a},q_{5}  \right)  =2.$

\item[$\bullet$] The automaton has no final sets since $\varphi$ contains no
$U$ operators.
\end{itemize}
\end{example}

\begin{corollary}
Let $\varphi\in t$-$RULTL\left(  K,AP\right)  .$ Then, we can effectively
construct a wBa over $\mathcal{P}\left(  AP\right)  $ and $K$ recognizing
$\left\Vert \varphi\right\Vert .$
\end{corollary}

\bigskip

\newpage
\section{Conclusion}
In this paper we introduced a weighted \textit{LTL} over product $\omega
$-valuation monoids (resp. generalized product $\omega$-valuation monoids) that satisfy specific properties, and proved that for
every formula $\varphi$ of a syntactic fragment of the weighted $\mathit{LTL}$
we can effectively construct a weighted generalized Büchi automaton with
$\varepsilon$-transitions $\mathcal{A}_{\varphi\text{ }}$ whose behavior
coincides with the semantics of $\varphi.$ We provided
in this way a theoretical basis for the definition of quantitative
model-checking algorithms. The structure of product $\omega$-valuation
monoids and generalized product $\omega$-valuation monoids, that was used for the domain of weights, refers to an interesting
range of possible applications. Naturally, in order to reach the goal of
quantitative reasoning it is necessary to further investigate complexity and
decidability results, providing in this way more arguments for the definition
of model-checking algorithms incorporating the proposed weighted \textit{LTL.}
In \cite{Ha-Ku} the authors introduced the notion of safety in the weighted
setting. More precisely, for a rational number $q$, a finite series $s$ over a
given alphabet and $%
\mathbb{Q}
$ is called $q$-safe if every word with coefficient at least $q$ has a prefix
all whose extensions have coefficient at least $q$. Given a deterministic
weighted automaton, the authors relate the safety of its behavior with its
structure. They also propose the extension of their theory to infinite words
as a challenging perspective and we further add that the definition of the
notion of safety for infinitary series could be related with syntactical
fragments of the weighted \textit{LTL} and the structural properties of the
corresponding weighted generalized Büchi automaton with $\varepsilon
$-transitions that we propose in this paper. Finally, another interesting road
for extending the theory of our weighted \textit{LTL}, is to study its
relation with weighted \textit{FO} logic, $\omega$-star-free series and
weighted counter-free automata on infinite words.
\newpage


\newpage

\textbf{Appendix.}

In the following proof, we prove that the structure presented in Example \ref{example2} is indeed a generalized product $\omega$-valuation monoid.\\
\begin{proof}
We prove first the distributivity of $Val^{\omega}$ over finite sums for generalized product $\omega $-valuation monoids.

Let $L\subseteq _{fin}\overline{%
\mathbb{R}
},$ finite index sets $I_{j}\left(j\geq 0\right),$ and $k_{i_{j}}\in L\left( i_{j}\in
I_{j},j\geq 0\right) $ such that for all but a finite number of $j\geq 0,$ it holds $%
k_{i_{j}}\in L\backslash \left\{ \infty ,\mathbf{-\infty }\right\} $ for all
$i_{j}\in I_{j},$ or $k_{i_{j}}\in \left\{ \mathbf{\infty },\mathbf{-\infty }%
\right\} $ for all $i_{j}\in I_{j}$. We will prove that
\begin{equation*}
\emph{liminf}\left( \underset{i_{j}\in I_{j}}{\sup }k_{i_{j}}\right) _{j\in
\mathbb{N}
}=\underset{\left( i_{j}\right) _{j}\in I_{0}\times I_{1}\times \ldots }{%
\sup }\left( \emph{liminf}\left( k_{i_{j}}\right) _{j\in
\mathbb{N}
}\right) .
\end{equation*}%
We set $A=$\emph{liminf}$\left( \underset{i_{j}\in I_{j}}{\sup }k_{i_{j}}\right)
_{j\in
\mathbb{N}
},$ and $B=\underset{\left( i_{j}\right) _{j}\in I_{0}\times I_{1}\times
\ldots }{\sup }\left( \emph{liminf}\left( k_{i_{j}}\right) _{j\in
\mathbb{N}
}\right) .$

Assume that there exists an $l\geq 0,$ such that $k_{i_{l}}=-\infty $ for all
$i_{l}\in I_{l}$. Then, for all $\left( i_{j}\right) _{j}\in I_{0}\times
I_{1}\times \ldots ,$ \emph{liminf}$\left( k_{i_{j}}\right) _{j\in
\mathbb{N}
}=-\infty ,$ i.e., $B=-\infty .$ Moreover, $\underset{i_{l}\in I_{l}}{\sup }%
k_{i_{l}}=-\infty ,$ and thus $A=-\infty $ as wanted. Otherwise, we point
out the following cases:

(I) Assume that for all $j\geq 0$ there exists $i_{j}\in I_{j}$ such that $k_{i_{j}}=\infty .$ Then,
there exist $\left( i_{j}\right) _{j}\in I_{0}\times I_{1}\times \ldots $
such that $\emph{liminf}$$\left( k_{i_{j}}\right) _{j\in
\mathbb{N}
}=\infty $ which implies that $B=\infty .$ In addition, we get that $%
\underset{i_{j}\in I_{j}}{\sup }k_{i_{j}}=\infty $ for all $j\geq 0,$ i.e., $%
A=\infty $ as well.

(II) Assume that there exists finitely many $j\geq 0,$ such that $%
k_{i_{j}}\neq \infty $ for all $i_{j}\in I_{j},$ then $\underset{i_{j}\in
I_{j}}{\sup }k_{i_{j}}\neq \infty $ for only a finite number of $j\geq 0,$
which implies that $A=\inf \left\{ \underset{i_{j}\in I_{j}}{\sup }%
k_{i_{j}}\mid j\geq 0\text{ with}\underset{i_{j}\in I_{j}}{\sup }%
k_{i_{j}}\neq \infty \right\} $. Moreover, for all but a finite number of $j\geq$ it holds $k_{i_{j}}\in L\backslash \left\{ \mathbf{%
\infty },\mathbf{-\infty }\right\} $ for all $i_{j}\in I_{j},$ or $%
k_{i_{j}}\in \left\{ \mathbf{\infty },\mathbf{-\infty }\right\} $ for all $%
i_{j}\in I_{j}$, thus there exists finitely many $j\geq0$ such that $k_{i_{j}}\in L\backslash \left\{ \mathbf{%
\infty },\mathbf{-\infty }\right\} $ for all $i_{j}\in I_{j},$ which implies that the following equalities are true.
\begin{align*}
 B=\underset{k_{i_{j}}\neq
-\infty ,j\geq 0}{\underset{\left( i_{j}\right) _{j}\in I_{0}\times
I_{1}\times \ldots }{\sup }}\left( \emph{liminf}\left( k_{i_{j}}\right)
_{j\in
\mathbb{N}
}\right) \\=\underset{k_{i_{j}}\neq -\infty ,j\geq 0}{\underset{\left(
i_{j}\right) _{j}\in I_{0}\times I_{1}\times \ldots }{\sup }}\left( \inf
\left\{ k_{i_{j}}\mid j\geq 0\text{ with }k_{i_{j}}\neq \infty \text{ }%
\right\} \right)\\ =\inf \left\{ \underset{i_{j}\in I_{j}}{\sup }k_{i_{j}}\mid
j\geq 0\text{ with}\underset{i_{j}\in I_{j}}{\sup }k_{i_{j}}\neq \infty
\right\} =A.
\end{align*}

(III) Assume that there exists infinitely many $j\geq 0$ such that $%
k_{i_{j}}\neq \infty $ for all $i_{j}\in I_{j}.$ Since $I_{k}$ are finite
for every $h\geq 0,$ it holds $\underset{i_{h}\in I_{h}}{\sup }
k_{i_{h}} \in \left\{ k_{i_{h}}\mid i_{h}\in I_{h}\right\} $ for
every $h\geq 0,$ hence there exist a sequence $\left( i_{j}\right) _{j}\in
I_{0}\times I_{1}\times I_{2}\times \ldots $ such that $A=$\emph{liminf}$\left(
k_{i_{j}}\right) _{j\geq 0},$ and thus $A\leq B.$

Let now $\left( i_{j}\right) _{j}\in I_{0}\times I_{1}\times I_{2}\times
\ldots $, and $h\geq 0$ be the maximum $j\geq 0$ such that $\left\{
k_{i_{h}}\mid i_{h}\in I_{h}\right\} \cap L\backslash \left\{ \infty
,-\infty \right\} \neq \emptyset ,$ $\left\{ k_{i_{h}}\mid i_{h}\in
I_{h}\right\} \cap \left\{ \infty ,-\infty \right\} \neq \emptyset $, then%
\begin{eqnarray*}
\emph{liminf}\left( k_{i_{j}}\right) _{j\in
\mathbb{N}
} &=&\underset{j\geq 0}{\text{lim}}\left( \inf \left\{ k_{i_{l}}\mid
l\geq j,k_{i_{l}}\neq \infty \right\} \right) \\
&=&\underset{j>h}{\text{lim}}\left( \inf \left\{ k_{i_{l}}\mid l\geq
j,k_{i_{l}}\neq \infty \right\} \right) \\
&\leq &\underset{j>h}{\text{lim}}\left( \inf \left\{ \underset{i_{l}\in I_{l}%
}{\sup }k_{i_{l}}\mid l\geq j,\underset{i_{l}\in I_{l}}{\sup }%
k_{i_{l}}\neq \infty \right\} \right) \\
&=&\underset{j\geq 0}{\text{lim}}\left( \inf \left\{ \underset{i_{l}\in I_{l}%
}{\sup }k_{i_{l}}\mid l\geq j,\underset{i_{l}\in I_{l}}{\sup }%
k_{i_{l}}\neq \infty \right\} \right) \\
&=&\emph{liminf}\left( \underset{i_{j}\in I_{j}}{\sup }k_{i_{j}}\right)
_{j\in
\mathbb{N}
}=A
\end{eqnarray*}
where the second and third equality hold because the sequences
\begin{equation*}
\left( \inf \left\{ k_{i_{l}}\mid l\geq j,k_{i_{l}}\neq \infty \right\}
\right) _{j\geq 0},\left( \inf \left\{ \underset{i_{l}\in I_{l}}{\sup }%
k_{i_{l}}\mid l\geq h,\underset{i_{l}\in I_{l}}{\sup }k_{i_{l}}\neq
\infty \right\} \right) _{j\geq 0}
\end{equation*}
are increasing, and the inequality holds by the fact that
\[
\inf \left\{
k_{i_{l}}\mid l\geq j,k_{i_{l}}\neq \infty \right\} \leq \inf \left\{
\underset{i_{l}\in I_{l}}{\sup }k_{i_{l}}\mid l\geq j,\underset{i_{l}\in
I_{l}}{\sup }k_{i_{l}}\neq \infty \right\} ,
\]
for every $j>h.$ Thus, $A=B$
as wanted.

We prove now Property 2. We will prove that for all $k_{i}\left( i\geq
1\right) \in \overline{%
\mathbb{R}
},$ \emph{liminf}$\left( \infty \mathbf{,}k_{1},k_{2},\ldots \right) =$\emph{liminf}$%
\left( k_{1},k_{2},\ldots \right) .$ First we assume that $\exists i\geq 1,$
such that $k_{i}=-\infty ,$ then  \emph{liminf}$\left( \infty \mathbf{,}%
k_{1},k_{2},\ldots \right) =$\emph{liminf}$\left( k_{1},k_{2},\ldots \right)
=-\infty .$ Otherwise we point out the following cases.

- If $\forall i\geq 1,k_{i}=\infty ,$ then \emph{liminf}$\left( \infty \mathbf{,}%
k_{1},k_{2},\ldots \right) =$\emph{liminf}$\left( k_{1},k_{2},\ldots \right)
=\infty .$

- If there exist infinitely many $i\geq 1$ such that $k_{i}\neq \infty ,$
then\\
\emph{ liminf}$\left( \infty \mathbf{,}k_{1},k_{2},\ldots \right) =\underset{%
j\geq 1}{\text{lim}}\left( \inf \left\{ k_{i_{j}}\mid i_{j}\geq
j,k_{i_{j}}\neq \infty \right\} \right) =$\emph{liminf}$\left( k_{1},k_{2},\ldots
\right) $.

- Finally, if there exist finitely many $i\geq 1$ such that $k_{i}\neq
\infty ,$ then \emph{liminf}$\left( \infty \mathbf{,}k_{1},k_{2},\ldots \right)
= \inf \left\{ k_{i_{j}}\mid j\geq 1,k_{i_{j}}\neq \infty \right\}
 =$\emph{liminf}$\left( k_{1},k_{2},\ldots \right) .$

\bigskip

We prove now Property 3. Let $k\in \overline{%
\mathbb{R}
},$ then we point out the cases $k=\infty ,k=-\infty $ or $k\neq \infty
,-\infty ,$ and we conclude by the definition that $k=$\emph{liminf}$\left(
k,\infty ,\infty ,\ldots \right) $ in all three cases, as wanted.

The rest of the properties of generalized product $\omega $-valuation monoids are
concluded in a straightforward way by the definition of the \emph{liminf}-function,
and the operations of $\sup ,\inf .$
\end{proof}
\\\\
\textbf{Definition of Procedure 3}
Let $\left\{\lambda_{1},\ldots,\lambda_{k}\right\}$ be the subset of
$cl\left( \psi \right) \cap cl\left( \xi \right) $ containing all formulas of the form $\chi U\varrho $, and let $\varphi_{1}U\varphi_{2}=\lambda_{1}$.
We let
$\widetilde{P}_{1}$$=B_{\widetilde{\xi }^{0}}B_{\widetilde{\xi }%
_{re}^{0}}B_{\widetilde{\xi }^{1}}B_{\widetilde{\xi }_{re}^{1}}\ldots $ where $%
B_{\widetilde{\xi }^{i}},B_{\widetilde{\xi }_{re}^{i}}$, $i\geq 0,$ are
defined in the following way. For all $i<i_{1}-1$ we set $\widetilde{\xi}^{i}=\xi^{i},$ $B_{\widetilde{\xi
}^{i}}=B_{\xi^{i}},$ and $B_{\widetilde{\xi}_{re}^{i}}=B_{\xi_{re}^{i}}.$ We
set $\widetilde{\xi}^{i_{1}-1}=\xi^{i_{1}-1},$ and $B_{\widetilde{\xi}%
^{i_{1}-1}}=B_{\xi^{i_{1}-1}},$ and $B_{\widetilde{\xi}_{re}^{i_{1}-1}}%
=B_{\xi_{re}^{i_{1}-1}}\cup M_{B_{\psi^{i_{1}-1},\varphi_{1}U\varphi_{2}}}.$ We have that $\psi^{i_{1}}=\beta^{i_{1}}\wedge\left(
\beta^{\prime}\right)  ^{i_{1}}$ with $\left(\beta^{\prime}\right)^{i_1}\in next\left(
M_{B_{\psi^{i_{1}-1},\varphi_{1}U\varphi_{2}}}\right)  .$ $P_{w}^{1}$
satisfies the acceptance condition for $\varphi_{1}U\varphi_{2}$ at position $i_1$, i.e.,
$\left(\beta^{\prime}\right)^{i_1}\in next\left(  M_{B_{\psi^{i_{1}-1},\varphi_{2}}}\right)  .$ Then,
$\psi_{re}^{i_{1}}=\left(  \beta_{re}^{i_{1}}\wedge\left(  \beta^{\prime
}\right)  _{re}^{i_{1}}\right)  _{re}$\ , and let\textbf{ sequence 1}%
\[
B_{\beta_{re}^{i_{1}}}\overset{\pi_{i_{1}}}{\rightarrow}B_{\beta^{i_{1}+1}%
}\overset{\varepsilon}{\rightarrow}B_{\beta_{re}^{i_{1}+1}}\overset{\pi
_{i_{1}+1}}{\rightarrow}B_{\beta^{i_{1}+2}}\overset{\varepsilon}{\rightarrow
}B_{\beta_{re}^{i_{1}+2}}\rightarrow\ldots,
\]
and \textbf{sequence 2}
\[
B_{\left(  \beta^{^{\prime}}\right)  _{re}^{i_{1}}}\overset{\pi_{i_{1}}%
}{\rightarrow}B_{\left(  \beta^{^{\prime}}\right)  ^{i_{1}+1}}\overset
{\varepsilon}{\rightarrow}B_{\left(  \beta^{^{\prime}}\right)  _{re}^{i_{1}%
+1}}\overset{\pi_{i_{1}+1}}{\rightarrow}B_{\left(  \beta^{^{\prime}}\right)
^{i_{1}+2}}\overset{\varepsilon}{\rightarrow}B_{\left(  \beta^{^{\prime}}\right)
_{re}^{i_{1}+2}}\rightarrow\ldots
\]
be the sequences obtained by Procedure 1. Also, $\xi^{i_{1}}=\zeta^{i_{1}%
}\wedge\left(  \zeta^{\prime}\right)  ^{i_{1}}$, where $\left(  \zeta^{\prime}\right)  ^{i_{1}}$ is the
conjunction of elements of $next\left(  M_{B_{\xi^{i_{1}-1},\varphi
_{1}U\varphi_{2}}}\right)  \backslash next\left(  M_{B_{\xi^{i_{1}-1}%
,\varphi_{2}}}\right)  $ that appear in $\xi^{i_{1}}.$ Then,
$\xi_{re}^{i_{1}}=\left(  \zeta_{re}^{i_{1}}\wedge\left(  \zeta^{\prime}\right)
_{re}^{i_{1}}\right)  _{re},$ and let \textbf{sequence 3}%
\[
B_{\zeta_{re}^{i_{1}}}\overset{\pi_{i_{1}}}{\rightarrow}B_{\zeta^{i_{1}+1}%
}\overset{\varepsilon}{\rightarrow}B_{\zeta_{re}^{i_{1}+1}}\overset{\pi
_{i_{1}+1}}{\rightarrow}B_{\zeta^{i_{1}+2}}\overset{\varepsilon}{\rightarrow
}B_{\zeta_{re}^{i_{1}+2}}\rightarrow\ldots,
\]
and \textbf{sequence 4}
\[
B_{\left(  \zeta^{^{\prime}}\right)  _{re}^{i_{1}}}\overset{\pi_{i_{1}}%
}{\rightarrow}B_{\left(  \zeta^{^{\prime}}\right)  ^{i_{1}+1}}\overset
{\varepsilon}{\rightarrow}B_{\left(  \zeta^{^{\prime}}\right)  _{re}^{i_{1}%
+1}}\overset{\pi_{i_{1}+1}}{\rightarrow}B_{\left(  \zeta^{^{\prime}}\right)
^{i_{1}+2}}\overset{\varepsilon}{\rightarrow}B_{\left(  \zeta^{^{\prime}%
}\right)  _{re}^{i_{1}+2}}\rightarrow\ldots
\]
be the sequences obtained by Procedure 1. Then, we obtain \ $\widetilde{\xi}%
^{i},B_{\widetilde{\xi}_{re}^{i}}$ for $i\in\left\{  i_{1}+1,\ldots
,i_{2}-1\right\}  $ following procedure 2 for sequences 2, and 3$.$ We apply
inductively the construction for all $i_{j},j\geq2.$ Observe that since
$\varphi_{1}U\varphi_{2}$ does appear in the scope of a next operator, then
whenever $\varphi_{1}U\varphi_{2}$ appears as part of the conjunction of form A
of the maximal formula of a state in $P_{w}^{1}$, $P_{w}^{2}$, then it is
obtained from a next formula of the non-empty $\varphi_{1}U\varphi_{2}%
$-consistent subset of the previous state.

For every $2\leq m\leq k$ we obtain $\widetilde{P}_{m}$ by applying the previous procedure for $\widetilde{P}_{m-1}$ and $P_{w}^{1}.$ Then, we set $\widehat{P}_{w}^{2}=\widetilde{P}_{k}$.
  \\\\\\

We present the proof of Lemma \ref{boithitiko until 1}.\\

\begin{proof}
Let $\psi=\underset{1\leq l\leq m}{%
{\displaystyle\bigvee}
}\left(  k_{l}\wedge\psi_{l}\right)  $ where $k_{l}\in K,$ and $\psi_{l}\in
bLTL\left(  K,AP\right)  .$ First, we assume that all $\xi_{j}\left(  1\leq
j\leq n\right)  $ are different from $true$ and $\psi\notin bLTL\left(
K,AP\right)  ,$ and we point out the following cases.

(a) $\xi_{j}=\underset{1\leq i_{j}\leq m_{j}}{%
{\displaystyle\bigwedge}
}\xi^{\left(  j,i_{j}\right)  }$ with $\xi^{\left(  j,i_{j}\right)  }\in
bLTL\left(  K,AP\right)  $ for every $i_{j}\in\left\{  1,\ldots,m_{j}\right\}
,j\in\left\{  1,\ldots,k\right\}  ,$ and for every $j\in\left\{
2,\ldots,k\right\}  $ there exist $i_{j}^{1},\ldots,i_{j}^{h_{j}}\in\left\{
1,\ldots,m_{j}\right\}  $, such that for every $k\in\left\{1,\ldots,h_{j}\right\}$, $\xi^{^{\left(  j,i_{j}^{k}\right)  }%
}=\xi^{\left(  j^{\prime},i_{j^{\prime}}\right)  }$ for some $j^{\prime}%
\in\left\{  1,\ldots,j-1\right\}  ,i_{j^{\prime}}\in\left\{  1,\ldots
,m_{j^{\prime}}\right\}  .$ Then, $\xi_{j}^{\prime}=\underset{1\leq i_{j}\leq
m_{j}}{%
{\displaystyle\bigwedge}
}\left(  \xi^{\left(  j,i_{j}\right)  }\right)  ^{\prime},$ where $\left(
\xi^{\left(  j,i_{j}\right)  }\right)  ^{\prime}\in\widehat{next}\left(
M_{B_{\xi_{j}},\xi^{\left(  j,i_{j}\right)  }}\right)  $ for every $i_{j}%
\in\left\{  1,\ldots,m_{j}\right\}  ,$ $\ \ j\in\left\{  1,\ldots,k\right\}
.\ $ Moreover, it holds $\psi\in stLTL\left(  K,AP\right)  $, and thus $\psi^{\prime}\in
bLTL\left(  K,AP\right)  $. Clearly,%
\[
\varphi=\left(  \psi U\xi\right)  \wedge\left(  \left(  \underset{1\leq
i_{1}\leq m_{1}}{%
{\displaystyle\bigwedge}
}\xi^{\left(  j,i_{j}\right)  }\right)  \wedge\left(  \underset{2\leq j\leq
k}{%
{\displaystyle\bigwedge}
}\left(  \underset{i_{j}\neq i_{j}^{1},\ldots,i_{j}\neq i_{j}^{h_{j}}%
}{\underset{1\leq i_{j}\leq m_{j}}{%
{\displaystyle\bigwedge}
}}\xi^{\left(  j,i_{j}\right)  }\right)  \right)  \right)  .
\]
Let $B_{\varphi}=\left\{  \psi U\xi\right\}  \cup B_{\psi}\cup\left(
\underset{1\leq j\leq k}{%
{\displaystyle\bigcup}
}\left(  \underset{1\leq i_{j}\leq m_{j}}{%
{\displaystyle\bigcup}
}M_{B_{\xi_{j}},\xi^{\left(  j,i_{j}\right)  }}\right)  \right)  .$ We can
prove that $B_{\varphi}$ is a $\varphi$-consistent set following the arguments
of proof of Lemma 107 in \cite{Ma-Co}. Moreover, $M_{B_{\xi_{j}},\xi^{\left(
j,i_{j}\right)  }}\subseteq M_{B_{\varphi},\xi^{\left(  j,i_{j}\right)  }}$
$\left(  1\leq j\leq k,1\leq i_{j}\leq m_{j}\right)  $ and $B_{\psi}\subseteq
M_{B_{\varphi},\psi}$. Then, with the same arguments used in Lemma \ref{boithitiko lemma 1}, we get that $\varphi^{\prime}=\left(  \psi
U\xi\right)  \wedge\psi^{\prime}\wedge\left(  \underset{1\leq i_{1}\leq m_{1}%
}{%
{\displaystyle\bigwedge}
}\left(  \xi^{\left(  1,i_{1}\right)  }\right)  ^{\prime}\right)
\wedge\left(  \underset{2\leq j\leq k}{%
{\displaystyle\bigwedge}
}\left(  \underset{i_{j}\neq i_{j}^{1},\ldots,i_{j}\neq i_{j}^{h_{j}}%
}{\underset{1\leq i_{j}\leq m_{j}}{%
{\displaystyle\bigwedge}
}}\left(  \xi^{\left(  j,i_{j}\right)  }\right)  ^{\prime}\right)  \right)
\in\widehat{next}\left(  B_{\xi}\right)  .$ Therefore, $\left(  B_{\varphi
},\pi,B_{\varphi^{\prime}}\right)  $ is a next transition and
\begin{align*}
v_{B_{\varphi}}\left(  \varphi^{\prime}\right)   &  =v_{B_{\psi}}\left(
\psi^{\prime}\right)  \cdot\underset{1\leq i_{1}\leq m_{1}}{%
{\displaystyle\prod}
}v_{M_{B_{\xi_{j},\xi^{\left(  1,i_{1}\right)  }}}}\left(  \left(
\xi^{\left(  1,i_{1}\right)  }\right)  ^{\prime}\right)  \\
&  \cdot\left(  \underset{2\leq j\leq k}{%
{\displaystyle\prod}
}\left(  \underset{i_{j}\neq i_{j}^{1},\ldots,i_{j}\neq i_{j}^{h_{j}}%
}{\underset{1\leq i_{j}\leq m_{j}}{%
{\displaystyle\prod}
}}\left(  v_{M_{B_{\xi_{j},\xi^{\left(  j,i_{j}\right)  }}}}\left(  \left(
\xi^{\left(  j,i_{j}\right)  }\right)  ^{\prime}\right)  \right)  \right)
\right)
\end{align*}
\begin{align*}
&  \geq v_{B_{\psi}}\left(  \psi^{\prime}\right)  \cdot\underset{1\leq
i_{1}\leq m_{1}}{%
{\displaystyle\prod}
}v_{M_{B_{\xi_{j},\xi^{\left(  1,i_{1}\right)  }}}}\left(  \left(
\xi^{\left(  1,i_{1}\right)  }\right)  ^{\prime}\right)  \\
&  \cdot\left(  \underset{2\leq j\leq k}{%
{\displaystyle\prod}
}\left(  \underset{i_{j}\neq i_{j}^{1},\ldots,i_{j}\neq i_{j}^{h_{j}}%
}{\underset{1\leq i_{j}\leq m_{j}}{%
{\displaystyle\prod}
}}\left(  v_{M_{B_{\xi_{j},\xi^{\left(  j,i_{j}\right)  }}}}\left(  \left(
\xi^{\left(  j,i_{j}\right)  }\right)  ^{\prime}\right)  \right)  \right)
\right)  \\
&  \cdot\underset{2\leq j\leq k}{%
{\displaystyle\prod}
}\left(  \underset{i_{j}\in\left\{  i_{j}^{1},\ldots,i_{j}^{h_{j}}\right\}  }{%
{\displaystyle\prod}
}\left(  v_{M_{B_{\xi_{j},\xi^{\left(  j,i_{j}\right)  }}}}\left(  \left(
\xi^{\left(  j,i_{j}\right)  }\right)  ^{\prime}\right)  \right)  \right)  \\
&  =v_{B_{\psi}}\left(  \psi^{\prime}\right)  \cdot\underset{1\leq j\leq k}{%
{\displaystyle\prod}
}v_{B_{\xi_{j}}}\left(  \xi_{j}^{\prime}\right)
\end{align*}
where the inequality is obtained using Lemma \ref{boithitiko lemma 1}, Remark
 \ref{remark_commutative}, and the same arguments that are used in the proof of the corresponding inequality of Lemma \ref{boithitiko lemma 3}.

We have completed the proof of (i). We prove now (ii). The claim of (ii)
trivially holds for $\psi^{\prime\prime}=\psi^{\prime},\xi_{1}^{\prime\prime
}=\xi_{1}^{\prime}.$ For $2\leq j\leq k$ we set $\xi_{j}^{\prime\prime
}=\underset{i_{j}\neq i_{j}^{1},\ldots,i_{j}\neq i_{j}^{h_{j}}%
}{\underset{1\leq i_{j}\leq m_{j}}{%
{\displaystyle\bigwedge}
}}\left(  \xi^{\left(  j,i_{j}\right)  }\right)  ^{\prime}.$ It holds%
\[
\varphi^{\prime}=\left(  \psi U\xi\right)  \wedge\psi^{\prime}\wedge\left(
\underset{1\leq i_{1}\leq m_{1}}{%
{\displaystyle\bigwedge}
}\left(  \xi^{\left(  1,i_{1}\right)  }\right)  ^{\prime}\right)
\wedge\left(  \underset{2\leq j\leq k}{%
{\displaystyle\bigwedge}
}\left(  \underset{i_{j}\neq i_{j}^{1},\ldots,i_{j}\neq i_{j}^{h_{j}}%
}{\underset{1\leq i_{j}\leq m_{j}}{%
{\displaystyle\bigwedge}
}}\left(  \xi^{\left(  j,i_{j}\right)  }\right)  ^{\prime}\right)  \right)  ,
\]
and
\[
\varphi_{re}^{\prime}=\left(  \left(  \psi U\xi\right)  \wedge\psi
_{re}^{^{\prime\prime}}\wedge\left(  \xi_{1}^{\prime\prime}\right)
_{re}\wedge\left(  \underset{2\leq j\leq k}{%
{\displaystyle\bigwedge}
}\left(  \xi_{j}^{\prime\prime}\right)  _{re}\right)  \right)  _{re},
\]
Let now $j\in\left\{  2,\ldots,k\right\}  .$ We consider now the infinite
sequence of next and $\varepsilon$-reduction transitions
\[
B_{\psi^{0}}\overset{\pi_{0}}{\rightarrow}B_{\psi^{1}}\overset{\varepsilon
}{\rightarrow}B_{\psi_{re}^{1}}\overset{\pi_{1}}{\rightarrow}B_{\psi^{2}%
}\overset{\varepsilon}{\rightarrow}B_{\psi_{re}^{2}}\ldots
\]
with $\psi^{0}=\left(  \xi_{j}^{\prime}\right)  _{re}$ and $v_{B_{\psi
_{re}^{i}}}\left(  \psi^{i+1}\right)  \neq\mathbf{0}$ $\left(  i\geq0\right)
.$ Clearly,
\[
\psi^{0}=\left(  \xi_{j}^{\prime}\right)  _{re}=\left(  \left(  \xi
_{j}^{\prime\prime}\right)  _{re}\wedge\left(  \left(  \xi^{\left(
j,i_{j}^{1}\right)  }\right)  ^{\prime}\wedge\ldots\wedge\left(  \xi^{\left(
j,i_{j}^{h_{j}}\right)  }\right)  ^{\prime}\right)  _{re}\right)  _{re}.
\]
Then, for $\lambda^{0}=\left(  \xi_{j}^{\prime\prime}\right)  _{re}$, and
$\zeta^{0}=\left(  \left(  \xi^{\left(  j,i_{j}^{1}\right)  }\right)
^{\prime}\wedge\ldots\wedge\left(  \xi^{\left(  j,i_{j}^{h_{j}}\right)
}\right)  ^{\prime}\right)  _{re}$, by induction on $i$ and Lemma \ref{boithitiko lemma 2}, we
obtain that for every $i\geq0$, there exist a $\lambda_{re}^{i}$-consistent
set $B_{\lambda_{re}^{i}}$, and a $\zeta_{re}^{i}$-consistent set
$B_{\zeta_{re}^{i}}$, and formulas $\lambda^{i+1}\in next\left(
B_{\lambda_{re}^{i}}\right)  ,$ $\zeta^{i+1}\in next\left(  B_{\zeta_{re}^{i}%
}\right)  $ such that%
\[
\psi_{re}^{i+1}=\left(  \lambda_{re}^{i+1}\wedge\zeta_{re}^{i+1}\right)
_{re},
\]
and
\[
v_{B_{\psi_{re}^{i}}}\left(  \psi^{i+1}\right)  =v_{B_{\lambda_{re}^{i}}%
}\left(  \lambda^{i+1}\right)  \cdot v_{B_{\zeta_{re}^{i}}}\left(  \zeta
^{i+1}\right)  .
\]
For every $i\geq0$, $v_{B_{\psi_{re}^{i}}}\left(  \psi^{i+1}\right)
\neq\mathbf{0}$ and $\zeta_{re}^{i}$\ is boolean, hence $v_{B_{\zeta_{re}^{i}%
}}\left(  \zeta^{i+1}\right)  =\mathbf{1}$, i.e., $v_{B_{\psi_{re}^{i}}%
}\left(  \psi^{i+1}\right)  =v_{B_{\lambda_{re}^{i}}}\left(  \lambda
^{i+1}\right)  $ for every $i\geq0$. So, the sequence
\[
B_{\lambda^{0}}\overset{\pi_{0}}{\rightarrow}B_{\lambda^{1}}\overset{\ast
}{\rightarrow}B_{\lambda_{re}^{1}}\overset{\pi_{1}}{\rightarrow}B_{\lambda
^{2}}\overset{\ast}{\rightarrow}B_{\lambda_{re}^{2}}\ldots
\]
satisfies the lemma's claim.

(b) If $\left(  \left(  \psi U\xi\right)  \wedge\left(  \underset{1\leq j\leq
k}{%
{\displaystyle\bigwedge}
}\xi_{j}\right)  \right)  _{re}=\left(  \psi U\xi\right)  \wedge\left(
\underset{1\leq j\leq k}{%
{\displaystyle\bigwedge}
}\xi_{j}\right)  ,$ we set $B_{\varphi}=\left\{  \psi U\xi\right\}  \cup
B_{\psi}\cup\left(  \underset{1\leq j\leq k}{\cup}B_{\xi_{j}}\right)  $, and we
proceed in the same way.

Finally, we use the same arguments to prove our claim in the cases where at
least one of $\xi_{j}\left(  1\leq j\leq k\right)  $ equals to $true.$
\end{proof}

\end{document}